\documentclass[12pt,letterpaper]{article}

\usepackage{linguex} 
\usepackage{times} 
\usepackage{lipsum} 
\usepackage{ragged2e}
\usepackage{amsmath}
\usepackage{amsthm}
\usepackage{amsfonts}
\usepackage[utf8]{inputenc}
\usepackage{graphicx}
\usepackage{algorithm}
\usepackage{algpseudocode}
\usepackage{float}
\usepackage{subfloat}
\usepackage{subfig}
\usepackage{tikz}
\usetikzlibrary{shapes}
\usepackage{bm}
\usepackage{pgfplots}
\usetikzlibrary{ decorations.markings}
\pgfplotsset{compat = 1.3}
\usepackage{cases}
\usepackage{setspace}
\usepackage{empheq,etoolbox}
\usepackage[table]{xcolor}
\usepackage{multirow}
\usepackage{enumitem}
\usepackage{colortbl}
\usepackage{multirow} 

\definecolor{lightBlue}{HTML}{000080}
\definecolor{lightGreen}{HTML}{228B22}
\definecolor{lightRed}{HTML}{F08080} % lightCoral
\definecolor{darkRed}{HTML}{960018}
\definecolor{lightlightBlue}{HTML}{708090} % 
\definecolor{darkBlue}{HTML}{191970} % 
\definecolor{lightGray}{gray}{0.9}
\definecolor{mediumGray}{gray}{0.4}
\definecolor{darkGray}{gray}{0.25}

\RaggedRight

\usepackage{scrextend} 
\deffootnote[.5em]{0em}{1em}{\textsuperscript{\thefootnotemark}\,}

\newcounter{savefootnote} 
\newcounter{symfootnote}
\newcommand{\symfootnote}[1]{%
	\setcounter{savefootnote}{\value{footnote}}%
	\setcounter{footnote}{\value{symfootnote}}%
	\ifnum\value{footnote}>8\setcounter{footnote}{0}\fi%
	\let\oldthefootnote=\thefootnote%
	\renewcommand{\thefootnote}{\fnsymbol{footnote}}%
	\footnote{#1}%
	\let\thefootnote=\oldthefootnote%
	\setcounter{symfootnote}{\value{footnote}}%
	\setcounter{footnote}{\value{savefootnote}}%
}
\newtheorem{theorem}{Theorem}
\newtheorem{lemma}{Lemma}
\newtheorem{property}{Property}
\newtheorem{assumption}{Assumption}

\newtheorem{remark}{Remark}

\usepackage[labelsep=period]{caption}
\numberwithin{equation}{section}

\usepackage[margin=1.0in]{geometry}
\usepackage[compact]{titlesec}
\titleformat{\section}[runin]{\normalfont\bfseries}{\thesection.}{.5em}{}[.]
\titleformat{\subsection}[runin]{\normalfont\scshape}{\thesubsection.}{.5em}{}[.]
\titleformat{\subsubsection}[runin]{\normalfont\scshape}{\thesubsubsection.}{.5em}{}[.]
\usepackage[colorlinks,allcolors={black},urlcolor={blue}]{hyperref} 

\usepackage[
backend=bibtex,
style=ieee,
citestyle=numeric,
sorting=nty,
maxbibnames=99,
mincitenames=1,
maxcitenames=2
]{biblatex}
\renewenvironment{abstract}{%
	\noindent\begin{minipage}{1\textwidth}
		\setlength{\leftskip}{0.4in}
		\setlength{\rightskip}{0.4in}
		\textbf{Abstract.}}
	{\end{minipage}}

\newenvironment{keywords}{%
	\vspace{.5em}
	\noindent\begin{minipage}{1\textwidth}
		\setlength{\leftskip}{0.4in}
		\setlength{\rightskip}{0.4in}
		\textbf{Keywords.}}
	{\end{minipage}}

\begin{document}	
	\setlength{\Extopsep}{6pt}
	\setlength{\Exlabelsep}{9pt} 
	
	\begin{center}
		\normalfont\bfseries
		A one-dimensional reduced plasma model for the electrical treeing
        % no charge concentration
        % buttare dentro un model reduction
		\vskip .5em
		\normalfont
		{Beatrice Crippa$^\star$, Anna Scotti$^\star$, Andrea Villa$^\dagger$} \\
		\vskip .5em
		\footnotesize{$^\star$ MOX-Laboratory for Modeling and Scientific Computing, Department of Mathematics, Politecnico di Milano, 20133 Milan, Italy \\
			$^\dagger$ Ricerca Sul Sistema Energetico (RSE), 20134 Milano, Italy}
	\end{center}
	
	\justifying
	
	\begin{abstract}
        Plasma models, consisting of advection--diffusion Partial Differential Equations coupled with chemical reactions, are widely adopted to describe corona, streamers and dielectric barrier discharges. However,
        the complex geometry of the electrical treeing represents an obstacle for numerical simulations.
		We develop a reduced one-dimensional formulation of a plasma model for the electrical treeing, describing the evolution of charge concentrations under the effect of an electric field. The reduced system consists of weakly coupled advection--diffusion--reaction equations for charge concentrations inside the treeing and on the dielectric surface, coupled with production--destruction Ordinary Differential Equations for the dipole moment. 
        A numerical scheme based on Finite Volumes and Patankar--type methods allows efficient simulations, while preserving key physical properties. The model is tested on increasingly complex geometries, from a straight line to a realistic electrical treeing.
	\end{abstract}
	
	\begin{keywords}
		Reduced model; drift-diffusion; electrical treeing; plasma; partial discharges.
	\end{keywords}

    \section{Introduction}
	    % Focus of the article
    %%% - 1d reduction of drit-diffusion of charges in thin defenct in insulators, under the effect of intense electric fields
    This work is focused on the derivation of a one--dimensional (1D) model for the movement of different charged species (electrons, positive and negative ions) in the electrical treeing, also taking into account exchange of charges with the dielectric surface.\\
    The electrical treeing is one of the main causes of deterioration of solid insulators~\cite{bahadoorsingh2007role, buccella2023computational}, and it consists of a self--propagating thin and elongated ramified fracture~\cite{schurch2014imaging,schurch20193d}. Its development is caused by a prolonged action of intense electric fields on the dielectric components of the power system.
    The complex geometry of the treeing limits Partial Differential Equation (PDE)--based plasma simulations to the first stages of partial discharges, or to very simplified structures~\cite{callender2019plasma, villa2022towards}. Thus, the first step towards complete simulations of this phenomenon is the reduction of its complexity. Since the branches of the defect are typically very thin and elongated, they can be approximated by 1D lines, and the whole treeing structure can be reduced to a 1D graph by collapsing it to its skeleton. This geometrical reduction strategy is widely applied in many different fields to describe thin ramified inclusions in three--dimensional (3D) bulks. Applications include fluid flow in fractured porous media~\cite{lesinigo2011multiscale}, root--soil systems~\cite{koch2022projection} and blood flow in vascular networks~\cite{koppl20203d}.
    
    In this work we perform a rigorous geometrical reduction of the system of equations describing the movement and interaction between different charged species in the gas volume and on the dielectric surface.
    % Literature on the application
    %%% - We consider continuous fluid-dynamics approximation of charge interactions --> cite georghiou2005numerical
    We rely on the hydrodynamic approximation of partial discharges~\cite{georghiou2005numerical, rauf1999dynamics}, combining the continuity equations for charged species~\cite{morrow1985theory} and a Poisson equation for the electric field in the gas, leading to a formulation similar to the one presented in~\cite{VILLA2017687} for the 3D case.
    %%% - All the possible chemical reactions (see long article with tables)
    %%% - Simplifications of chemical reactions: only ionization and attachment
    While more complex chemical models can be introduced in the gas~\cite{komuro2013behaviour}, in this work we only take into account attachment and ionization reactions, producing negative and positive ions due to electron collisions. Similar chemical models, also involving recombination effects, in the context of partial discharges in air, are examined in~\cite{morrow1997streamer, tran2010numerical}.
    
    We introduce a numerical scheme for the solution to such coupled problem reduced to 1D, based on an implicit fractional--step method~\cite{villa2017implicit}, with operator splitting~\cite{glowinski2017splitting, lie1893theorie} to separate chemistry from transport, yielding an Ordinary Differential Equation (ODE) system and a drift--diffusion PDE system.
    Charge conservation and positivity of the solution are preserved discretizing the ODEs with the Patankar--Euler method~\cite{patankar2018numerical} and the PDEs with the Finite Volume method with upwind fluxes and Two--Point Flux Approximation (TPFA), as discussed in~\cite{crippa2024numericalmethods}.

    In thin elongated structures, such as the typical electrical treeing geometry, the longitudinal and transverse components of the electric field can be decoupled. The longitudinal electric field can be computed on the basis of the mixed--dimensional 3D--1D model presented in~\cite{crippa2024mixed}, describing the electrostatic potential in the defect, coupled with the electric field and potential in the surrounding insulator. In this work we address instead an efficient evaluation of the transverse electric field, while assuming for simplicity the longitudinal electric field in the gas and the electric field in the solid to be known quantities. A complete framework coupling the two problems will be object of future work. As discussed in~\cite{georghiou2005numerical, morrow1985theory, morrow1997streamer}, an accurate evaluation of the electric field transverse with respect to the centerline of the gas domain is essential for a precise simulation of the phenomenon. Indeed, it influences the advective flux that moves the charged particles towards the dielectric surface, causing an exchange of charge between volume and interface.
    Thus, we also introduce an efficient evaluation of the transverse components of the electric field on cross--sections of the gas domain. Instead of more expensive Finite Element solutions of the two--dimensional (2D) Poisson problem on cross--sections of the 3D gas domain, we use superimposition of effects, considering contributions from internal charge, surface charge, dipole moments of the surface charge distribution and electric field in the dielectric, which can be analytically obtained on simplified geometries, and then extended to more complex ones. This approach relies on a multi--scale approximation that allows to separate the longitudinal and transverse components of the field in the gas.
    
    We first test the model on a simple geometry, analyzing the convergence of the numerical method with respect to the time discretization. Then, we study the behavior of the solution on a simple branched domain, and finally we simulate the first stages of electron avalanche~\cite{meek1954electrical} on a realistic geometry of the electrical treeing. The coefficients for the diffusion--transport--reaction equations depend on the electric field and can be approximated through the solution to the stationary Boltzmann equation, and many estimates are provided in the literature~\cite{kang2003numerical, morrow1997streamer, nikonov2001surface, steinle1999two}. In all the tests presented in this work, we fix the value of the electric field tangential to the treeing centerline and compute the coefficients with the estimates provided by Morrow and Lowke~\cite{morrow1997streamer}.

    % Structure of the paper
    This manuscript is organized as follows: in Section~\ref{section:reduction1d:3dProblem} we present the full 3D problem, while in Section~\ref{section:reduction1d:modelReduction} we introduce some useful hypotheses for its reduction to 1D; in Section~\ref{section:reduction1d:electricField} we discuss the approximation of the transverse components of the electric field via superimposition of effects and in Section~\ref{section:reduction1d:properties_1D_reduced_problem} we analyze the conservativity and monotonicity of the equations. The numerical schemes used for the simulations are detailed in Section~\ref{section:reduction1d:numerical_methods}, and some numerical examples on increasingly complex geometries are shown in Section~\ref{section:reduction1d:results}. 

	Details of some standard but long computations, needed to prove some results presented in the paper, are reported in Appendices~\ref{appendix:exact_sol_superimposition_of_effects},~\ref{appendix:FG} and~\ref{appendix:reduction1d:characteristic_time}.

    \section{The 3D model problem}
    \label{section:reduction1d:3dProblem}
    Consider a three-dimensional domain $\Omega = \Omega_s \cup \Omega_g$, given by the union of a gas domain $\Omega_g$ and a solid dielectric one $\Omega_s$, such that $\Omega_g \cap \Omega_s = \emptyset$. The two domains, made of materials with different electrical conductivity properties, summarized by the the dielectric constants $\epsilon_g$ and $\epsilon_s$, are separated by a surface $\Gamma = \bar{\Omega}_g\cap\bar{\Omega}_s$, as displayed in Figure~\ref{figure:3d-3d:domain} for a simplified setting where the two domains are represented by coaxial cylinders. To model the evolution of charged particles in the gas domain $\Omega_g$ we rely on the problem discussed in~\cite{VILLA2017687}. We consider three families of charged species, whose volume concentrations in $\Omega_g$ are denoted by $c_p,\ p=1,\,2,\,3,$, corresponding to electrons, negative ions and positive ions, respectively, with unit-charges given by $\omega_p=\{-1,\,-1,\,1\}$. The movement of charged particles in $\Omega_g$ is governed, on a macroscopic scale, by advection-diffusion-reaction equations, with diffusion coefficients $\nu_p=\{\nu_1,\,0,\,0\}$, and transport velocity dependent on the electronic mobilities $\mu_p$ and on the electric field $\mathbf{E}_g$ present in $\Omega_g$. We consider a very simple model of chemical reactions among the considered particles, given by the photo-ionization $S_p^\text{ph}$, and the ionization and attachment effects, modeled by the terms $C_p,\, p=1,\,2,\,3$. In particular, we consider a simplified version of the model discussed by~\textcite{morrow1997streamer}:

    \begin{equation}
        C_p =
        \begin{cases}
            \left( \alpha - \eta \right) \mu_1 |\mathbf{E}| c_1, & \text{if\ } p=1,\\
            \alpha \mu_1 |\mathbf{E}| c_1, & \text{if\ } p=2,\\
            \eta \mu_1 |\mathbf{E}| c_1, & \text{if\ } p=3,
        \end{cases}
        \label{eq:chemistry:expression}
    \end{equation}

    \noindent where $\alpha(|\mathbf{E}|)$ and $\eta(|\mathbf{E}|)$ denote the ionization and attachment coefficients, respectively. More complex chemical schemes can be introduced~\cite{komuro2013behaviour}, also describing the interactions among the different species in every family $p$. We also consider transport equations for the surface concentrations of positive (holes) and negative charges, denoted by $c_{\Gamma,h}$ and $c_{\Gamma,e}$, respectively, on $\Gamma$, with transport velocities dependent on the electric field $\mathbf{E}_\Gamma$ on $\Gamma$ and the mobility coefficients $\mu_{\Gamma,e}$ and $\mu_{\Gamma,h}$.\\
    We report the complete 3D system presented in~\cite{VILLA2017687} for clarity:
    % C_1 = (alpha - eta) me_e |E|
    % C_2 = ??
	
	\begin{subnumcases}{\label{eq:3d-3d}}
		\dfrac{\partial c_p}{\partial t} + \nabla\cdot(\omega_p \mu_p c_p \mathbf{E}_g) - \nabla\cdot(\nu_p \nabla c_p) = C_p + S_p^{\text{\text{ph}}}, \quad p=1,\,2,\,3, &  $ \text{in} \ \Omega_g, $
		\label{eq:3d-3d:1} \\
		\dfrac{\partial c_{\Gamma,e}}{\partial t}-\nabla_{\Gamma}\cdot(c_{\Gamma,e}\mu_{\Gamma,e}\mathbf{E}_{\Gamma}) =
        \sum_{p=1}^2 c_p \mu_p \mathbf{E}_g\cdot\mathbf{n}_g -\sum_{p=1}^2\nu_p\nabla c_p \cdot \mathbf{n}_g, \qquad & $ \text{on}\ \Gamma, $
		\label{eq:3d-3d:2} \\
		\dfrac{\partial c_{\Gamma,h}}{\partial t}+\nabla_\Gamma\cdot(c_{\Gamma,h}\mu_{\Gamma,h}\mathbf{E}_\Gamma) = c_3 \mu_3\mathbf{E}_g\cdot\text{n}_g, \qquad & $ \text{on} \ \Gamma $,\\
		\label{eq:3d-3d:3} 
		c_p=c_p^b \quad \forall p=1,\,2,\,3,\,4, \qquad & $ \text{on}\ \Gamma $,
		\label{eq:3d-3d:4}
	\end{subnumcases}

    \noindent with homogeneous Dirichlet boundary conditions for every family, except for electrons, that take into account Schottky~\cite{schottky1923kalte} and secondary emission~\cite{hagstrum1954theory}:

    \begin{equation}
         c_p^b = 
         \begin{cases}
            \gamma\dfrac{\mu_3}{\mu_1}c_3 + \dfrac{AgT^2}{e\mu_1|\mathbf{E}_g|}\exp\left(-\dfrac{w -\Delta w(|\mathbf{E}_g|)}{kT}\right), & \text{if\ } p=1,\\
            0, & \text{if\ } p=2,\,3,
        \end{cases}
        \label{eq:3d-3d:bc}
    \end{equation}

    \noindent where $\gamma$ denotes the secondary emission coefficient, $T$ the temperature, $A_g = 1.2017 \cdot 10^6 Am^{-2} K^{-2}$, $k$ is the Boltzmann constant, $w$ the material-specific work functions, $\Delta w = \dfrac{e^3|\mathbf{E}_g|}{4\pi\epsilon_0}$ the electric field correction, and $e\approx 1.6022\cdot 10^{-19} C$ the electron charge.\\
    The electric field $\mathbf{E}$ and electrostatic potential $\Phi$ on the coupled domains can be expressed as follows:
    
	\begin{center}
		\begin{minipage}{.4\textwidth}    
        	\begin{equation}\nonumber
        		\mathbf{E}(t,\mathbf{x})=
        		\begin{cases}
        			\mathbf{E}_s(t,\mathbf{x}) \ &\text{in} \ \Omega_s,\\
        			\mathbf{E}_g(t,\mathbf{x}) \ &\text{in} \ \Omega_g,\\
        			\mathbf{E}_\Gamma(t,\mathbf{x}) \ &\text{on} \ \Gamma,
        		\end{cases}
        	\end{equation}
		\end{minipage}
		\hspace{.5em}
		\begin{minipage}{.4\textwidth}
			\begin{equation}\nonumber
        		\Phi(t,\mathbf{x})=
        		\begin{cases}
        			\Phi(t,\mathbf{x}) \ &\text{in} \ \Omega_s,\\
        			\Phi_g(t,\mathbf{x}) \ &\text{in} \ \Omega_g,\\
        			\Phi_\Gamma(t,\mathbf{x}) \ &\text{on} \ \Gamma.
        		\end{cases}
        	\end{equation}
		\end{minipage}
	\end{center}

	\noindent Moreover, we denote by $g = \dfrac{e}{\epsilon_0}(c_{\Gamma,h}-c_{\Gamma,e})$ and $q=\dfrac{e}{\epsilon_0}\sum_{p=1}^3\omega_p c_p$ the total surface charge on $\Gamma$ and the total volume charge in $\Omega_g$, respectively.\\
	The electric field $\mathbf{E}$ is computed as the solution of the electrostatic problem on coupled domains (see Figure~\ref{figure:3d-3d:domain}):
	
	\begin{equation}
        \begin{cases}
    	    \nabla\cdot(\epsilon_g\mathbf{E}_g) = q, &  \text{on}\ \Omega_g,\\
    		\nabla\cdot(\epsilon_s\mathbf{E}_s) = 0, &  \text{on}\ \Omega_s,\\
    		\mathbf{E} = -\nabla \Phi,  &  \text{on}\ \Omega,\\
    		\Phi_g = \Phi_s, & \text{on}\ \Lambda, \\
    		\epsilon_s\mathbf{E}_s \cdot\mathbf{n}_s + \epsilon_g\mathbf{E}_g \cdot\mathbf{n}_g = -g, & \text{on}\ \Lambda,\\
    		\mathbf{E}_s\cdot\mathbf{n} = E^b, & \text{on}\ \partial\Omega.
        \end{cases}
        \label{eq:electric_field:3d:dim}
	\end{equation}

	\begin{figure}
		\centering
		\begin{tikzpicture}
            \scriptsize
            \node (b) [cylinder, draw, shape border rotate=90, minimum height=2cm, minimum width=.7cm, fill=gray!40, line width = .5,yshift = .2cm] {};
			\node (a) [cylinder, draw, shape border rotate=90, minimum height=2.5cm, minimum width=3.5cm,line width = .5] {};
			\draw [->] (0.35,.6) -- (.7,.6) node[midway,below] {$\mathbf{n}_g$};
			\node [yshift=-.4cm] {$\Omega_g$};
			\node [xshift=1cm, yshift=1cm] {$\Omega_s$};
			\node [yshift=1.5cm] {$\partial\Omega_N$};
			\node [yshift=-1.4cm] {$\partial\Omega_N$};
			\node [xshift=-2.5cm] {$\partial\Omega_D$};
			\node [xshift=-.5cm] {$\Gamma$};
			\node [xshift=-.2cm,yshift=.5cm] {$\Lambda$};
			\draw [->] (1.75,.4) -- (2.1,.4) node[midway,below] {$\mathbf{n}$};
			\draw [->] (0.5,0.3) -- (0,0.3) node[midway,below] {$\mathbf{n}_s$};
			\draw[dashed] (0,-1) -- (0,1.2);
		\end{tikzpicture}
		\caption{Simplified domain $\Omega$, given by two coaxial cylinders, corresponding to the dielectric and gas subdomains, $\Omega_s$ and $\Omega_g$, respectively. We call the portion of axis inside the gas domain $\Lambda$ and the interface between the two subdomains $\Gamma$.}
		\label{figure:3d-3d:domain}
	\end{figure}
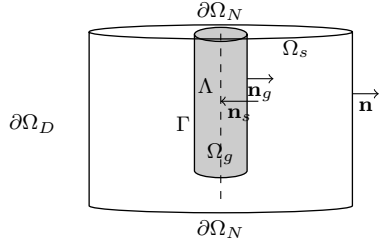
	
	\noindent Mixed Dirichlet and Neumann boundary conditions are set on $\partial\Omega^D$ and $\partial\Omega^N$, that form a partition of the boundary $\partial\Omega$.
    % !!!!!!!!!!!!!!!!!!!!!!!!!!!!
    % dire prima di \omega_p che le popolazioni che ho sono elettroni, ioni negativi e ioni positivi
    % C_1 = (alpha - eta) me_e |E|
    % C_2 = ?? v. Moro, Lock "Streamer propagation in air"
    % dire che rispetto a Moro, Lock il modello è semplificato ma a livello di riduzione non è significativo
	
	Finally, we introduce an additional Robin interface condition on the diffusive flux through $\Gamma$ for the volume charge concentrations, modeling the probability of electrons to attach to $\Gamma$ instead of bouncing back into $\Omega_g$~\cite{wilson2018investigation, boeuf1995two}:
    % !!!!!!!!!!!!!!!!!!!!!!!!!!!!
    % modella il fatto che gli ELETTRONI che collidono con la parete hanno una certa propbabilità di rimanere attaccati, e quindi data una concentrazione di elettroni posso modellare un flusso
    
	\begin{assumption}\label{assumption:robin}
		$\forall p\in\{1,\,2,\,3\}\ \exists K_p\in\mathbb{R}\ \text{such that}\ -\nu_p\nabla c_p\cdot\mathbf{n}_g = K_p c_p\ \text{on}\ \Sigma \quad \forall q\in\{1,\dots,\mathcal{N}_p\}$, with $K_p = 0,\, \forall p\neq 1$.
	\end{assumption}
	
	\noindent The gas domain represents the electrical treeing and consists of a very thin ramification inside $\Omega_s$, approximable by a one-dimensional graph.
    
	\section{Model reduction}
	\label{section:reduction1d:modelReduction}

    In the following, we reduce the problem~\eqref{eq:3d-3d} introduced in Section~\ref{section:reduction1d:3dProblem}, in the 3D gas domain $\Omega_g$, to a 1D problem defined on its skeleton. We start by a simple configuration where the 1D domain consists of a straight line, but the obtained problem can be extended to arbitrary ramified graph domains.

    Assume, for simplicity, that the gas domain consists in a cylinder of radius $R_g$, as in Figure~\ref{figure:3d-3d:domain}, and introduce a parameterization of its centerline $\Lambda$, defining the coordinate $s\in[0,S]$ along it. If one basis of the cylinder $\Omega_g$ belongs to the external boundary $\partial\Omega$ and the other one is internal to the solid domain $\Omega_s$, we will assume that the point of coordinate $s=0$ coincides with the center of the former and $s=S$ with the center of the latter. Then, for all $\mathbf{x}_s\in\Lambda$ of coordinate $s\in[0,S]$, define the transversal section $\mathcal{D}(s)$ of $\Omega_g$, orthogonal to its centerline and centered in $\mathbf{x}_s$, as $\mathcal{D}(s) = \{ \mathbf{x}\in\Omega_g \ :\: |\mathbf{x} - \mathbf{x}_s| < R_g \text{\ and\ } (\mathbf{x} - \mathbf{x}_s) \perp \Lambda \}$. Then, the interface between the two domains is $\Gamma = \mathcal{D}(S)\cup\left(\bigcup_{s\in[0,S]} \partial\mathcal{D}(s)\right)$. Assume moreover that the cylindrical gas domain is thin and elongated:

    \begin{assumption}
    	The gas domain is a cylinder with radius $R_g$ and length $L$, with $R_g\ll L$ and $R_g\ll R_\text{ext}$, where $R_\text{ext}$ denotes the characteristic dimension of the external domain $\Omega_s$.
    	\label{assumption:coupled3d1d:thin}
    \end{assumption}
    
    \noindent Then, we can approximate problem~\eqref{eq:3d-3d} by 1D equations, collapsing $\Omega_g$ onto its centerline $\Lambda$, and identifying, in the electrostatic problem~\eqref{eq:electric_field:3d:dim}, the solid domain $\Omega_s$ with the whole 3D domain $\Omega$.
    
    In the following, we perform the reduction to 1D of equation~\eqref{eq:3d-3d}, concerning the concentration of charged particles, following a similar approach as Laurino and Zunino~\cite{laurino2019derivation} and Masri et al.~\cite{masri2024modelling}, while a mixed--dimensional 3D--1D formulation of the electrostatic problem was derived in~\cite{crippa2024mixed}. This mixed--dimensional electrostatic problem allows a direct computation of the external electric field in the 3D bulk and the electrostatic potential on the 1D gas domain, and as a byproduct the component of the electric field in the gas tangential to the 1D domain can be determined as a function of the gradient of the potential.
	
	\subsection{Volume charge concentration}
	\label{section:reduction1d:volume_charge_concentration}

    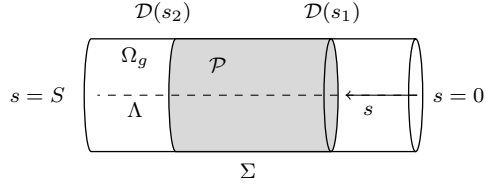
\begin{figure}
        \centering
        \begin{tikzpicture}\scriptsize
            \node (a) [cylinder, draw, minimum height=4.5cm, minimum width=1.5cm, line width = .5] {};
    		\node (b) [cylinder, draw, minimum height=2.25cm, minimum width=1.5cm, fill=gray!30, line width = .5] {};
            \draw [dashed] (2.25,0) -- (-2,0);
    		\draw [ ->] (2.25,0) -- (1.3,0);
    		\node [ yshift=-.2cm, xshift=1.6cm] {$s$};
    		\node [ yshift=-.2cm, xshift=-1.5cm] {$\Lambda$};
    		\node [ yshift=.5cm, xshift=-1.5cm] {$\Omega_g$};
    		\node [yshift=-1cm] {$\Sigma$};
    		\node [yshift=.4cm,xshift=-.4cm] {$\mathcal{P}$};
    		\node [yshift = 1.1cm, xshift = 1.125cm] {$\mathcal{D}(s_1)$};
    		\node [yshift = 1.1cm, xshift = -1.125cm] {$\mathcal{D}(s_2)$};
    		\node [xshift=2.8cm] {$s=0$};
    		\node [xshift=-2.8cm] {$s=S$};
        \end{tikzpicture}
    	\caption{Cylindrical portion $\mathcal{P}$, with lateral surface $\Sigma$, of the gas domain $\Omega_g$, on which we integrate the equation in the gas.}
    	\label{figure:coupled3d1d:reduction:domain}
    \end{figure}
    
	Let $\mathcal{P}$ be the portion of $\Omega_g$ enclosed between the sections $\mathcal{D}(s_1)$ and $\mathcal{D}(s_2)$, $s_1, \ s_2\in\Lambda$, with boundary $\partial\mathcal{P}$ and lateral surface $\Sigma = \partial\mathcal{P} \cap \Gamma$, as in Figure~\ref{figure:coupled3d1d:reduction:domain}. From now on we will assume that the volume concentration of charge is approximately constant over the sections $\mathcal{D}$, as well as the mobility and diffusion coefficients, and the chemical reaction terms.	Moreover, in thin long domains, the effect of photo--ionization is negligible with respect to the emission of electrons from the boundary. We consider the concentrations of the same three families of charged particles $c_p,\, p=1,\,2,\,3$ in the gas volume, as in the 3D case.
    
	\begin{assumption}\label{assumption:reduction1d:cpq_constant}
		$c_p\simeq\dfrac{1}{|\mathcal{D}(s)|}\bar{c}_p(s),\, \forall p\in\{1,\,2,\,3\},$ where $\bar{c}_p(s)=\displaystyle\int_{\mathcal{D}(s)}c_p$, indicates the integral mean concentration of the $p^{th}$ family of charged species over a section, $\forall s\in\Lambda$.
	\end{assumption}
    % !!!!!!!!!!!!!!!!!!!!!!!!!!!!
    % Vantaggio di questa formulazione: \bar{c}_p non cambia quando cambio raggio
    
	\begin{assumption}\label{assumption:reduction1d:mup_nup_constant}
		$\mu_p\simeq \bar{\mu}_p(s)$ and $\nu_p\simeq \bar{\nu}_p(s)$ on $\mathcal{D}(s)\ $ $\forall p\in\{1,\,2,\,3\},\quad\forall s\in\Lambda$, where $\bar{\mu}_p$ and $\bar{\nu}_p$ denote the integral means of $\mu_p$ and $\nu_p$ on cross--sections of the gas domain, $p=1,\,2,\,3$.
	\end{assumption}
	
	\begin{assumption}
		The chemical source terms $C_p$ are constant on sections of $\mathcal{P}$ orthogonal to $\Lambda$, i.e. $C_p\simeq \bar{C}_p(s)\ \text{on}\ \mathcal{D}(s),\ \forall p\in\{1,\,2,\,3\}, \ \forall s\in\Lambda$, where $\bar{C}_p$ denotes the integral mean of $C_p$ on cross--sections of the gas domain, $p=1,\,2,\,3$.
		\label{assumption:reduction1d:constant_c}
	\end{assumption}

%    \noindent Moreover, in thin and long domains, the effect of photo-ionization is not relevant, since it describes the production from other sources of first emission electrons, giving rise to the avalanche effect under intense electric fields. The produced electrons come from the boundary of $\Omega_g$, but this effect is not visible in thin canals and, in particular, when they are reduced to one-dimensional lines.

    \begin{assumption}
        The photo-ionization terms are negligible, i.e. $S_p^\text{ph}\simeq 0,\ \forall p=1,\,2,\,3$.
    \end{assumption}
    
    % !!!!!!!!!!!!!!!!!!!!!!!!!!!!
    % aggiungere remark su perché i coefficienti non sono divisi per D mentre la soluzione lo è.
    % A questo punto E non è diviso per D pervché qui è uno dei coefficienti noti!
    % inoltre E definito in questo modo non è conforme al codice

    \noindent The electric field generated by a constant charge distribution presents non-negligible radial components, thus forbidding to reduce equations~\eqref{eq:electric_field:3d:dim} to a purely tangential problem while being able to retrieve information on the transverse component. The mixed--dimensional model for the electrostatic problem presented in~\cite{crippa2024mixed} provides a method for computing the longitudinal component of the electric field along $\Lambda$ and also the radial components produced by a constant distribution of charge on cross sections. However, this model does not take into account generic transverse components of the electric field, due to the surface charge distribution and external effects. Thus, we will compute the transverse components by superimposition of effects, as detailed in Section~\ref{section:reduction1d:electricField}. On the other hand, we assume the longitudinal component of $\mathbf{E}_g$ along $\Lambda$ to be constant on the cross sections $\mathcal{D}$:
	
	\begin{assumption}\label{assumption:reduction1d:Egs_constant}
		$\mathbf{E}_g \cdot \mathbf{s} \simeq E_\Lambda(s),\, \forall s\in\Lambda$, where $\mathbf{s}$ denotes the unit vector tangential to $\Lambda$, directed from $s=0$ to $s=S$ (for instance, in Figure~\ref{figure:3d-3d:domain}, from the top basis of the inner gray cylinder to the bottom one).
	\end{assumption}

%    \begin{remark}
%        As defined in Assumption~\ref{assumption:reduction1d:cpq_constant}, $\bar{c}_p$ is independent of the sections $\mathcal{D}$ of $\Omega_g$. An advantage of this choice is flexibility in the redefinition of the model on different geometries, and in particular when changing the radius of the gas domain.
%    \end{remark}
    
	\noindent We start the geometrical reduction by integrating equation~\eqref{eq:3d-3d:1} over $\mathcal{P}$	
	% \begin{equation}\nonumber
	% 	\int_{\mathcal{P}}\dfrac{\partial c_p}{\partial t}+\int_{\mathcal{P}}\nabla\cdot(\omega_p\mu_p c_p\mathbf{E}_g)-\int_{\mathcal{P}}\nabla\cdot(\nu_p \nabla c_p) = \int_{\mathcal{P}}C_p.
	% \end{equation}
	%\noindent By 
    and applying the divergence theorem to the second and third terms of the left-hand-side:
	
	\begin{equation}\nonumber
		\int_{\mathcal{P}}\dfrac{\partial c_p}{\partial t}+\int_{\partial\mathcal{P}}\omega_p\mu_p c_p\mathbf{E}_g\cdot\mathbf{n}-\int_{\partial\mathcal{P}}\nu_p \nabla c_p\cdot\mathbf{n} = \int_{\mathcal{P}}C_p,
	\end{equation}
	
	\noindent where $\mathbf{n}$ denotes the outward unit normal vector to $\partial\mathcal{P}$.\\
	Since the closed surface $\partial\mathcal{P}$ is the union of $\Sigma$ and the sections $\mathcal{D}(s_1)$ and $\mathcal{D}(s_2)$, we can split the integrals over $\partial\mathcal{P}$ in the sum of integrals over these three surfaces, whose unit normal vectors are $\mathbf{n}_g$, $-\mathbf{s}$ and $\mathbf{s}$, respectively. Thanks to Assumptions~\ref{assumption:reduction1d:cpq_constant} -~\ref{assumption:reduction1d:Egs_constant}, we can replace $c_p,\ \mu_p,\ \nu_p\ \text{and}\ C_p$ on $\mathcal{D}$ with their constant values on the section, and neglect $S_p^\text{ph}$:
	
	\begin{multline}
		\int_{\mathcal{P}}\dfrac{1}{|\mathcal{D}|}\dfrac{\partial \bar{c}_p}{\partial t} -
        \int_{\mathcal{D}(s_1)}\omega_p\bar{\mu}_p(s_1) \dfrac{\bar{c}_p(s_1)}{|\mathcal{D}(s_1)|}E_\Lambda +
        \int_{\mathcal{D}(s_2)}\omega_p\bar{\mu}_p(s_2) \dfrac{\bar{c}_p(s_2)}{|\mathcal{D}(s_2)|}E_\Lambda
        + \int_{\Sigma}\omega_p\bar{\mu}_p \dfrac{\bar{c}_p}{|\mathcal{D}|}\mathbf{E}_g\cdot\mathbf{n}_g +\\
        + \int_{\mathcal{D}(s_1)}\bar{\nu}_p(s_1)\dfrac{\nabla\bar{c}_p(s_1)}{|\mathcal{D}(s_1)|} \cdot\mathbf{s} -
        \int_{\mathcal{D}(s_2)}\bar{\nu}_p(s_2)\dfrac{\nabla\bar{c}_p(s_2)}{|\mathcal{D}(s_2)|}\cdot\mathbf{s}
        - \int_{\Sigma}\bar{\nu}_p \dfrac{\nabla \bar{c}_p}{|\mathcal{D}|}\cdot\mathbf{n}_g =
        \int_{\mathcal{P}}\bar{C}_p.
        \label{eq:reduction1d:cpq:divergencethm}
	\end{multline}
	
	\noindent According to the fundamental theorem of calculus:

    \begin{subequations}
    	\begin{align}
    		&%\begin{multlined}
    			- \int_{\mathcal{D}(s_1)}\omega_p\bar{\mu}_p(s_1) \dfrac{\bar{c}_p(s_1)}{|\mathcal{D}(s_1)|}E_\Lambda +
            	\int_{\mathcal{D}(s_2)}\omega_p\bar{\mu}_p(s_2) \dfrac{\bar{c}_p(s_2)}{|\mathcal{D}(s_2)|} E_\Lambda %= \\
            	=\int_{s_1}^{s_2}\dfrac{\partial}{\partial s}\left(\int_{\mathcal{D}(s)}\omega_p\bar{\mu}_p(s)\dfrac{\bar{c}_p(s)}{|\mathcal{D}(s)|} E_\Lambda\right);
            %\end{multlined}
            \\
            &\int_{\mathcal{D}(s_1)}\bar{\nu}_p(s_1)\dfrac{\nabla\bar{c}_p(s_1)}{|\mathcal{D}(s_1)|}\cdot\mathbf{s} -
            \int_{\mathcal{D}(s_2)}\bar{\nu}_p(s_2)\dfrac{\nabla\bar{c}_p(s_2)}{|\mathcal{D}(s_2)|}\cdot\mathbf{s} =
            - \int_{s_1}^{s_2}\dfrac{\partial}{\partial s}\left(\int_{\mathcal{D}(s)}\bar{\nu}_p(s) \dfrac{\nabla \bar{c}_p(s)}{|\mathcal{D}(s)|}\cdot\mathbf{s}\right).
    	\end{align}
        \label{eq:reduction1d:cpq:fundamentaltheorem}
    \end{subequations}
	
	\noindent Then, exploiting the cylindrical shape of $\Omega_g$ (Assumption~\ref{assumption:coupled3d1d:thin}) and the Robin condition introduced in Assumption~\ref{assumption:robin}, we can write the diffusive flux of charges through the lateral surface as follows:
	
	\begin{equation}\label{eq:reduction1d:cpq:robin}
		-\int_\Sigma \bar{\nu}_p\dfrac{\nabla \bar{c}_p}{|\mathcal{D}|}\cdot\mathbf{n}_g \simeq
        \int_\Sigma K_p \dfrac{\bar{c}_p}{|\mathcal{D}|} =
        \int_{s_1}^{s_2}\int_{\partial\mathcal{D}(s)}K_p \dfrac{\bar{c}_p}{|\mathcal{D}|} =
        \int_{s_1}^{s_2} \dfrac{|\partial\mathcal{D}|}{|\mathcal{D}|} K_p \bar{c}_p,
	\end{equation}
	
	\noindent and the integrals over $\mathcal{P}$ as:
    
	\begin{subequations}
        \begin{align}
    		\int_{\mathcal{P}}\dfrac{1}{|\mathcal{D}|}\dfrac{\partial \bar{c}_p}{\partial t} &=
            \int_{s_1}^{s_2} \dfrac{1}{|\mathcal{D}(s)|} \int_{\mathcal{D}(s)} \dfrac{\partial \bar{c}_p}{\partial t} =
            \int_{s_1}^{s_2}\dfrac{\partial \bar{c}_p}{\partial t},
    		\label{eq:reduction1d:cpq:dt}\\
    		\int_{\mathcal{P}}C_p&=
    		\int_{s_1}^{s_2} \int_{\mathcal{D}(s)} \bar{C}_p
            \int_{s_1}^{s_2} |\mathcal{D}|\bar{C}_p.
            \label{eq:reduction1d:cpq:c}
        \end{align}
	\end{subequations}
    
	\noindent If we substitute equations~\eqref{eq:reduction1d:cpq:fundamentaltheorem},~\eqref{eq:reduction1d:cpq:robin},~\eqref{eq:reduction1d:cpq:dt} and~\eqref{eq:reduction1d:cpq:c} into~\eqref{eq:reduction1d:cpq:divergencethm}, we obtain:
	
	\begin{multline}
        \int_{s_1}^{s_2}\dfrac{\partial \bar{c}_p}{\partial t} +
        \int_{s_1}^{s_2}\dfrac{\partial}{\partial s}\left(\omega_p\bar{\mu}_p\bar{c}_pE_\Lambda\right) +
        \int_{\Sigma}\omega_p\bar{\mu}_p \dfrac{\bar{c}_p}{|\mathcal{D}|}\mathbf{E}_g\cdot\mathbf{n}_g
        - \int_{s_1}^{s_2}\dfrac{\partial}{\partial s}\left(\bar{\nu}_p \dfrac{\partial \bar{c}_p}{\partial s}\right) + \\
        + \int_{s_1}^{s_2}\dfrac{|\partial\mathcal{D}|}{|\mathcal{D}|} K_p \bar{c}_p
        = \int_{s_1}^{s_2}|\mathcal{D}|\bar{C}_p,
        \ \forall s_1,\ s_2\in\Lambda.
        \label{eq:reduction1d:cpq:1}
	\end{multline}
	
	\noindent Consider now the integral over the lateral surface $\Sigma$ in equation~\eqref{eq:reduction1d:cpq:1}, representing the advective flux through the interface. We can partition $\Sigma$ into an inflow and an outflow part, for each family of charged particles (positive and negative), whose movement, due to the electric field, depends on their unit charge $\omega_p$. We set
	
	\begin{equation*}
		\Sigma = \Sigma^\text{in}_p\cup \Sigma^\text{out}_p = \bigcup_{s\in[s_1,s_2]} \left( \partial\mathcal{D}^{\text{in}}_p(s) \cup \partial\mathcal{D}^{\text{out}}_p(s)\right),
	\end{equation*}
	
	\noindent where $\partial\mathcal{D}^{\text{in}}_p$ and $\partial\mathcal{D}^{\text{out}}_p$ denote the inflow and outflow boundary of a section $\mathcal{D}$, respectively, for the $p^\text{th}$ family, and are defined as follows:
    % !!!!!!!!!!!!!!!!!!!!!!!!!!!!
    % dipendono da p
    % scrivere anche def con formula
    \begin{equation}
        \begin{aligned}
            \partial\mathcal{D}^{\text{in}}_p &= \{\mathbf{x}\in \partial\mathcal{D} \ :\: \omega_p\mathbf{E}_g\cdot\mathbf{n}_g (\mathbf{x}) \leq 0 \},\\
            \partial\mathcal{D}^{\text{out}}_p &= \{\mathbf{x}\in \partial\mathcal{D} \ :\: \omega_p\mathbf{E}_g\cdot\mathbf{n}_g (\mathbf{x}) \geq 0 \},
        \end{aligned}        
        \label{eq:reduction1d:def:inflow-outflow}
    \end{equation}
    
	\noindent for $p=1,\,2,\,3$. On the inflow boundary $\Sigma^{\text{in}}_p$, the concentrations of charged particles are set by the interface conditions~\eqref{eq:3d-3d:bc}, while their values are unknown on the outflow boundary $\Sigma^{\text{out}}_p$. Thus, also recalling that on the inflow boundary the charge concentrations are known from the boundary conditions, we can rewrite the integral over $\Sigma$ in equation~\eqref{eq:reduction1d:cpq:1} as follows:
	
	\begin{multline}
		\int_\Sigma \omega_p \bar{\mu}_p \dfrac{\bar{c}_p}{|\mathcal{D}| }\mathbf{E}_g\cdot\mathbf{n}_g =
		\int_{s_1}^{s_2}\bar{\mu}_p \dfrac{\bar{c}_p^b}{|\mathcal{D}|} \int_{\partial\mathcal{D}^{\text{in}}_p(s)} \omega_p \mathbf{E}_g\cdot\mathbf{n}_g + \int_{s_1}^{s_2}\bar{\mu}_p \dfrac{\bar{c}_p}{|\mathcal{D}|} \int_{\partial\mathcal{D}^{\text{out}}_p(s)} \omega_p \mathbf{E}_g\cdot\mathbf{n}_g = \\
		\int_{s_1}^{s_2}\bar{\mu}_p \dfrac{\bar{c}_p^b}{|\mathcal{D}| }\int_{\partial\mathcal{D}} \min\{0,\omega_p  \mathbf{E}_g\cdot\mathbf{n}_g\} + \int_{s_1}^{s_2}\bar{\mu}_p \dfrac{\bar{c}_p}{|\mathcal{D}|} \int_{\partial\mathcal{D}(s)} \max\{0,\omega_p \mathbf{E}_g\cdot\mathbf{n}_g\}.
		\label{eq:reduction1d:cpq:boundary}
	\end{multline}
	
	\noindent As mentioned above, the computation of the normal electric field $\mathbf{E}_g\cdot\mathbf{n}_g$ is not trivial, and it cannot be approximated as a constant, because of all the effects influencing the transversal components.
    % !!!!!!!!!!!!!!!!!!!!!!!!!!!!
    % prima di dire quella cosa sul campo elettrico normale, introdurre il punto che conoscere il campo elettrico trasversale nnormale non è banale (v. articolo vecchio), e non può essere costante. Vedremo dopo però che in qualche modo lo troviamo.
    In Section~\ref{section:reduction1d:electricField} we discuss our strategy to compute these integrals, exploiting the multi--scale nature of the problem. For now let us introduce the functionals $F_p^+\ : \ \Lambda \to \mathbb{R}$ and $F_p^-\ : \ \Lambda \to \mathbb{R}$, for $p=1,2,3$, defined as follows:

    \begin{equation}
        F_p^+(s) := \int_{\partial\mathcal{D}(s)} \max\{0,\omega_p\mathbf{E}_g\cdot\mathbf{n}_g\}, \quad
        F_p^-(s) := \int_{\partial\mathcal{D}(s)} \min\{0,\omega_p\mathbf{E}_g\cdot\mathbf{n}_g\}, \ \forall s \in\Lambda.
        \label{eq:reduction1d:def:F}
    \end{equation}

    \noindent We substitute them into~\eqref{eq:reduction1d:cpq:boundary} and then into~\eqref{eq:reduction1d:cpq:1}:
	
	\begin{multline}
		\int_{s_1}^{s_2}\dfrac{\partial \bar{c}_p}{\partial t} +
        \int_{s_1}^{s_2}\dfrac{\partial}{\partial s}\left(\omega_p\bar{\mu}_p\bar{c}_pE_\Lambda\right) 
		+ \int_{s_1}^{s_2}\bar{\mu}_p \dfrac{\bar{c}_p}{|\mathcal{D}| } F_p^+ +
        \int_{s_1}^{s_2}\bar{\mu}_p \dfrac{\bar{c}_p^b}{|\mathcal{D}|} F_p^- + \\
        - \int_{s_1}^{s_2}\dfrac{\partial}{\partial s}\left(\bar{\nu}_p \dfrac{\partial \bar{c}_p}{\partial s}\right) +
        \int_{s_1}^{s_2} \dfrac{|\partial\mathcal{D}|}{|\mathcal{D}|} K_p \bar{c}_p 
        = \int_{s_1}^{s_2}|\mathcal{D}|\bar{C}_p, \ \forall s_1,s_2\in\Lambda.
		\label{eq:reduction1d:cpq:2}
	\end{multline}
	
	\noindent Finally, since equation \eqref{eq:reduction1d:cpq:2} holds for any arbitrary $s_1$ and $s_2$ on the line $\Lambda$, we obtain the one-dimensional equation for the volume charge concentrations:
	
	\begin{equation}
		\dfrac{\partial \bar{c}_p}{\partial t} +
        \dfrac{\partial}{\partial s}\left(\omega_p\bar{\mu}_p\bar{c}_pE_\Lambda\right) +
        \bar{\mu}_p \dfrac{\bar{c}_p^b}{|\mathcal{D}|} F_p^- +
        \mu_p \dfrac{\bar{c}_p}{|\mathcal{D}|} F_p^+ %+ \\
        - \dfrac{\partial}{\partial s}\left(\bar{\nu}_p \dfrac{\partial \bar{c}_p}{\partial s}\right) +
        \dfrac{|\partial\mathcal{D}|}{|\mathcal{D}|} K_p\bar{c}_p = |\mathcal{D}|\bar{C}_p, \quad \text{on}\ \Lambda.
		\label{eq:reduction1d:cpq:1d}
	\end{equation}
	
	\subsection{Surface charge concentration}
	\label{section:reduction1d:surface_charge_concentration}
    
	Equations~\eqref{eq:3d-3d:2} and~\eqref{eq:3d-3d:3} describe the space-time evolution of the concentration of surface charge on the dielectric interface $ \Gamma $. We focus on the reduction of the second equation, concerning the negative charges, but the same procedure can be also applied to the equation for positive charge.
	
	Consider an arbitrary portion $\Sigma$ of $\Gamma$, defined as in Section~\ref{section:reduction1d:volume_charge_concentration}, and integrate equation~\eqref{eq:3d-3d:2} over it:
	
	\begin{equation}
		\int_\Sigma\dfrac{\partial c_{\Gamma,e}}{\partial t}-\int_\Sigma\nabla_{\Gamma}\cdot(c_{\Gamma,e}\mu_{\Gamma,e}\mathbf{E}_{\Gamma}) = \int_\Sigma\sum_{p=1}^2 c_p \mu_p \mathbf{E}_g\cdot\mathbf{n}_g -\int_\Sigma\sum_{p=1}^2\nu_p\nabla c_p \cdot \mathbf{n}_g
		\label{eq:reduction1d:cGammae:3d}
	\end{equation}
	
	\noindent By applying the divergence theorem on a 2D surface, we obtain:
	
	\begin{multline}
		-\int_\Sigma\nabla_{\Gamma}\cdot(c_{\Gamma,e}\mu_{\Gamma,e}\mathbf{E}_{\Gamma}) = -\int_{\partial\mathcal{D}(s_2)}c_{\Gamma,e}\mu_{\Gamma,e}\mathbf{E}_{\Gamma}\cdot\mathbf{s} + \int_{\partial\mathcal{D}(s_1)}c_{\Gamma,e}\mu_{\Gamma,e}\mathbf{E}_{\Gamma}\cdot\mathbf{s} = \\
		= - \int_{s_1}^{s_2}\dfrac{\partial}{\partial s}\int_{\partial\mathcal{D}(s)} c_{\Gamma,e}\mu_{\Gamma,e}\mathbf{E}_{\Gamma}\cdot\mathbf{s}.
        \label{eq:reduction1d:cGammae:divergence}
	\end{multline}
	
	\noindent We assume that the surface charge concentrations and the mobilities are constant on the boundary of sections of the gas domain $\Omega_g$, as well as the longitudinal component of $\mathbf{E}_\Gamma$:
	
	\begin{assumption}
		$c_{\Gamma,e}\simeq\dfrac{1}{|\partial\mathcal{D}(s)|}\bar{c}_{\Gamma,e}(s)$ and $c_{\Gamma,h}\simeq\dfrac{1}{|\partial\mathcal{D}(s)|}\bar{c}_{\Gamma,h}(s),\ $ where $\ \bar{c}_{\Gamma,e}(s)=\displaystyle\int_{\partial\mathcal{D}(s)}c_{\Gamma,e}$ and $\ \bar{c}_{\Gamma,h}(s)=\displaystyle\int_{\partial\mathcal{D}(s)}c_{\Gamma,h}, \quad \forall s\in\Lambda.$
		\label{assumption:reduction1d:c_gamma_const}
	\end{assumption}
	
	\begin{assumption}
		$\mu_{\Gamma,e}\simeq \bar{\mu}_{\Gamma,e}(s)$ and $\mu_{\Gamma,e}\simeq \bar{\mu}_{\Gamma,e}(s),\quad\forall s\in\Lambda$, where $\bar{\mu}_{\Gamma,e}$ denotes the integral means of $\mu_{\Gamma,e}$ on cross--sections of the gas domain, respectively.
		\label{assumption:reduction1d:mu_gamma_constant}
	\end{assumption}
	
	\begin{assumption}
        $\mathbf{E}_\Gamma\cdot\mathbf{s} \simeq E_\Lambda(s),\,\forall s \in \Lambda.$
		\label{assumption:reduction1d:E_gamma_const}
	\end{assumption}
	
	\noindent Keeping into account the cylindrical shape of $\Omega_g$ and Assumptions~\ref{assumption:reduction1d:c_gamma_const},~\ref{assumption:reduction1d:mu_gamma_constant} and~\ref{assumption:reduction1d:E_gamma_const}, and substituting equation~\eqref{eq:reduction1d:cGammae:divergence} into~\eqref{eq:reduction1d:cGammae:3d}, we obtain:
	
	\begin{multline}
		\int_{s_1}^{s_2} \dfrac{\partial \bar{c}_{\Gamma,e}}{\partial t}  - \int_{s_1}^{s_2}\dfrac{\partial}{\partial s}\left(\dfrac{\bar{c}_{\Gamma,e}}{|\partial\mathcal{D}|} \bar{\mu}_{\Gamma,e}E_\Lambda\right)  = \\
		= \sum_{p=1}^2\int_{s_1}^{s_2}\bar{\mu}_p \int_{\partial\mathcal{D}} \dfrac{\bar{c}_p}{|\mathcal{D}|} \mathbf{E}_g\cdot\mathbf{n}_g +
        \sum_{p=1}^2 \int_{s_1}^{s_2} K_p\dfrac{\bar{c}_p}{|\mathcal{D}|}|\partial\mathcal{D}|.
		\label{eq:reduction1d:cGammae:2}
	\end{multline}
	
	\noindent Recalling the definition~\eqref{eq:reduction1d:def:inflow-outflow} of inflow and outflow boundaries of sections, we set $\bar{c}_p = \bar{c}_p^b$ on $\partial\mathcal{D}^{\text{in}}_p,\, p=1,\,2$, where $\bar{c}_p^b$ is given by the boundary condition~\eqref{eq:3d-3d:bc}, and obtain:
	
	\begin{equation}
		\int_{\partial\mathcal{D}} \dfrac{\bar{c}_p}{|\mathcal{D}|} \mathbf{E}_g\cdot\mathbf{n}_g = \dfrac{\bar{c}_p}{|\mathcal{D}|} F_p^+ + \dfrac{\bar{c}_p^b}{|\mathcal{D}|} F_p^-.
		\label{eq:reduction1d:cGammae:boundary}
	\end{equation}
	
	\noindent Finally, we substitute~\eqref{eq:reduction1d:cGammae:boundary} into~\eqref{eq:reduction1d:cGammae:divergence} and get:
	
	\begin{multline}\nonumber
		\int_{s_1}^{s_2}\dfrac{\partial \bar{c}_{\Gamma,e}}{\partial t}  -
        \int_{s_1}^{s_2}\dfrac{\partial}{\partial s}\left(\bar{c}_{\Gamma,e}\bar{\mu}_{\Gamma,e} E_\Lambda\right) = \\ =
        \sum_{p=1}^2\int_{s_1}^{s_2}\bar{\mu}_p\dfrac{\bar{c}_p^b}{|\mathcal{D}|} F_p^- +
        \sum_{p=1}^2\int_{s_1}^{s_2}\bar{\mu}_p\dfrac{\bar{c}_p}{|\mathcal{D}|} F_p^+ +
        \int_{s_1}^{s_2}\sum_{p=1}^2 K_p\bar{c}_p\dfrac{|\partial\mathcal{D}|}{|\mathcal{D}|}, \quad \forall s_1,s_2\in\Lambda.
	\end{multline}
	
	\noindent Since this holds for any arbitrary $s_1$ and $s_2\in\Lambda$, then:
	
	\begin{equation}
		\dfrac{\partial \bar{c}_{\Gamma,e}}{\partial t}  -
        \dfrac{\partial}{\partial s}\left(\bar{c}_{\Gamma,e}\bar{\mu}_{\Gamma,e} E_\Lambda\right) =
        \sum_{p=1}^2\dfrac{\bar{\mu}_p\bar{c}_p}{|\mathcal{D}|}F_p^+ +
        \sum_{p=1}^2\dfrac{\bar{\mu}_p \bar{c}_p^b}{|\mathcal{D}|}F_p^- +
        \sum_{p=1}^2 K_p\bar{c}_p\dfrac{|\partial\mathcal{D}|}{|\mathcal{D}|}, \quad \text{on}\ \Lambda.
		\label{eq:reduction1d:cGammae:1d:negative}
	\end{equation}
	
	\noindent With a similar procedure, we obtain the following equation for the surface concentration of positive charge:
	
	\begin{equation}
		\dfrac{\partial \bar{c}_{\Gamma,h}}{\partial t}  +
        \dfrac{\partial}{\partial s}\left(\bar{c}_{\Gamma,h}\bar{\mu}_{\Gamma,h} E_\Lambda\right) =
        \dfrac{\mu_3\bar{c}_3}{|\mathcal{D}|}F_3^+ +
        \dfrac{\bar{\mu}_3 \bar{c}_3^b}{|\mathcal{D}|}F_3^- +
        K_3\bar{c}_3\dfrac{|\partial\mathcal{D}|}{|\mathcal{D}|}, \quad \text{on}\ \Lambda.
		\label{eq:reduction1d:cGammah:1d:negative}
	\end{equation}

    \noindent Notice that equations~\eqref{eq:reduction1d:cGammae:1d:negative} and~\eqref{eq:reduction1d:cGammah:1d:negative} have the similar structure, but the former, describing negative surface charge concentration, only depends on the negative volume charge concentrations $\bar{c}_1$ and $\bar{c}_2$ while the latter, describing positive surface charge concentration, only depends on the positive volume charge concentration $\bar{c}_3$. This models the exchange of charges between gas volume and dielectric interface.

    \begin{remark}
        The presence of positive surface charges, called \textit{holes}, is due to the absence of electrons, exchanged between the surface and the gas volume. In particular, a decrease in the concentration of surface electrons is related to a negative incoming flux of negative volume charged species ($\bar{c}_p,\, p=1,\,2$), which corresponds to a positive incoming flux of positive charged species ($\bar{c}_3$), while the opposite holds for the surface holes. Because of the conservation of total charge, we can substitute the inflow fluxes of negative charged particles with the opposite of the inflow fluxes of positive charges, and vice versa, obtaining:\begin{subequations}
        \begin{alignat}{3}
            \dfrac{\partial \bar{c}_{\Gamma,e}}{\partial t}  -
            \dfrac{\partial}{\partial s}\left(\bar{c}_{\Gamma,e}\bar{\mu}_{\Gamma,e} E_\Lambda\right) &=
            \sum_{p=1}^2\dfrac{\bar{\mu}_p\bar{c}_p}{|\mathcal{D}(s)|}F_p^+ -
            \dfrac{\mu_3\bar{c}_3^b}{|\mathcal{D}|}F_3^- +
            \sum_{p=1}^2 K_p\bar{c}_p\dfrac{|\partial\mathcal{D}|}{|\mathcal{D}|}, &\quad \text{on}\ \Lambda.
    		\label{eq:reduction1d:cGammae:1d}\\
            \dfrac{\partial \bar{c}_{\Gamma,h}}{\partial t}  +
            \dfrac{\partial}{\partial s}\left(\bar{c}_{\Gamma,h}\bar{\mu}_{\Gamma,h} E_\Lambda\right) &=
            \dfrac{\mu_3\bar{c}_3}{|\mathcal{D}|}F_3^+ 
            -\sum_{p=1}^2\dfrac{\bar{\mu}_p\bar{c}_p^b}{|\mathcal{D}|}F_p^- +
            K_3\bar{c}_3\dfrac{|\partial\mathcal{D}|}{|\mathcal{D}|}, &\quad \text{on}\ \Lambda.
    		\label{eq:reduction1d:cGammah:1d}
        \end{alignat}
    \end{subequations}
    \end{remark}

	\subsection{Dipole moment}
	\label{section:reduction1d:dipole_moment}
    Equations~\eqref{eq:reduction1d:cGammae:1d} and~\eqref{eq:reduction1d:cGammah:1d} allow to compute the average surface charge on cross--sections of the gas domain, namely its moment of order 0. However, if we want to have more precise information on the charge distribution on cross--sections, we also need to consider the moments of order 1.\\
	Consider equations~\eqref{eq:3d-3d:2} and~\eqref{eq:3d-3d:3}, and denote by $e$ the electron charge. By multiplying the former by $-e$ and the latter by $e$, and then summing the two equations together, we obtain an equation for the total surface charge $q_\Gamma = e \left( c_{\Gamma,h} - c_{\Gamma,e} \right)$:
	
	\begin{equation}
		\dfrac{\partial q_\Gamma}{\partial t}+\nabla_\Gamma \cdot(\sigma_\Gamma \mathbf{E}_\Gamma) = \sum_{p=1}^3 e c_p \mu_p (\mathbf{E}_g\cdot \mathbf{n}_g) \omega_p - \sum_{p=1}^3 e \nu_p (\nabla c_p \cdot \mathbf{n}_g) \omega_p,
		\label{eq:reduction1d:dipole:1}
	\end{equation}
    % !!!!!!!!!!!!!!!!!!!!!!!!!!!!
    % q_\Gamma è una funzione che dà la distribuzione di carica sulla superficie, non è il momento di dipolo
    % la carica totale netta è il momento do ordine 0
    % il momento di dipolo lungo x è q_gamma^1, qiello lungo y è q_gamma^2
    % approssimazione con tre polinomi di Legendre, di cui l'interpretazione è momento di dipolo x e y
	
	\noindent where $\sigma_\Gamma = e c_{\Gamma, h} \mu_{\Gamma,h} + e c_{\Gamma, e} \mu_{\Gamma,e}$ is the electrical conductivity on $\Gamma$. From Assumptions \ref{assumption:reduction1d:c_gamma_const} and \ref{assumption:reduction1d:mu_gamma_constant} follows that the both these quantities are constant on sections, i.e. $q_\Gamma \simeq \dfrac{1}{|\partial\mathcal{D}|} \bar{q}_\Gamma$, with $\bar{q}_\Gamma = e\bar{c}_{\Gamma,h} - e\bar{c}_{\Gamma,e}$, and $\sigma_\Gamma \simeq \dfrac{1}{|\partial\mathcal{D}|}\bar{\sigma}_\Gamma$ on $\mathcal{D}$, with $\bar{\sigma}_\Gamma = e(\bar{c}_{\Gamma, e}\mu_{\Gamma,e} + \bar{c}_{\Gamma, h}\mu_{\Gamma,h})$. \\
    However, the non--radial transverse components of the electric field are strongly influenced by their distribution on cross--sections, which we represent by exploiting dipole moments. We assume that $q_\Gamma$ can be expressed as a linear combination of three basis functions $\phi_i,\ i=0,\,1,\,2$, defined on $\partial\mathcal{D}$.

    \begin{assumption}
		$q_\Gamma = q_\Gamma^0 \phi_0 + \sum_{j=1}^2 q_\Gamma^j \phi_j$, where $\left\{q_\Gamma^j\right\}_{j=0}^2$ are the expansion coefficients, constant on $\partial\mathcal{D}$, with respect to the basis $\{\phi_i\}_{i=0}^2$, such that the functions $\phi_i: \partial\mathcal{D}\to\mathbb{R},\, i=0,\,1,\,2,$ satisfy the following properties:
        \begin{property}
            $\phi_0 = 1$;
            \label{property:phi_01}
        \end{property}
        \begin{property}
            $\displaystyle\int_{\partial \mathcal{D}}\phi_0 \phi_i =0, \quad i=1,2$.
            \label{property:phi_othogonality}
        \end{property}
		\label{assumption:reduction1d:basis}
	\end{assumption}
    % !!!!!!!!!!!!!!!!!!!!!!!!!!!!
    % sottolineare che non ho fissato ancora la base ma qualsiasi essa sia deve soddisfare questa proprietà
    % non è per forza una base ortogonale, ma almeno \phi_1 e \phi_2 devono essere ortogonali a \phi_0
	
	\noindent If we multiply equation \eqref{eq:reduction1d:dipole:1} by $\phi_i, \ i=1,\,2$, and integrate over a portion $\Sigma$ of $\Gamma$, delimited by sections $\partial\mathcal{D}(s_1)$ and $\partial\mathcal{D}(s_2)$, $s_1,\ s_2\in\Lambda$, we obtain:
	
	\begin{equation}
		\int_{s_1}^{s_2}\int_{\mathcal{D}(s)}\dfrac{\partial q_\Gamma}{\partial t}\phi_i+
        \int_\Sigma\nabla_\Gamma \cdot(\sigma_\Gamma \mathbf{E}_\Gamma)\phi_i = 
        \int_\Sigma\sum_{p=1}^3 e c_p \mu_p \mathbf{E}_g\cdot \mathbf{n}_g \phi_i %+ \\
        -\int_\Sigma\sum_{p=1}^3 e \nu_p (\nabla c_p \cdot \mathbf{n}_g) \omega_p \phi_i,\quad i=1,2.
        \label{eq:reduction1d:dipole:3d}
	\end{equation}
	
	\noindent We integrate by parts the divergence term:
    
	\begin{equation}\nonumber
		\int_\Sigma\nabla_\Gamma \cdot(\sigma_\Gamma \mathbf{E}_\Gamma)\phi_i = 
         \int_\Sigma \nabla_\Gamma \cdot (\sigma_\Gamma \mathbf{E}_\Gamma \phi_i) -
         \int_\Sigma\sigma_\Gamma\mathbf{E}_\Gamma\cdot\nabla\phi_i,\quad i = 1,2,
    \end{equation}

    \noindent and apply the divergence theorem and fundamental theorem of calculus to the first integral, noticing that $\partial\Sigma = \partial\mathcal{D}(s_1) \cup \partial\mathcal{D}(s_2)$:
    
    \begin{multline*}
        \int_\Sigma \nabla_\Gamma \cdot (\sigma_\Gamma \mathbf{E}_\Gamma \phi_i) =
        \int_{\partial\mathcal{D}(s_2)} \sigma_\Gamma (\mathbf{E}_\Gamma \cdot \mathbf{s})\phi_i - \int_{\partial \mathcal{D}(s_1)}\sigma_\Gamma(\mathbf{E}_\Gamma \cdot \mathbf{s})\phi_i  = \\ =
		\int_{s_1}^{s_2}\dfrac{\partial}{\partial s}\left(\int_{\partial\mathcal{D}(s)} \sigma_\Gamma(\mathbf{E}_\Gamma\cdot\mathbf{s})\phi_i\right), \quad i=1,2.
	\end{multline*}
	
	\noindent	Substituting $q_\Gamma$ with the expansion~\eqref{eq:reduction1d:dipole:3d} introduced in Assumption~\ref{assumption:reduction1d:basis}, we get, for $i=1,\,2$:
	
	\begin{equation}\nonumber
		\int_\Sigma\dfrac{\partial q_\Gamma}{\partial t}\phi_i =
        \int_{s_1}^{s_2}\int_{\partial\mathcal{D}(s)}\dfrac{\partial q_\Gamma^0}{\partial t}\phi_0\phi_i +
        \int_{s_1}^{s_2}\int_{\partial\mathcal{D}(s)}\sum_{j=1}^2\dfrac{\partial q_\Gamma^j}{\partial t}\phi_j\phi_i =
        \int_{s_1}^{s_2}\sum_{j=1}^2\dfrac{\partial q_\Gamma^j}{\partial t} \int_{\partial\mathcal{D}(s)}\phi_j\phi_i,
	\end{equation}
	
	\noindent because of Properties~\ref{property:phi_01} and~\ref{property:phi_othogonality}. Moreover, from Property~\ref{property:phi_01}, recalling Assumption~\ref{assumption:reduction1d:E_gamma_const}, we can deduce that

    \begin{equation*}
        \int_{\partial\mathcal{D}(s)} \sigma_\Gamma (\mathbf{E}_\Gamma \cdot \mathbf{s})\phi_i \simeq \dfrac{\bar{\sigma}_\Gamma E_\Lambda}{|\partial\mathcal{D}(s)|}\int_{\partial\mathcal{D}(s)}\phi_i\phi_0 = 0,\quad \text{for}\ i=1,\,2,\ \forall s\in[0,S],
    \end{equation*}

    \noindent and, from Assumptions~\ref{assumption:robin} and~\ref{assumption:reduction1d:cpq_constant},
    % !!!!!!!!!!!!!!!!!!!!!!!!!!!!
    %specificare meglio nei conti

    \begin{multline*}
        \int_{s_1}^{s_2}\int_{\partial\mathcal{D}(s)}\sum_{p=1}^3 e \nu_p\omega_p (\nabla c_p \cdot \mathbf{n}_g)\phi_i \simeq
        \int_{s_1}^{s_2}\int_{\partial\mathcal{D}(s)} \sum_{p=1}^3 e\omega_p K_p \dfrac{\bar{c}_p}{|\mathcal{D}|}\phi_i =\\ =
        \sum_{p=1}^3 \int_{s_1}^{s_2} e\omega_p K_p \dfrac{\bar{c}_p}{|\mathcal{D}|} \int_{\partial\mathcal{D}} \phi_i\phi_0 =0.
    \end{multline*}

    \noindent Finally, if we split each cross--section of $\Sigma$ into inflow and outflow boundaries, and define the functionals $G_{p,i}^+\ : \ \Lambda\to\mathbb{R}$ and $G_{p,i}^-\ : \ \Lambda\to\mathbb
    R$, for $p=1,2,3$ and $i=1,2$, as:

    \begin{equation}
        G_{p,i}^+(s) := \int_{\partial \mathcal{D}(s)}\max\{0,\omega_p\mathbf{E}_g\cdot\mathbf{n}_g\}\phi_i, \quad
        G_{p,i}^-(s) := \int_{\partial \mathcal{D}(s)}\min\{0,\omega_p\mathbf{E}_g\cdot\mathbf{n}_g\}\phi_i, \ \forall s \in\Lambda,
        \label{eq:reduction1d:def:G}
    \end{equation}

    \noindent we obtain the following equations:
	
	\begin{equation}
		\sum_{j=1}^2 \dfrac{\partial q_\Gamma^j}{\partial t} M_{ij} - \dfrac{\bar{\sigma}_\Gamma}{|\partial\mathcal{D}|}H_i = \sum_{p=1}^3 e\bar{c}_p\dfrac{\bar{\mu}_p}{|\mathcal{D}|} G_{p,i}^+ + \sum_{p=1}^3 e\bar{c}_p^b \dfrac{\bar{\mu}_p }{|\mathcal{D}|} G_{p,i}^-, \quad i=1,2, \quad \mathrm{on} \ \Lambda,
		\label{eq:reduction1d:dipole:2}
	\end{equation}
	
	\noindent where 
    \begin{equation}
        M_{ij}(s):=\int_{\partial \mathcal{D}}\phi_i\phi_j, \ i,\, j=1,\,2,\ \forall s\in [0,S],
        \label{eq:reduction1d:def:M}
    \end{equation}
    \noindent and
    \begin{equation}
        H_i(s) := \int_{\partial\mathcal{D}(s)} \mathbf{E}_\Gamma\cdot\nabla\phi_i, \ i=1,\, 2,\ \forall s\in [0,S].
        \label{eq:reduction1d:def:H}
    \end{equation}

    \noindent Note that also $G_{p,i}^+,\ G_{p,i}^- \ \text{and}\ H_i$ are integrals of 3D fields on curves, as $F_p^+$ and $F_p^-$, and will be treated in a similar way, as we will see in Section~\ref{section:reduction1d:FGH}.

    \begin{remark}
        We have not defined explicit expressions of the basis functions introduced in Assumption~\ref{assumption:reduction1d:basis} yet, and the previous discussion is independent of them, as long as Properties~\ref{property:phi_01} and~\ref{property:phi_othogonality} are satisfied. In Section~\ref{section:reduction1d:electricField} we introduce a definition of $\phi_1$ and $\phi_2$, according to which $q_\Gamma^1\phi_1$ and $q_\Gamma^2\phi_2$ can be interpreted as orthogonal components of the dipole moment on cross--sections of $\Omega_g$.
    \end{remark}
    
    % \begin{remark}
    %     If we multiply equation~\eqref{eq:reduction1d:dipole:1} by $\phi_0$ and integrate over a portion $\Sigma$ of $\Gamma$, we only obtain a linear combination of equations~\eqref{eq:reduction1d:cGammae:1d} and~\eqref{eq:reduction1d:cGammah:1d}.
    % \end{remark}
	
	\subsection{Complete system for the charge concentration}
	\label{section:reduction1d:complete_system_for_the_charge_concentration}
	The complete reduced 1D system for charge concentrations is given by ~\eqref{eq:reduction1d:cpq:1d},~\eqref{eq:reduction1d:cGammae:1d}, \eqref{eq:reduction1d:cGammah:1d} and~\eqref{eq:reduction1d:dipole:2}:
	
    \begin{subnumcases}{\label{eq:reduction1d:charge:1d}}
        \begin{split}
            \dfrac{\partial \bar{c}_p}{\partial t} +
            \dfrac{\partial}{\partial s}  \left(\omega_p\bar{\mu}_p\bar{c}_pE_\Lambda\right) +
            \bar{\mu}_p \dfrac{\bar{c}_p}{|\mathcal{D}|} F_p^+ &+
            \bar{\mu}_p \dfrac{\bar{c}_p^b}{|\mathcal{D}|} F_p^- 
            - \dfrac{\partial}{\partial s}\left(\bar{\nu}_p \dfrac{\partial \bar{c}_p}{\partial s}\right) +\\
            & +\dfrac{|\partial\mathcal{D}|}{|\mathcal{D}|} K_p\bar{c}_p = |\mathcal{D}|\bar{C}_p,\  p=1,2,3,
        \end{split}
         & $\text{on}\ \Lambda,$\\
		\label{eq:reduction1d:charge:1d:1}
        \dfrac{\partial \bar{c}_{\Gamma,e}}{\partial t}  -
        \dfrac{\partial}{\partial s}\left(\bar{c}_{\Gamma,e}\bar{\mu}_{\Gamma,e} E_\Lambda\right) =
        \sum_{p=1}^2\dfrac{\bar{\mu}_p\bar{c}_p}{|\mathcal{D}|}F_p^+ -
        \dfrac{\bar{\mu}_3 \bar{c}_3^b}{|\mathcal{D}|}F_3^- +
        \sum_{p=1}^2 K_p\bar{c}_p\dfrac{|\partial\mathcal{D}|}{|\mathcal{D}|}, & $\text{on}\ \Lambda,$ \\
		\label{eq:reduction1d:charge:1d:2}
        \dfrac{\partial \bar{c}_{\Gamma,h}}{\partial t}  +
        \dfrac{\partial}{\partial s}\left(\bar{c}_{\Gamma,h}\bar{\mu}_{\Gamma,h} E_\Lambda\right) =
        \dfrac{\mu_3\bar{c}_3}{|\mathcal{D}|}F_3^+ -
        \sum_{p=1}^2\dfrac{\bar{\mu}_p \bar{c}_p^b}{|\mathcal{D}|}F_p^- +
        K_3\bar{c}_3\dfrac{|\partial\mathcal{D}|}{|\mathcal{D}|}, & $\text{on}\ \Lambda,$ \\
		\label{eq:reduction1d:charge:1d:3}
        \sum_{j=1}^2 \dfrac{\partial q_\Gamma^j}{\partial t} M_{ij} = \dfrac{\bar{\sigma}_\Gamma}{|\partial\mathcal{D}|}H_i + \sum_{p=1}^3 \left(e\bar{c}_p\dfrac{\bar{\mu}_p}{|\mathcal{D}|} G_{p,i}^+ + e\bar{c}_p^b \dfrac{\bar{\mu}_p }{|\mathcal{D}|} G_{p,i}^+\right), \quad i=1,2, & $\mathrm{on} \ \Lambda.$
		\label{eq:reduction1d:charge:1d:4}
    \end{subnumcases}
	
	\noindent Observe that the coupling among these equations is weak and embedded in $F_p^\pm$ and $G_{p,i}^\pm$, $p=1,\,2,\,3,\ i=1,\,2$ because of their dependence on the transverse electric field, which in turn depends on the charges, as detailed in Section~\ref{section:reduction1d:electricField}. However, it is mutual and they cannot be decoupled.
%	Indeed, $F_p^\pm$ depend on $q_\Gamma^j,\, j=1,\,2,$ for all $p=1,\,2,\,3$, thus equations~\eqref{eq:reduction1d:charge:1d:1} -~\eqref{eq:reduction1d:charge:1d:3} depend on equation~\eqref{eq:reduction1d:charge:1d:4}. Moreover, equation~\eqref{eq:reduction1d:charge:1d:4} depends on~\eqref{eq:reduction1d:charge:1d:1},~\eqref{eq:reduction1d:charge:1d:2} and~\eqref{eq:reduction1d:charge:1d:3}, through $G_{p,i}^\pm$ and $\bar{\sigma}_s$, and finally equations~\eqref{eq:reduction1d:charge:1d:2} and~\eqref{eq:reduction1d:charge:1d:3} depend on equation~\eqref{eq:reduction1d:charge:1d:1}, as their right--hand sides are funcitons of the solutions to~\eqref{eq:reduction1d:charge:1d:1}.
    % If the values of $F_p^\pm$ were known a priori, we could decouple the system by solving the first independent equation~\eqref{eq:reduction1d:charge:1d:1}, substituting its solution into the right--hand side of~\eqref{eq:reduction1d:charge:1d:2} and~\eqref{eq:reduction1d:charge:1d:3}, and finally substituting the solutions to~\eqref{eq:reduction1d:charge:1d:1},~\eqref{eq:reduction1d:charge:1d:2} and~\eqref{eq:reduction1d:charge:1d:3} into~\eqref{eq:reduction1d:charge:1d:4}. This way, we would be left with four separate problems on the 1D domain, which could be solved in series: a diffusion-transport-reaction problem for the volume charge concentrations, two transport problems for the surface charge concentrations and a first order Ordinary Differential Equation (ODE) for each component of the dipole moment.
	
	We point out that this problem can be extended to ramified domains with thin branches, such as the electrical treeing, where the reduced domain coincides with a 1D graph, by imposing charge conservation conditions at the junctions.

	\section{Electric field}
	\label{section:reduction1d:electricField}
	As discussed in Section~\ref{section:reduction1d:3dProblem}, Assumption~\ref{assumption:reduction1d:cpq_constant} implies the presence of a transverse electric field in both $\Omega_g$ and $\Omega_s$. Thanks to the model derived in~\cite{crippa2024mixed}, the longitudinal component of $\mathbf{E}_g$ along $\Lambda$ can be computed, as well as the whole complete electric field $\mathbf{E}_s$ in the external dielectric domain. However, an approximation of the transversal electric field in the gas cannot be trivially obtained from the reduced model discussed in~\cite{crippa2024mixed}, since only a constant charge distribution on sections of $\Omega_g$ and their boundary is considered. 
    %This model is accurate enough to compute the external electric field and the profile of the potential inside $\Omega_g$, but not the transversal components of the electric field in the gas. Moreover, the restriction of $\mathbf{E}_s$ on $\Lambda$ is not an accurate enough approximation of the transverse component of $\mathbf{E}_g$ both because it relies on the assumption of homogeneous charge distribution on sections.
    \\
    To compute the transverse electric field in the gas we exploit the thin and elongated shape on the treeing branches. As we will see in Section~\ref{section:reduction1d:adimensionalization}, this allows to decouple its longitudinal component from the transverse component. Moreover, since the branches of the tree are thinner than the external bulk (Assumption~\ref{assumption:coupled3d1d:thin}), we can introduce a mid--scale cylinder $\Omega_m$, containing $\Omega_g$ as in Figure~\ref{figure:3d-3d:domain}, which is large enough not to be affected by the presence of charge in the gas, but still smaller than the dielectric bulk. This way we can use the mixed--dimensional electrostatic model presented in~\cite{crippa2024mixed} to compute the longitudinal electric field in the gas and on $\partial\Omega_s$. We call the latter ``external electric field''.
    In this section we compute by superimposition of effects the transverse components of the electric field $\mathbf{E}_g$ on sections of the gas domain, taking into account different sources: two orthogonal components of a given external electric field, the surface charge distributions, described by $q_\Gamma^j,\ j=0,1,2$, and the average volume total charge concentration, given by the sums of the 1D unknown charge concentrations multiplied by the respective unite charges of each family. %As ``external electric field'' we consider the electric field inside the dielectric domain, at a relatively far distance from the gas, where the effects of the charges present in the micro--scale of the gas domain are negligible on the macro--scale of the dielectric. Thus, for the solution to the electrostatic problem~\eqref{eq:electric_field:3d:dim} on sections, we will denote by $\Omega_s$ a mid--scale cylinder, coaxial to $\Omega_g$ and containing it (Figure~\ref{figure:3d-3d:domain}). 
    
    In the following we discuss an approach for the computation of the transverse electric field.

    \subsection{Adimensionalization}
    \label{section:reduction1d:adimensionalization}
    In order to get rid of the dependence of the problem on the dimension of the inner gas domain $\Omega_g$, we start by adimensionalizing the equations. Given a quantity $f$, we split it into the product of an adimensional part $\tilde{f}$ and the corresponding units of measure $\hat{f}$, so that it can be expressed as $f = \tilde{f}\hat{f}$. Moreover, we normalize the distance from the axis of $\Omega$ by dividing it by the radius $R_g$ of the 3D gas domain $\Omega_g$, and we normalize the length of $\Omega$ by dividing it by the length $L$ of the gas domain $\Omega_g$. \\
    Consider a cylindrical coordinate system $(r,\theta,s)$, centered on $\Lambda$. Then, the adimensionalized coordinates are $(\tilde{r},\tilde{\theta},\tilde{s}) = \left(\dfrac{r} {R_g},\theta, \dfrac{s}{L}\right)$.
    
    Let us rewrite the adimensionalized divergence $\tilde{\nabla}\cdot$ of the electric field, in terms of the potential:

    \begin{align*}
        -\nabla\cdot\mathbf{E} &=
        \nabla\cdot(\nabla\Phi) = 
        \dfrac{1}{r} \dfrac{\partial}{\partial r} \left( r \dfrac{\partial \Phi}{\partial r} \right) + \dfrac{1}{r^2} \dfrac{\partial^2 \Phi}{\partial \theta^2} + \dfrac{\partial^2 \Phi}{\partial s^2} = \\ &=
        \dfrac{1}{\tilde{r}R_g} \dfrac{\partial}{\partial \tilde{r}}\left( \tilde{r}R_g \dfrac{\partial \Phi}{\tilde{r}} \dfrac{\partial\tilde{r}}{\partial r} \right)\dfrac{\partial\tilde{r}}{\partial r} + \dfrac{1}{\tilde{r}^2 R_g^2} \dfrac{\partial^2 \Phi}{\partial\tilde{\theta}^2} + \dfrac{\partial}{\partial \tilde{s}} \left( \frac{\partial \Phi}{\partial \tilde{s}} \dfrac{\partial\tilde{s}}{\partial s} \right)\dfrac{\partial\tilde{s}}{\partial s} = \\ &=
        \dfrac{1}{\tilde{r}R_g} \dfrac{\partial}{\partial \tilde{r}}\left( \tilde{r}R_g \dfrac{\partial \Phi}{\tilde{r}} \dfrac{1}{R_g} \right)\dfrac{1}{R_g} + \dfrac{1}{\tilde{r}^2 R_g^2} \dfrac{\partial^2 \Phi}{\partial\tilde{\theta}^2} + \dfrac{\partial}{\partial \tilde{s}} \left( \frac{\partial \Phi}{\partial \tilde{s}} \dfrac{1}{L} \right)\dfrac{1}{L} = \\ &=
        \dfrac{1}{R_g^2} \left( \left( \dfrac{1}{\tilde{r}}\dfrac{\partial}{\partial \tilde{r}} \left( \tilde{r}\dfrac{\partial\tilde{\Phi}}{\partial\tilde{r}} \right) + \dfrac{1}{\tilde{r}} \dfrac{\partial^2 \tilde{\Phi}}{\partial\tilde{\theta}^2} \right) \hat{\Phi} +\frac{\partial^2 \Phi}{\partial \tilde{s}^2}\dfrac{R_g^2}{L^2} \right).
    \end{align*}

    \noindent Notice that the last term can be neglected since, by Assumption~\ref{assumption:coupled3d1d:thin}, $R_g\ll L$, and the first term represents the 2D divergence $\tilde{\nabla}_{\tilde{\mathcal{D}}}\cdot$ of the electric field on an adimensional section $\tilde{\mathcal{D}}$ of $\Omega_g$, of radius 1. Thus,

    \begin{equation}
        \nabla\cdot\mathbf{E} \approx \dfrac{\hat{\Phi}}{R_g^2} \tilde{\nabla}_{\tilde{\mathcal{D}}} \cdot\tilde{\mathbf{E}} .
        \label{eq:reduction1d:adim:divE}
    \end{equation}

    \noindent Moreover, the radial component of the electric field can be written as:

    \begin{equation}
        \mathbf{E}\cdot\mathbf{r} =
        -\dfrac{\partial \Phi}{\partial r} = 
        -\dfrac{\partial\tilde{\Phi}}{\partial\tilde{r}}\hat{\Phi}\dfrac{\partial\tilde{r}}{\partial r} =
        -\dfrac{\partial\tilde{\Phi}}{\partial\tilde{r}}\dfrac{\hat{\Phi}}{R_g} = 
        \tilde{\mathbf{E}}\cdot\mathbf{r}\dfrac{\hat{\Phi}}{R_g}.
        \label{eq:reduction1d:adim:En}
    \end{equation}
    
    \noindent Substituting all the quantities in problem~\eqref{eq:electric_field:3d:dim} by separating the adimensional and dimensional parts, and replacing the divergence and the radial component of the electric field as in equations~\eqref{eq:reduction1d:adim:divE} and~\eqref{eq:reduction1d:adim:En}, we obtain the following 2D system:

    \begin{subnumcases}{\label{eq:reduction1d:electric_field:3d:adim}}
        \tilde{\nabla}_{\tilde{\mathcal{D}}}\cdot\tilde{\mathbf{E}}_g = \mathcal{K}\tilde{q}, &$\text{in}\ \tilde{\mathcal{D}},$
        \label{eq:reduction1d:electric_field:3d:1:adim}\\
        \tilde{\nabla}_{\tilde{\mathcal{D}}}\cdot(\epsilon_s\tilde{\mathbf{E}}_s) = 0, &$\text{in}\ \tilde{\mathcal{D}}_s, $
        \label{eq:reduction1d:electric_field:3d:2:adim}\\
        \tilde{\mathbf{E}} = -\tilde{\nabla}_{\tilde{\mathcal{D}}}\tilde{\Phi}, & $\text{in}\ \tilde{\mathcal{D}}\cup\tilde{\mathcal{D}}_s,$
        \label{eq:reduction1d:electric_field:3d:3:adim}\\
        \tilde{\Phi}_g = \tilde{\Phi}_s, & $\text{on}\ \partial\tilde{\mathcal{D}},$
        \label{eq:reduction1d:electric_field:3d:4:adim} \\
        \epsilon_s \tilde{\mathbf{E}}_s\cdot\mathbf{n}_s + \tilde{\mathbf{E}}_g\cdot\mathbf{n}_g = -\mathcal{K}_\Gamma \tilde{q}_\Gamma, &$ \text{on}\ \partial\tilde{\mathcal{D}}, $
        \label{eq:reduction1d:electric_field:3d:5:adim}\\
        \tilde{\mathbf{E}}_s \cdot\mathbf{n} = \tilde{E}^b, & $\text{on}\ \partial(\tilde{\mathcal{D}}\cup\tilde{\mathcal{D}}_s),$
        \label{eq:reduction1d:electric_field:3d:6:adim}
    \end{subnumcases}

    \noindent where $\mathcal{K}=\dfrac{\hat{q}R_g^2}{\hat{\Phi}\epsilon_0}$ and $\mathcal{K}_\Gamma = \dfrac{\hat{q}_\Gamma R_g}{\epsilon_0 \hat{\Phi}}$, and $\tilde{\mathcal{D}}_s$ represents an annulus of the adimensionalized domain $\tilde{\Omega}_s$, included between $\tilde{\mathcal{D}}_g$ and a circle of radius $\tilde{R}_s=\dfrac{R_s}{R_g}$, such that $\tilde{\mathcal{D}}_g \cup \tilde{\mathcal{D}}_s = \tilde{\mathcal{D}}$ (see Figure~\ref{figure:reduction1d:adimensional_sections}).	

    \subsection{Superimposition of effects}
    \label{section:reduction1d:superimposition_of_effects}
	The transversal components of the electric field, given by the solution to problem~\eqref{eq:reduction1d:electric_field:3d:adim}, can be seen as a superimposition of six effects.
    In the following, we decompose the external electric field into its $x$ and $y$ components, considering $x$ and $y$ as orthogonal axes of a cartesian system of coordinates, where the third orthogonal axis $z$ locally coincides with $\Lambda$. With a little abuse of notation, the effects of the six above mentioned effects can be determined as solutions $\mathbf{\Psi}^i,\ i=1,\,2,\,3,\,4,\,5,\,6$ to problem~\eqref{eq:reduction1d:electric_field:3d:adim}, with different data:
    
	\begin{alignat}{4}
		&\text{1.\ Effect of the}\ x\ \text{component of the external electric field:}&\quad &
            \tilde{q} = 0, \, \tilde{q}_\Gamma = 0, \, \tilde{E}^b = \cos\tilde{\theta};
            \tag{S1}
            \label{eq:reduction1d:superimposition:x}\\[1.8ex]
		&\text{2.\ Effect of the}\ y\ \text{component of the external electric field:}&\quad &
			\tilde{q} = 0, \, \tilde{q}_\Gamma = 0, \, \tilde{E}^b = \sin\tilde{\theta};
            \tag{S2}
			\label{eq:reduction1d:superimposition:y}\\[1.8ex]
		&\text{3.\ Effect of the surface charge distribution}\ q_\Gamma^1:&\ &
			\tilde{q} = 0, \, \tilde{q}_\Gamma = \frac{\phi_1}{\mathcal{K}_\Gamma}, \, \tilde{E}^b = 0;
            \tag{S3}
			\label{eq:reduction1d:superimposition:q1}\\[1.8ex]
		&\text{4.\ Effect of the surface charge distribution}\ q_\Gamma^2:&\ &
			\tilde{q} = 0, \, \tilde{q}_\Gamma = \frac{\phi_2}{\mathcal{K}_\Gamma}, \, \tilde{E}^b = 0;
            \tag{S4}
			\label{eq:reduction1d:superimposition:q2}\\[1.5ex]
		&\text{5.\ Effect of the\ mean\ total volume\ charge}\ \bar{q}:&\ &
			\tilde{q}=\frac{1}{\mathcal{K}}, \, \tilde{q}_\Gamma = 0, \, \tilde{E}^b = \frac{1}{2\tilde{R}_s\epsilon_s};
            \tag{S5}
			\label{eq:reduction1d:superimposition:q}\\
		&\text{6.\ Effect of the mean total surface charge}\ q_\Gamma^0:&\ &
			\tilde{q} = 0,\, \tilde{q}_\Gamma = \frac{1}{\mathcal{K}_\Gamma},\, \tilde{E}^b = \frac{1}{\epsilon_s\tilde{R}_s}.
            \tag{S6}
			\label{eq:reduction1d:superimposition:q0}
	\end{alignat}
    %!!!!!!!!!!!!!!!!!!!!!!!!!!!!!!!!!!!!
    % qui diciamo solo che prendiamo queste \phi per semolificare i conti
    % poi vediamo l'interpretazione

	\noindent In problems~\eqref{eq:reduction1d:superimposition:q1} and~\eqref{eq:reduction1d:superimposition:q2}, $\phi_1$ and $\phi_2$ denote the two basis functions introduced in Section~\ref{section:reduction1d:dipole_moment}. We can choose $\phi_1$ and $\phi_2$ in such a way that problems~\eqref{eq:reduction1d:superimposition:q1} and~\eqref{eq:reduction1d:superimposition:q2} describe the effect of the surface charge distribution in the $x$ and $y$ directions, respectively: $\phi_1=-\mathbf{\Psi}_g^1\cdot\mathbf{n}_g$ and $\phi_2=-\mathbf{\Psi}_g^2\cdot\mathbf{n}_g$. This way, $\mathbf{\Psi}^3$, such that $\mathbf{\Psi}^3_g = \mathbf{\Psi}_g^1$ and $\mathbf{\Psi}_s^3 =0$, and $\mathbf{\Psi}^4$, such that $\mathbf{\Psi}^4_g = \mathbf{\Psi}_g^2$ and $\mathbf{\Psi}_s^4 =0$, are solutions of problems~\eqref{eq:reduction1d:superimposition:q1} and~\eqref{eq:reduction1d:superimposition:q2}, respectively. This way, as we will see in equation~\eqref{eq:reduction1d:superimposition:fields:g}, the transverse field in the gas only depends on three parameters. Moreover, using equations~\eqref{eq:reduction1d:electric_field:3d:1:adim} and~\eqref{eq:reduction1d:electric_field:3d:3:adim}, we can show that such basis functions satisfy the orthogonality Property~\ref{property:phi_othogonality} with respect to $\phi_0=1$:
	\begin{equation}\nonumber
		\int_{\partial\mathcal{D}}\phi_i\phi_0 = \int_{\partial\mathcal{D}}\phi_i = -\dfrac{\epsilon_0}{e}\int_{\partial\mathcal{D}}\mathbf{\Psi}_g^i\cdot\mathbf{n}_g = -\dfrac{\epsilon_0}{e}\int_{\mathcal{D}} \nabla\cdot\mathbf{\Psi}_g^i=0, \quad i=1,\,2.
	\end{equation}

	\noindent If the total electric field $\tilde{\mathbf{E}}$ satisfying~\eqref{eq:reduction1d:electric_field:3d:adim} is given by a linear combination of the solutions to problems~\eqref{eq:reduction1d:superimposition:x}-~\eqref{eq:reduction1d:superimposition:q0}:
	
	\begin{equation}
		\tilde{\mathbf{E}} = \sum_{k=1}^6 \tilde{a}_k \mathbf{\Psi}^k,
		\label{eq:reduction1d:superimposition:sum}
	\end{equation}
	
	\begin{figure}
		\centering
		\begin{tikzpicture}\scriptsize
			\draw (0,0) circle (1.8cm);
			\draw[fill=gray!30, dashed] (0,0) circle (.9cm) node at (.5,.5) {$\tilde{\mathcal{D}}_s$};
			\draw[fill=gray!60] (0,0) circle (.5cm) node at (.22,.22) {$\tilde{\mathcal{D}_g}$};
			\node at (.8,.8) {$\tilde{\mathcal{D}}$};
			
			\draw (0,0)--(-.9,0) node[below, pos=0.7] {$\tilde{R}_s$};
		\end{tikzpicture}
		\caption{Section of the adimensionalized domain $\tilde{\Omega}$, on which we define equation~\eqref{eq:reduction1d:electric_field:3d:adim}. The section $\tilde{\mathcal{D}}_g$ is a circle of radius 1, $\tilde{\mathcal{D}}$ is circle of radius $\tilde{R}_s$ and $\tilde{\mathcal{D}}_s = \tilde{\mathcal{D}}\setminus \tilde{\mathcal{D}}_g$ denotes the region between the two.}
		\label{figure:reduction1d:adimensional_sections}
	\end{figure}
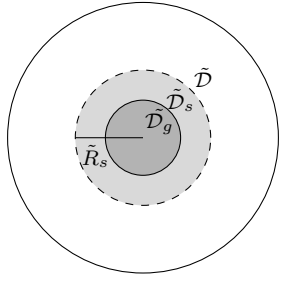
	
	\begin{figure}
		\centering
		\scriptsize
		\subfloat[Radial component \label{figure:reduction1d:soe:radial}]{
			\centering
			\begin{tikzpicture}
				\draw[fill=gray!30] (0,0) circle (1cm);
				\draw[->, thick] (0,0)--(1.2,0);
				\draw[->, thick] (0,0)--(-1.2,0);
				\draw[->, thick] (0,0)--(0,1.2);
				\draw[->, thick] (0,0)--(0,-1.2);
				\draw[->, thick] (0,0)--(.9,.9);
				\draw[->, thick] (0,0)--(-.9,.9);
				\draw[->, thick] (0,0)--(-.9,-.9);
				\draw[->, thick] (0,0)--(.9,-.9);
			\end{tikzpicture}}
			\hspace{1cm}
			\subfloat[$x$ component \label{figure:reduction1d:soe:x}]{
			\centering
			\begin{tikzpicture}
				\draw[fill=gray!30] (0,0) circle (1cm);
				\draw[->, thick] (-1.2,0)--(1.2,0);
				\draw[->, thick] (-1.2,.4)--(1.2,.4);
				\draw[->, thick] (-1.2,.8)--(1.2,.8);
				\draw[->, thick] (-1.2,-.4)--(1.2,-.4);
				\draw[->, thick] (-1.2,-.8)--(1.2,-.8);
			\end{tikzpicture}}
			\hspace{1cm}
			\subfloat[$y$ component \label{figure:reduction1d:soe:y}]{
			\centering
			\begin{tikzpicture}
				\draw[fill=gray!30] (0,0) circle (1cm);
				\draw[->, thick] (0,-1.2)--(0,1.2);
				\draw[->, thick] (.4,-1.2)--(.4,1.2);
				\draw[->, thick] (.8,-1.2)--(.8,1.2);
				\draw[->, thick] (-.4,-1.2)--(-.4,1.2);
				\draw[->, thick] (-.8,-1.2)--(-.8,1.2);
			\end{tikzpicture}}
			\caption{Different components of the electric field given by the considered effects: the radial component in Figure~\ref{figure:reduction1d:soe:radial} is given by the effect of the constant charge concentrations ($\mathbf{\Psi}_5$, $\mathbf{\Psi}_6$), while the components along the $x$ and $y$ axes are given by the respective components of the external electric field and of the dipole moment ($\mathbf{\Psi}_1$ and $\mathbf{\Psi}_3$ along $x$, and $\mathbf{\Psi}_2$ and $\mathbf{\Psi}_4$ along $y$).}
			\label{figure:reduction1d:soe}
	\end{figure}
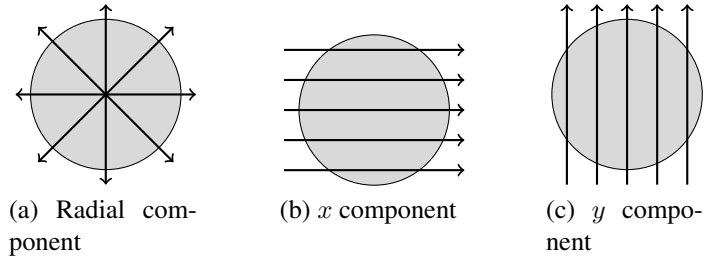

	\noindent then, we can compute the coefficients $\tilde{a}_k,\ k=1,\dots,6$, by substituting~\eqref{eq:reduction1d:superimposition:sum} in~\eqref{eq:reduction1d:electric_field:3d:adim}. Since the divergence of all the solutions but $\mathbf{\Psi}^5$ is zero in $\Omega_g$, the divergence of the electric field in the gas is given by
	
	\begin{equation*}
		\nabla_{\tilde{\mathcal{D}}}\cdot\tilde{\mathbf{E}}_g = \sum_{k=1}^6 \tilde{a}_k\nabla_{\tilde{\mathcal{D}}}\cdot\mathbf{\Psi}_g^k = \tilde{a}_5\nabla_{\tilde{\mathcal{D}}}\cdot\mathbf{\Psi}_g^5 = \tilde{a}_5.
	\end{equation*}

	\noindent Substituting it into equation~\eqref{eq:reduction1d:electric_field:3d:1:adim}, we obtain $ \tilde{a}_5 =  \tilde{q}\mathcal{K}. $\\
	Substitute now equation~\eqref{eq:reduction1d:superimposition:sum} into~\eqref{eq:reduction1d:electric_field:3d:4:adim}, where the only non-null contributions to the jump on the interface are given by $\mathbf{\Psi}^3$, $\mathbf{\Psi}^4$ and $\mathbf{\Psi}^6$:
	
	\begin{equation*}
		-\mathcal{K}_\Gamma \tilde{q}_\Gamma =
        -\mathcal{K}_\Gamma \sum_{j=0}^2 \tilde{q}_\Gamma^j\phi^j =
        -\tilde{a}_3\phi_1 -\tilde{a}_4\phi_2 -\tilde{a}_6\phi_0.
	\end{equation*}

	\noindent Then, $\tilde{a}_3 = \mathcal{K}_\Gamma\tilde{q}_\Gamma^1$, $\tilde{a}_4 = \mathcal{K}_\Gamma\tilde{q}_\Gamma^2$ and $\tilde{a}_6 = \mathcal{K}_\Gamma\tilde{q}_\Gamma^0$.\\
	Finally, the right--hand side of the Neumann condition~\eqref{eq:reduction1d:electric_field:3d:6:adim} is given by the contributions of the volume and surface charge distribution, described by $\mathbf{\Psi}^i,\ i=3,\dots,6$, and the action of an external, known and possibly non-null, electric field $\tilde{\mathbf{E}}^\text{ext}$:
	
	\begin{equation}
		\tilde{\mathbf{E}}^\text{ext}\cdot\mathbf{n} =
		% \tilde{E}^b - \sum_{k=3}^6 \tilde{a}_k\mathbf{\Psi}^k\cdot\mathbf{n} =
		% \tilde{\mathbf{E}}\cdot\mathbf{n} - \sum_{k=3}^6 \tilde{a}_k\mathbf{\Psi}^k\cdot\mathbf{n} =
		\sum_{k=1}^2 \tilde{a}_k\mathbf{\Psi}^k\cdot\mathbf{n}.
		\label{eq:reduction1d:superimposition:Eext}
	\end{equation}
    
	\noindent If we denote by $\tilde{E}^\text{ext}_x$ and $\tilde{E}^\text{ext}_y$ components of $\tilde{\mathbf{E}}^\text{ext}$ along the directions of the axes $x$ and $y$, respectively, then from equation~\eqref{eq:reduction1d:superimposition:Eext} we obtain $\tilde{a}_1 = \tilde{E}^\text{ext}_x$ and $\tilde{a}_2 = \tilde{E}^\text{ext}_y.$
	
	\noindent Finally, we can write the total electric field as follows:
	
	\begin{align*}
		\mathbf{E} = \dfrac{\hat{\Phi}}{R_g}\tilde{\mathbf{E}} &=
        \dfrac{\hat{\Phi}}{R_g}
        \left(
        \tilde{E}^\text{ext}_x\mathbf{\Psi}^1 +
        \tilde{E}^\text{ext}_y\mathbf{\Psi}^2 +
        \mathcal{K}_\Gamma  \tilde{q}_\Gamma^1\mathbf{\Psi}^3 +
        \mathcal{K}_\Gamma  \tilde{q}_\Gamma^2\mathbf{\Psi}^4 +
        \mathcal{K} \tilde{q}\mathbf{\Psi}^5 +
        \mathcal{K}_\Gamma  \tilde{q}_\Gamma^0\mathbf{\Psi}^6
        \right) = \\
        &=
        E^\text{ext}_x\mathbf{\Psi}^1 +
        E^\text{ext}_y\mathbf{\Psi}^2 +
        \frac{{q}_\Gamma^1}{\epsilon_0}\mathbf{\Psi}^3 +
        \frac{{q}_\Gamma^2}{\epsilon_0}\mathbf{\Psi}^4 +
        R_g \frac{q}{\epsilon_0}\mathbf{\Psi}^5 +
        \frac{q_\Gamma^0}{\epsilon_0}\mathbf{\Psi}^6
	\end{align*}
	
	\noindent In particular, noticing that $\mathbf{\Psi}_g^6=0$ and that on $\mathbf{E}_\Gamma$ the only effects are produced by the external field and the surface charges $q_\Gamma^1$ and $q_\Gamma^2$, we obtain:
	
    \begin{subnumcases}{}
        \mathbf{E}_g = 
        a_1 \mathbf{\Psi}_g^1+ 
        a_2 \mathbf{\Psi}_g^2+
        a_5 \mathbf{\Psi}_g^5,
		\label{eq:reduction1d:superimposition:fields:g}\\
        {\mathbf{E}}_s = 
        a_1 \mathbf{\Psi}_s^1+ 
        a_2 \mathbf{\Psi}_s^2+
        a_5 \mathbf{\Psi}_s^5+
        a_6 \mathbf{\Psi}_s^6,
		\label{eq:reduction1d:superimposition:fields:s}\\
        {\mathbf{E}}_\Gamma = 
        a_1 \mathbf{\Psi}^1\big|_\Gamma+
        a_2 \mathbf{\Psi}^2\big|_\Gamma.
        \label{eq:reduction1d:superimposition:fields:h}
    \end{subnumcases}

    \noindent where $a_1 = E^\text{ext}_x + \dfrac{q_\Gamma^1}{\epsilon_0}$, $a_2 = E^\text{ext}_y + \dfrac{q_\Gamma^2}{\epsilon_0}$, $a_5 = R_g \dfrac{q}{\epsilon_0} = R_g \dfrac{e}{\epsilon_0}\sum_{p=1}^3\omega_p \bar{c}_p$ and $a_6=\dfrac{q_\Gamma^0}{\epsilon_0}$.

    \noindent We can compute the exact solutions to problems~\eqref{eq:reduction1d:superimposition:x} -~\eqref{eq:reduction1d:superimposition:q0}. In particular, $\bm{\Psi}_1$ and $\bm{\Psi}_2$ are directed as the axes $x$ and $y$, respectively, since their only source is the external electric field in the corresponding direction; moreover, from the boundary condition~\eqref{eq:reduction1d:electric_field:3d:6:adim} and from the interface condition~\eqref{eq:reduction1d:electric_field:3d:5:adim} we can deduce the magnitude of the solution in the two domains, obtaining

    \begin{equation}
        \begin{alignedat}{4}
            &\bm{\Psi}_g^1 &= \epsilon_s\mathbf{x}, \qquad &\bm{\Psi}_s^1 &= \mathbf{x}, \\
            &\bm{\Psi}_g^2 &= \epsilon_s\mathbf{y}, \qquad &\bm{\Psi}_s^2 &=\mathbf{y}.
        \end{alignedat}
        \label{eq:reduction1d:exact_psi1-2}
    \end{equation}

    \noindent By integrating equation~\eqref{eq:reduction1d:electric_field:3d:1:adim} on a circle $\mathcal{W}_g$ contained in $\tilde{\mathcal{D}}$ and then equation~\eqref{eq:reduction1d:electric_field:3d:2:adim} on a circle $\mathcal{W}_s$, such that $\tilde{\mathcal{D}} \subset \mathcal{W}_s \subset \tilde{\mathcal{D}}_s$, with the data defined in~\eqref{eq:reduction1d:superimposition:q}, we obtain the exact solution

    \begin{equation*}
        \bm{\Psi}_g^5 = \dfrac{\tilde{r}}{2}\mathbf{n}, \qquad \bm{\Psi}_s^5 = \dfrac{1}{2\tilde{r}\epsilon_s}\mathbf{n},
    \end{equation*}

    \noindent directed as the outgoing normal vector $\mathbf{n}$ to the lateral surface of the circles, and only dependent on the radial coordinate $\tilde{r}$. Finally, integrating equation~\eqref{eq:reduction1d:electric_field:3d:2:adim} with the data defined in~\eqref{eq:reduction1d:superimposition:q0} on the circle $\mathcal{W}_s$, we can compute the exact solution $\bm{\Psi}_s^6$ in $\tilde{\mathcal{D}}_s$, while the corresponding $\bm{\Psi}_g^6$ in $\tilde{\mathcal{D}}$ is given by the interface condition~\eqref{eq:reduction1d:electric_field:3d:5:adim}:

    \begin{equation*}
        \bm{\Psi}_g^6 = 0, \qquad \bm{\Psi}_s^6 = \dfrac{1}{\tilde{r}\epsilon_s}\mathbf{n}.
    \end{equation*}

    \noindent For further details about the computation of the analytical solutions, see Appendix~\ref{appendix:exact_sol_superimposition_of_effects}.

    \subsection{Properties of the coefficients F, G and H}
    \label{section:reduction1d:FGH}

    Given the exact solutions to problems~\eqref{eq:reduction1d:superimposition:x} and~\eqref{eq:reduction1d:superimposition:y}, we know the analytical expression of the basis functions $\phi_1$ and $\phi_2$, and we can easily compute their gradients. 

    \begin{equation*}
        \begin{alignedat}{4}
            &\phi_1 = -\epsilon_s\mathbf{x}\cdot\mathbf{n}_g =
            %-\epsilon_s\cos(\theta) =
            -\epsilon_s\dfrac{x}{R_g}, &\qquad
            &\phi_2 = -\epsilon_s\mathbf{y}\cdot\mathbf{n}_g =
            %-\epsilon_s\sin(\theta) =
            -\epsilon_s\dfrac{y}{R_g}, \\
            &\nabla\phi_1
            %= \dfrac{\epsilon_s}{R_g}\left[ -\sin^2(\theta), \sin(\theta)\cos(\theta), 0 \right]^T
            = \dfrac{\epsilon_s}{R_g}\left[ 1, 0, 0 \right]^T, &\qquad
            &\nabla\phi_1 
            %= \dfrac{\epsilon_s}{R_g}\left[\sin(\theta)\cos(\theta), -\sin^2(\theta)
            = \dfrac{\epsilon_s}{R_g}\left[0, 1, 0 \right]^T,
        \end{alignedat}
    \end{equation*}

    \noindent %where $\theta$ denotes the angular coordinate in a cylindrical reference system, and $R_g$ is the radius of $\Omega_g$.
    Then, substituting these gradients and the analytical expressions of $\phi_1$ and $\phi_2$ in equation~\eqref{eq:reduction1d:def:H}, we obtain explicit expressions of the coefficients $H_1$, $H_2$ and $M_{ij},\ i,j=1,2$:

    \begin{subequations}
        \begin{align}
            &H_1 = \int_{\partial\mathcal{D}} \left(a_1\epsilon_s\mathbf{x} + a_2\epsilon_s\mathbf{y} \right)\cdot \nabla\phi_1\big|_\Gamma = -\epsilon_s^2 \pi a_1, 
            \label{eq:reduction1d:H1:expression}\\
            &H_2 = \int_{\partial\mathcal{D}} \left( a_1\epsilon_s\mathbf{x} + a_2\epsilon_s\mathbf{y} \right)\cdot \nabla\phi_2\big|_\Gamma = -\epsilon_s^2 \pi a_2,
            \label{eq:reduction1d:H2:expression}\\
            &M_{11} = M_{22} = \epsilon_s^2 R_g \pi, \quad
            M_{12} = M_{21} = 0.
            \label{eq:reduction1d:M:expression}
        \end{align}
        \label{eq:reduction1d:HM:expression}
    \end{subequations}
    
    \noindent Moreover, in Appendices~\ref{appendix:F} and~\ref{appendix:G} we compute the values of the coefficients $F_p^+,\,F_p^-,\,G_{p,i}^+$ and $G_{p,i}^-$ for $p=1,2,3$ and $i=1,2$, as functions of $a_1,\,a_2$ and $a_5$ (equations~\eqref{eq:reduction1d:F:appendix:expression} and~\eqref{eq:reduction1d:G:expression}), from which we can deduce the following result, as discussed in Remarks~\ref{remark:appendix:F:sign} and~\ref{remark:appendix:G:sign} of Appendix~\ref{appendix:FG}:

    \begin{lemma}
        The following inequalities hold:
        \begin{enumerate}
            \item $F_p^+ \geq 0,\,$ for $ p=1,\,2,\,3$;
            \item $F_p^- \leq 0,\,$ for $ p=1,\,2,\,3$;
            \item $G_{p,i}^+ \geq 0$ if $a_i \geq 0,\,$ for $p=1,\,2$ and $i=1,\,2$;
            \item $G_{3,i}^+ \leq 0$ if $a_i \geq 0,\,$ for $i=1,\,2$;
            \item $G_{p,i}^- \geq 0$ if $a_i \geq 0,\,$ for $p=1,\,2$ and $i=1,\,2$;
            \item $G_{3,i}^- \leq 0$ if $a_i \geq 0,\,$ for $i=1,\,2$.
        \end{enumerate}
        \label{theorem:FG:sign}
    \end{lemma}

    \noindent Finally, if we substitute the values of $H_i$ and $M_{ij}$, $i,j=1,2$, and $q_\Gamma^i = a_i + E_i^\text{ext}$ explicitly obtained in equation~\eqref{eq:reduction1d:HM:expression}, into equation~\eqref{eq:reduction1d:charge:1d:4}, we obtain the following:

        \begin{equation}
            \dfrac{\partial a_i}{\partial t} =
            - \dfrac{\bar{\sigma}_\Gamma}{|\partial\mathcal{D}|}\dfrac{a_i}{\epsilon_0 R_g}
            + \dfrac{1}{\epsilon_s^3 R_g \pi} \dfrac{\partial E_i^\text{ext}}{\partial t} 
            + \sum_{p=1}^3 \dfrac{e}{\epsilon_0 \epsilon_s^2 \pi R_g} \dfrac{\bar{\mu}_p}{|\mathcal{D}|}\left( \bar{c}_p G_{p,i}^+ 
            + \bar{c}_p^b G_{p,i}^- \right), \quad i = 1,2,
            \label{eq:reduction1d:dipole:a}
        \end{equation}

        \noindent where $E_1^\text{ext} = E_x^\text{ext}$ and $E_2^\text{ext} = E_y^\text{ext}$.

        We discuss now the orders of magnitude of the terms on the right--hand side of equation~\eqref{eq:reduction1d:dipole:a}, so that we can keep into account only the predominant ones.
        The effect of the charge density, is almost null on the external electric field at a long distance from the gas domain. In fact, the variation of the external electric field is almost completely governed by the boundary conditions for the electrostatic problem, imposed on the dielectric as external boundary condition. In applications of our interest, i.e., AC regime at industrial frequency, the externally imposed electric field is given by a sinusoidal function, whose characteristic time of oscillation is typically of the order of milliseconds ($10^{-3} s$). On the other hand, the evolution in time of $a_i$ is much faster, and this allows to neglect the time derivatives of $E^\text{ext}_i,\, i=1,\,2,$ in equation~\eqref{eq:reduction1d:dipole:a}. Indeed, as detailed in Appendix~\ref{appendix:reduction1d:characteristic_time}, the right--hand side of equation~\eqref{eq:reduction1d:dipole:a} has a characteristic time of order $10^{-10} s$, thus the dynamics of the dipole moment is much faster than that of the external electric field, and therefore we can neglect the time derivative of $E_i^\text{ext},\, i=1,\,2$.

        \section{Properties of the 1D reduced problem}
        \label{section:reduction1d:properties_1D_reduced_problem}
        After the dimensional reduction carried out in Section~\ref{section:reduction1d:modelReduction} and the considerations on the transversal components of the electric field made in Section~\ref{section:reduction1d:electricField}, we end up with a problem made of the reduced equations for volume and surface charge concentrations~\eqref{eq:reduction1d:charge:1d} and the equations~\eqref{eq:reduction1d:dipole:a} in the unknowns $a_i,\, i=1,2$, which to represent two orthogonal components of the dipole moment and of the external electric field in the computation of the integrals of the transversal components of the electric field:

        \begin{subnumcases}{\label{eq:reduction1d:charge:1d:final}}
		\label{eq:reduction1d:charge:1d:final:1}
        \begin{split}
            \dfrac{\partial \bar{c}_p}{\partial t} +
            \dfrac{\partial}{\partial s} & \left(\omega_p\bar{\mu}_p\bar{c}_pE_\Lambda\right) +
            \bar{\mu}_p \dfrac{\bar{c}_p}{|\mathcal{D}|} F_p^+ +
            \bar{\mu}_p \dfrac{\bar{c}_p^b}{|\mathcal{D}|} F_p^- 
            - \dfrac{\partial}{\partial s}\left(\bar{\nu}_p \dfrac{\partial \bar{c}_p}{\partial s}\right) +\\
            & +\dfrac{|\partial\mathcal{D}|}{|\mathcal{D}|} K_p\bar{c}_p = |\mathcal{D}|\bar{C}_p,\  p=1,2,3,
        \end{split}
         & $\text{on}\ \Lambda,$\\
		\label{eq:reduction1d:charge:1d:final:2}
        \dfrac{\partial \bar{c}_{\Gamma,e}}{\partial t}  -
        \dfrac{\partial}{\partial s}\left(\bar{c}_{\Gamma,e}\bar{\mu}_{\Gamma,e} E_\Lambda\right) =
        \sum_{p=1}^2\dfrac{\bar{\mu}_p\bar{c}_p}{|\mathcal{D}|}F_p^+ -
        \dfrac{\bar{\mu}_3 \bar{c}_3^b}{|\mathcal{D}|}F_3^- +
        \sum_{p=1}^2 K_p\bar{c}_p\dfrac{|\partial\mathcal{D}|}{|\mathcal{D}|}, & $\text{on}\ \Lambda,$ \\
		\label{eq:reduction1d:charge:1d:final:3}
        \dfrac{\partial \bar{c}_{\Gamma,h}}{\partial t}  +
        \dfrac{\partial}{\partial s}\left(\bar{c}_{\Gamma,h}\bar{\mu}_{\Gamma,h} E_\Lambda\right) =
        \dfrac{\bar{\mu}_3\bar{c}_3}{|\mathcal{D}|}F_3^+ -
        \sum_{p=1}^2 \dfrac{\mu_p\bar{c}_p^b}{|\mathcal{D}|}F_p^- +
        % \dfrac{\bar{\mu}_3 \bar{c}_3^b}{|\mathcal{D}|}F_3^- +
        K_3\bar{c}_3\dfrac{|\partial\mathcal{D}|}{|\mathcal{D}|}, & $\text{on}\ \Lambda,$ \\
		\label{eq:reduction1d:charge:1d:final:4}
        \dfrac{\partial a_i}{\partial t} =
        - \dfrac{\bar{\sigma}_\Gamma}{|\partial\mathcal{D}|}\dfrac{a_i}{\epsilon_0 R_g}
        + \sum_{p=1}^3 \dfrac{e}{\epsilon_0\epsilon_s^2 \pi R_g} \dfrac{\bar{\mu}_p}{|\mathcal{D}|}\left( \bar{c}_p G_{p,i}^+ 
        + \bar{c}_p^b G_{p,i}^- \right), \quad i = 1,2, & $\mathrm{on} \ \Lambda,$
    \end{subnumcases}

        \noindent where $F_p^+,\,F_p^-,\,G_{p,i}^+$ and $G_{p,i}^-$, $p=1,2,3$ are given by equations~\eqref{eq:reduction1d:F:appendix:expression} and~\eqref{eq:reduction1d:G:expression} in Appendix~\ref{appendix:FG}.\\
        We can show that the solution $a_i(t,s),\, i=1,\,2,$ to equation~\eqref{eq:reduction1d:charge:1d:final:4} has the same sign as the initial condition $a_i(0,s),\, i=1,\,2,$ at every time instant $t>0$, for all $s\in[0,\,S],$ and that the charge concentrations, solutions to equations~\eqref{eq:reduction1d:charge:1d:final:1} -~\eqref{eq:reduction1d:charge:1d:final:3}, are non-negative.\\
        In order to do so, let us introduce \textit{production--destruction} ODEs as problems of the form

        \begin{equation*}
            \begin{cases}
                \dfrac{du}{dt} = P(u) - D(u), & t\in(0,T],\\
                u(0) = u_0,
            \end{cases}
        \end{equation*}

        \noindent with $P(u) \geq 0$ and $D(u) \geq 0$ if $u \geq 0$. We call $P$ \textit{production} term and $D$ \textit{destruction} term (see~\cite{burchard2003high} for details). The solution to a production--destruction ODE is non-negative if the initial condition $u_0$ is non-negative and the destruction term statisfies the following property:
        
        \begin{property}
            $D(u)\to 0$ as $u\to 0$.
            \label{property:D:limit}
        \end{property}

        \begin{theorem}
        \label{theorem:ai:positive}
            Equation~\eqref{eq:reduction1d:charge:1d:final:4} is monotone, i.e., its solution $a_i(t)$ has the same sign as the initial condition $a_i(0) = E_i^\text{ext}(0) + \dfrac{q_\Gamma^i(0)}{\epsilon_0}$ for all $t>0$.
        \end{theorem}

        \begin{proof}
        	From Theorem~\ref{theorem:FG:sign} we deduce that equation~\eqref{eq:reduction1d:dipole:a} is a production--destruction ODE, with production and destruction terms respectively given by:
            \begin{subequations}
                \begin{align}
                    P(a_i) &= 
                        \sum_{i=1}^2 \dfrac{e}{\epsilon_0\epsilon_s^2 \pi R_g} \dfrac{\bar{\mu}_p}{|\mathcal{D}|} \left(\bar{c}_p G_{p,i}^+ + \bar{c}_p^b G_{p,i}^- \right),
                    \label{eq:reduction1d:production}\\
                    D(a_i) &=
                        \dfrac{\bar{\sigma}_\Gamma}{|\partial\mathcal{D}|}\epsilon_s^2\pi a_i -
                        \dfrac{e}{\epsilon_0\epsilon_s^2 \pi R_g} \dfrac{\bar{\mu}_3}{|\mathcal{D}|} \left(\bar{c}_3 G_{3,i}^+ + \bar{c}_3^b G_{3,i}^- \right),
                    \label{eq:reduction1d:destruction}
                \end{align}
            \end{subequations}
            with $P(a_i) \geq 0$ if $a_i \geq 0$ and $D(a_i) \geq 0$ if $a_i \geq 0$, $i=1,2$.\\
            Moreover, in Appendix~\ref{appendix:F:proof:bounded}, we show that the destruction term $D$ defined as in equations~\eqref{eq:reduction1d:destruction} satisfies Property~\ref{property:D:limit}.\\
            As a consequence, the solution $a_i(t,s)$ to equation~\eqref{eq:reduction1d:charge:1d:final:4} has the same sign as the initial condition $a_i(0,s)=\dfrac{q_\Gamma^i(0)}{\epsilon_0} + E_i^\text{ext}(0)$, for all $t>0$ and $s\in[0,\,S]$.
        \end{proof}

    \noindent Equation~\eqref{eq:reduction1d:charge:1d:final:1}, instead, represents a weakly-coupled nonlinear system, made of three equations corresponding to $p=1,\,2,\,3$, where the nonlinearity is in the coupling terms $F_p^+$ and $F_p^-$, $p=1,2,3$. As anticipated in Section~\ref{section:reduction1d:3dProblem}, the diffusion coefficients are zero for $p=2,\,3$ and non-zero for $p=1$. Thus, equation~\eqref{eq:reduction1d:charge:1d:final:1} for $p=1$ is a parabolic advection-diffusion-reaction equation, while for $p=2,\,3$ it is an advection-reaction equation. Also for this system of equations we can prove a monotonicity result:

    \begin{theorem}
    \label{theorem:cp:positive}
        If $\bar{c}_p^b\geq 0$ and $ \bar{c}_p(0)\geq 0,\, \forall p=1,\,2,\,3,$ then the solution to equations ~\eqref{eq:reduction1d:charge:1d:final:1}-\eqref{eq:reduction1d:charge:1d:final:3} is non-negative for all $t>0$.
    \end{theorem}
    
    \begin{proof}
    We can rewrite equation~\eqref{eq:reduction1d:charge:1d:final:1} as follows:

    \begin{equation*}
        \dfrac{\partial \bar{c}_p}{\partial t}
        - \mathcal{L}_p(\bar{c}_p) 
        + r_p(\bar{\mathbf{c}})
        = g_p(\mathbf{\bar{c}}),
    \end{equation*}

    \noindent where $\mathbf{\bar{c}} = (\bar{c}_1,\,\bar{c}_2,\,\bar{c}_3)$ denotes the collection of all the volume charge concentrations, $\mathcal{L}_p,\, p=1,\,2,\,3,$ the differential operator

    \begin{equation*}
        \mathcal{L}_p(\bar{c}_p) =
        \dfrac{\partial}{\partial s}\left( \bar{\nu}_p \dfrac{\partial \bar{c}_p}{\partial s} \right)
        -\dfrac{\partial}{\partial s}\left( \omega_p \bar{\mu}_p \bar{E}_\Lambda \bar{c}_p \right),
        \quad p=1,\,2,\,3,
    \end{equation*}
    
    \noindent while $r_p (\mathbf{\bar{c}})$ and $g_p(\bar{\mathbf{c}}),\,p=1,\,2,\,3$, represent the linear and nonlinear reaction coupling terms, respectively:

    \begin{equation*}
        r_p(\mathbf{\bar{c}}) =
        \dfrac{|\partial\mathcal{D}|}{|\mathcal{D}|} K_p\bar{c}_p 
        - |\mathcal{D}|\bar{C}_p,
        \quad 
        g_p(\mathbf{\bar{c}}) =
        - \bar{c}_p \mathcal{N}_p^+(\mathbf{\bar{c}}) 
        - \mathcal{N}_p^-(\mathbf{\bar{c}}),\quad p=1,\,2,\,3.
    \end{equation*}

    \noindent where we have defined

    \begin{equation*}
        \mathcal{N}_p^+(\mathbf{\bar{c}}) := \bar{\mu}_p \dfrac{1}{|\mathcal{D}|} F_p^+(\mathbf{\bar{c}}), \qquad
        \mathcal{N}_p^-(\mathbf{\bar{c}}) := \bar{\mu}_p \dfrac{\bar{c}_p^b}{|\mathcal{D}|} F_p^-(\mathbf{\bar{c}}), \quad p=1,\,2,\,3.
        \label{eq:reduction1d:Np:def}
    \end{equation*}

    \noindent Our proof for $p=1$ relies on~\cite[Lemma 2.1]{pao2012nonlienar}, that we report here for the reader's convenience:

    \begin{lemma}~\cite[Lemma 2.1]{pao2012nonlienar}\\
    \label{theorem:pao}
        Let $\Omega$ be an open domain in $\mathbb{R}^n,\, n\in\mathbb{N}$, and let $\partial\Omega$ be the boundary of $\Omega$. For each $T>0$, let $D_T = (0,T]\times \Omega$ and $S_T = (0,T] \times \partial\Omega$. Denote by $\mathcal{C}^m(\bar{D}_T)$ the set of continuous functions on $\bar{D}_T$ and by $\mathcal{C}^{1,2}(D_T)$ the set of continuously differentiable functions in $t$ and twice continuously differentiable in $x$, for all $(x,t)\in D_T$. Let $w\in\mathcal{C}\left(\bar{D}_T\right) \cap \mathcal{C}^{1,2}(D_T)$ be such that
        \begin{equation*}
            \begin{dcases}
                w_t - Lw + cw \geq 0, & \text{in}\ D_T,\\
                \alpha_0 \dfrac{\partial w}{\partial \nu} + \beta_0 w \geq 0, & \text{on}\ S_T,\\
                w(0,x) \geq 0, & \text{in}\ \Omega,
            \end{dcases}
        \end{equation*}

    \noindent where $\alpha_0 \geq 0,\, \beta_0 \geq 0,\, \alpha_0+\beta_0 > 0$ on $S_T$, and $c\equiv c(t,x)$ is a bounded function in $D_T$. Then, $w(t,x) \geq 0$ in $\bar{D}_T$. Moreover, $w(t,x) > 0$ in $D_T$ unless it is identically zero.
    \end{lemma}

    \noindent As a consequence  of Lemma~\ref{theorem:FG:sign}, $\mathcal{N}_p^+ \geq 0$ and $\mathcal{N}_p^- \leq 0,\, \forall \mathbf{\bar{c}}$. Thus, if we suppose that the solution to equation~\eqref{eq:reduction1d:charge:1d:final:1} for $p=1$ is $\bar{c}_1 < 0$, we would have $g_1(\mathbf{\bar{c}}) \geq 0$. However, by Lemma~\ref{theorem:pao}, this implies that $\bar{c}_1 \geq 0$, which contradicts the previous assumptions. Therefore, if equation~\eqref{eq:reduction1d:charge:1d:final:1} for $p=1$ admits a continuous solution, it must be non-negative.
            
        Let us consider now the advection-reaction equation~\eqref{eq:reduction1d:charge:1d:final:1} for $p=2,\,3$, whose solution can be computed with the method of characteristics. Along the characteristic curves $x(x_0,v,t)$, where $v = \omega_p \bar{\mu}_p E_\Lambda$ denotes the transport velocity, $\bar{c}_p$ is given by the solution to the following ODE (see~\cite[Section 3.2]{evans2022partial} for details):

        \begin{equation}
            \begin{cases}
                \dfrac{d \bar{c}_p}{d t} = f_p({\mathbf{\bar{c}}}) - \dfrac{d E_\Lambda}{ds}{\bar{c}}_p, \quad t > 0,\\
                \bar{c}_p(t=0) = \bar{c}_p^0,
            \end{cases}
            \label{eq:reduction1d:cp:ode}
        \end{equation}

		\noindent where $E_\Lambda$ is a given tangential electric field on $\Lambda$, assumed in this work independent of $\bar{\mathbf{c}}$, and we have defined
		
		\begin{equation*}
			f(\mathbf{\bar{c}}) = -\bar{c}_p \mathcal{N}_p^+(\mathbf{\bar{c}}) - \mathcal{N}_p^-(\mathbf{\bar{c}}) - \dfrac{|\partial\mathcal{D}|}{|\mathcal{D}|}K_p \bar{c}_p + |\mathcal{D}|\bar{C}_p.
		\end{equation*}

        \noindent The ODE~\eqref{eq:reduction1d:cp:ode} can be seen as a production-destruction equation, with production and destruction terms respectively given by:

        \begin{alignat*}{2}
            P_p(c_p) &=
            \begin{dcases}
                -\mathcal{N}_p^-({\mathbf{\bar{c}}}) + |\mathcal{D}|\bar{C}_p, & \text{if\ } \dfrac{d E_\Lambda}{ds} \geq 0,\\
                -\mathcal{N}_p^-({\mathbf{\bar{c}}}) + |\mathcal{D}|\bar{C}_p - \dfrac{d E_\Lambda}{ds}\bar{c}_p, & \text{if\ } \dfrac{d E_\Lambda}{ds} \leq 0,
            \end{dcases}&
            \quad p=2,\,3 \\
            D_p(c_p) &=
            \begin{dcases}
                \dfrac{|\partial\mathcal{D}|}{|\mathcal{D}|}K_p\bar{c}_p + \bar{c}_p \mathcal{N}_p^+({\mathbf{\bar{c}}}) - \dfrac{d E_\Lambda}{ds}\bar{c}_p, & \text{if\ } \dfrac{d E_\Lambda}{ds} \geq 0,\\
                \dfrac{|\partial\mathcal{D}|}{|\mathcal{D}|}K_p\bar{c}_p + \bar{c}_p \mathcal{N}_p^+({\mathbf{\bar{c}}}), & \text{if\ } \dfrac{d E_\Lambda}{ds} \leq 0,
            \end{dcases}&
            \quad p=2,\,3.
        \end{alignat*}

		\noindent In~\ref{appendix:F:proof:bounded} we show that the following property holds:

		\begin{property}
			For all $p=1,\,2,\,3,$ $\exists\ \hat{F}_p^+\geq 0$ and $\tilde{F}_p^+ \geq 0$ such that $F_p^+ \leq \hat{F}_p^+ \sum_{q=1}^3 |\bar{c}_q| + \tilde{F}_p^+$.
			\label{property:reduction1d:F:bounded}
		\end{property}
		
		\noindent Thus, for $p=2,\,3,$

        \begin{equation*}
            \lim_{\bar{c}_p \to 0} F_p^+({\mathbf{c}})\bar{c}_p \leq
            \lim_{\bar{c}_p \to 0} \left( \hat{F}_p^+\sum_{q=1}^3 |\bar{c}_q| + \tilde{F}_p^+ \right) \bar{c}_p = 0, \quad p=2,\,3,
        \end{equation*}

        \noindent since $\hat{F}_p^+$, $\tilde{F}_p^-$ and $\sum_{q=1}^3|\bar{c}_q|$ are bounded for $\bar{c}_p \to 0$. As a consequence, $D_p$ satisfies Property~\ref{property:D:limit}, hence, the solution to~\eqref{eq:reduction1d:cp:ode} is non-negative, and thus $\bar{c}_p\geq 0$ along all the characteristics, for $p=1,\,2,\,3$.
    \end{proof}

    \noindent A similar result can be also proven for the linear transport equations~\eqref{eq:reduction1d:charge:1d:final:2} and~\eqref{eq:reduction1d:charge:1d:final:3}:

    \begin{theorem}
        If $\bar{c}_{\Gamma,e}(0,s) \geq 0$ and $\bar{c}_{\Gamma,h}(0,s) \geq 0$, then the solutions $\bar{c}_{\Gamma,e}$ and $\bar{c}_{\Gamma,h}$ to equations~\eqref{eq:reduction1d:charge:1d:final:2} and~\eqref{eq:reduction1d:charge:1d:final:3} are non--negative for all $t>0$.
        \label{theorem:cgamma:positive}
    \end{theorem}

    \begin{proof}
        The solutions to equations~\eqref{eq:reduction1d:charge:1d:final:2} and~\eqref{eq:reduction1d:charge:1d:final:3} along the characteristic curves $x(x_0, v_{\Gamma,\star}, t)$, $\star = e,\,h$, with $v_{\Gamma,e} = -\bar{\mu}_{\Gamma,e}E_\Lambda$, $v_{\Gamma,h} = \bar{\mu}_{\Gamma,h}E_\Lambda$, are given by the solutions to the following ODEs (see~\cite[Section 3.2]{evans2022partial} for details):
        
        \begin{equation*}
            \begin{cases}
                \dfrac{dw_{\Gamma,\star}}{dt}(t) = f_{\Gamma,\star}(\mathbf{\bar{c}}), \quad t\geq 0,\\
                w_{\Gamma,\star}(0) = \bar{c}_{\Gamma,\star}(x_0,0),     
            \end{cases}
            \quad \star = e,h,
        \end{equation*}

        \noindent where we have defined 
            
        \begin{align*}
            f_{\Gamma,e} & =
            \sum_{p=1}^2 \mathcal{N}_p^+(\bar{\mathbf{c}})\bar{c}_p 
            - \mathcal{N}_3^-(\bar{\mathbf{c}}) +
            \sum_{p=1}^2 K_p\bar{c}_p \dfrac{|\partial\mathcal{D}|}{|\mathcal{D}|}, \\
            f_{\Gamma,h} & =
            \mathcal{N}_3^+(\bar{\mathbf{c}})\bar{c}_3 - 
            \sum_{p=1}^2 \mathcal{N}_p^-(\bar{\mathbf{c}})+
            K_3\bar{c}_3 \dfrac{|\partial\mathcal{D}|}{|\mathcal{D}|}.
        \end{align*}

        \noindent As a consequence of Lemma~\ref{theorem:FG:sign} and Theorem~\ref{theorem:cp:positive}, $f_{\Gamma,\star} \geq 0,$ $\star=e,\,h$. Thus, $w(x,t) = \bar{c}_{\Gamma,\star}(x,0) + \displaystyle\int_0^t f_{\Gamma,\star}(x,t) dt \geq 0$ for all $x\in\Lambda, \,t>0$, and the solutions to equations~\eqref{eq:reduction1d:charge:1d:final:2} and~\eqref{eq:reduction1d:charge:1d:final:3} are non--negative along all the characteristics.
    \end{proof}

    \noindent Finally, we can show that the reduced 1D system~\eqref{eq:reduction1d:charge:1d:final} is conservative, i.e., that the total chemical reactions among the charged species sum to zero pointwise in the whole gas domain. We show that the reduced 1D system remains conservative, despite of the presence of the nonlinear reaction terms $F_p^\pm$ that did not appear in the original 3D model.

    \begin{theorem}
        The system of equations~\eqref{eq:reduction1d:charge:1d:final:1}-~\eqref{eq:reduction1d:charge:1d:final:3}, modeling the charge densities in 1D, conserves the total net charge.
        \label{theorem:charge:conservative}
    \end{theorem}

    \begin{proof}
        If we multiply equations~\eqref{eq:reduction1d:charge:1d:final:1}-~\eqref{eq:reduction1d:charge:1d:final:3} by the electron charge $e$ and by the unit charge of each species, and we sum them all together, we obtain an equation describing the total net charge $\bar{c} = e \left(\sum_{p=1}^3 \omega_p \bar{c}_p + \bar{c}_{\Gamma,h} - \bar{c}_{\Gamma,e} \right)$ in the 1D domain:

        \begin{multline*}
            \dfrac{\partial \bar{c}}{\partial t} +
            \dfrac{\partial}{\partial s} \left( e E_\Lambda \sum_{p=1}^3 \left(\bar{\mu}_p\bar{c}_p\right) + E_\Lambda\bar{\sigma}_\Gamma \right) -
            \dfrac{\partial }{\partial s}\left( e\sum_{p=1}^3 \omega_p \bar{\nu}_p \dfrac{\partial\bar{c}_p}{\partial s} \right) = + |\mathcal{D}| e\sum_{p=1}^3 \omega_p\bar{C}_p +\\
            - e \sum_{p=1}^3 \left(\omega_p \bar{\mu}_p\dfrac{\bar{c}_p}{|\mathcal{D}|}F_p^+ + \omega_p \bar{\mu}_p \dfrac{\bar{c}_p^b}{|\mathcal{D}|}F_p^- + \dfrac{|\partial \mathcal{D}|}{|\mathcal{D}|}\omega_p K_p \bar{c}_p \right) %+ |\mathcal{D}| e\sum_{p=1}^3 \omega_p\bar{C}_p 
            e\bar{\mu}_3 \dfrac{\bar{c}_3}{|\mathcal{D}|}F_3^+
            +\bar{\mu}_3 \dfrac{\bar{c}_3^b}{|\mathcal{D}|}F_3^-  + \\
            + \dfrac{|\partial \mathcal{D}|}{|\mathcal{D}|} K_3 \bar{c}_3
            + e \sum_{p=1}^2 \left( \bar{\mu}_p\dfrac{\bar{c}_p}{|\mathcal{D}|}F_p^+ +  \bar{\mu}_p \dfrac{\bar{c}_p^b}{|\mathcal{D}|}F_p^- + \dfrac{|\partial \mathcal{D}|}{|\mathcal{D}|} K_p \bar{c}_p \right) + 
            %e\bar{\mu}_3 \dfrac{\bar{c}_3}{|\mathcal{D}|}F_3^+ +  \bar{\mu}_3 \dfrac{\bar{c}_3^b}{|\mathcal{D}|}F_3^- + \dfrac{|\partial \mathcal{D}|}{|\mathcal{D}|} K_3 \bar{c}_3 %= \\
            = |\mathcal{D}| e \sum_{p=1}^3 \omega_p \bar{C}_p.
        \end{multline*}      
        
        \noindent And, by the definition of the chemical reaction term~\eqref{eq:chemistry:expression}, we have
        
        \begin{equation*}
        	|\mathcal{D}| e \sum_{p=1}^3 \omega_p \bar{C}_p =
        	|\mathcal{D}| e \left(-(\alpha -\eta) - \alpha + \eta\right) \mu_1 |\mathbf{E}_g|\bar{c}_1 = 0.
        \end{equation*}
        
        \noindent Thus, we can conclude that the net effect of all the chemical reactions is null on the total charge and the fluxes are balanced by the time derivative.
        
    \end{proof}
    
	\section{Numerical methods}
	\label{section:reduction1d:numerical_methods}
    
	The semi-discretization in time of problem~\eqref{eq:reduction1d:charge:1d:final} is based on a two-step operator-splitting first order scheme, introduced in~\cite{lie1893theorie}, allowing us to separate the different physical phenomena described by the equations. In particular, we split the problem into two steps, decoupling the drift-diffusion equations from the chemical reactions and the equations of the dipole moment~\eqref{eq:reduction1d:dipole:a}, as detailed in Algorithm~\ref{algorithm:time-splitting}.\\
    Let $t_0 = 0,\dots,t_N=T$, such that $t_n = t_{n-1} + \Delta t$, $n=1,\dots,N$ and denote by $u^n(\mathbf{x}) = u(\mathbf{x},t_n)$, $n=0,\dots,N$, the value of a generic function $u\ :\ \mathbb{R}\times[0,T]\to\mathbb{R}$ at time $t_n$. We discretize the time derivative as $\dfrac{du}{dt} = \dfrac{u^{n+1} - u^n}{\Delta t}$.\\
    At each time step $t_n$, we first compute a first approximation, that we denote with the superscript $\cdot^{n+1/2}$, of the solution at the next time step $t_{n+1}$, by solving the chemical reaction equations, that is then used as initial condition to solve the dirft-diffusion equations~\eqref{eq:reduction1d:predictor:physics}. For the drift--diffusion equations, we employ the implicit Euler--Finite Volume method with upwind and TPFA, extended to 1D graphs, as discussed in~\cite{crippa2024numericalmethods}. This method ensures positivity of the solution at each time step, provided that the intermediate solutions $\cdot^{n+\frac{1}{2}}$ are non-negative. In order to guarantee positivity of the numerical solution to the chemical reaction equations at the discrete level, we employ a Patankar--Euler scheme~\cite[Chapter~7]{patankar2018numerical}, which consists of the following weighting of the destruction term~\eqref{eq:reduction1d:destruction}:

    \begin{equation}
        \dfrac{u^{n+1} - u^n}{\Delta t} =
        P(u^n) - D(u^n) \dfrac{u^{n+1}}{u^n}, \quad i=1,\,2, \quad n=0,\dots,N-1.
        \label{eq:reduction1d:patankar-euler}
    \end{equation}

    \begin{algorithm}
        \caption{\texttt{splitting}}
        \begin{algorithmic}[1]
            \Require $\bar{c}_{\Gamma,e}^0,\ \bar{c}_{\Gamma,h}^0, \ a_1^0,\ a_2^0,\ \bar{c}_p^0,\, p=1,\,2,\,3 $.
            \Ensure $\bar{c}_{\Gamma,e}^{n+1},\ \bar{c}_{\Gamma,h}^{n+1},\ a_1^n,\ a_2^n,\ \bar{c}_p^{n+1},\, p=1,\,2,\,3,\ \ n=1,\dots,N$.

            \For{$n=0,\dots,N-1$}

            \State Solve the chemical reaction equations (Algorithm~\ref{algorithm:chemistry}):

            $\bar{c}_p^{n+1/2},\bar{c}_{\Gamma,e}^{n+1/2}, \bar{c}_{\Gamma,h}^{n+1/2}, a_i^{n+1} \gets \texttt{chemistry}(\bar{c}_p^n,\bar{c}_{\Gamma,e}^n, \bar{c}_{\Gamma,h}^n,a_i^{n})$
            
            \State Solve the drift and drift-diffusion equations:
            
            \begin{equation}
                \begin{dcases}
                    \dfrac{\partial \bar{c}_p}{\partial t}
                    + \dfrac{\partial}{\partial s} \left( \omega_p \dfrac{\bar{\mu}_p}{|\mathcal{D}|} \bar{c}_p \bar{E}_\Lambda\right)
                    - \dfrac{\partial}{\partial s} \left( \bar{\nu}_p \dfrac{\partial \bar{c}_p}{\partial s} \right) 
                    = 0, \quad p=1,\,2,\,3,
                    & \text{in\ } (t_n,t_{n+1}),\\
                    \dfrac{\partial \bar{c}_{\Gamma,e}}{\partial t}
                    - \dfrac{\partial}{\partial s}\left( \dfrac{\bar{\mu}_{\Gamma,e}}{|\partial\mathcal{D}|}\bar{c}_{\Gamma,e} \bar{E}_\Lambda \right)
                    = 0, & \text{in\ } (t_n,t_{n+1}),\\
                    \dfrac{\partial \bar{c}_{\Gamma,h}}{\partial t}
                    + \dfrac{\partial}{\partial s}\left( \dfrac{\bar{\mu}_{\Gamma,h}}{|\partial\mathcal{D}|}\bar{c}_{\Gamma,h} \bar{E}_\Lambda \right)
                    = 0, & \text{in\ } (t_n,t_{n+1}),
                \end{dcases}
                \label{eq:reduction1d:predictor:physics}
            \end{equation}
        
        with $\bar{c}_p(t_n) = \bar{c}_p^{n+1/2},\, \bar{c}_{\Gamma,e}(t_n) = \bar{c}_{\Gamma,e}^{n+1/2},\, \bar{c}_{\Gamma,h}(t_n) = \bar{c}_{\Gamma,h}^{n+1/2} $, with an upwind-TPFA Finite Volume method.

        \EndFor
        \end{algorithmic}
        \label{algorithm:time-splitting}
    \end{algorithm}

        \begin{algorithm}
        \caption{\texttt{$\texttt{chemistry}$}}
        \begin{algorithmic}[1]
            \Require $\bar{c}_{\Gamma,e}^{n},\ \bar{c}_{\Gamma,h}^{n},\ a_1^n,\ a_2^n,\ \bar{c}_p^{n},\, p=1,\,2\,3$.
            \Ensure $\bar{c}_{\Gamma,e}^{n+1/2},\ \bar{c}_{\Gamma,h}^{n+1/2},\ a_1^{n+1},\ a_2^{n_1},\ \bar{c}_p^{n+1/2},\, p=1,\,2,\,3 $.

            \State Compute $G_{p,i}^+(a_1^{n},a_2^n,\mathbf{\bar{c}}^{n})$ and $G_{p,i}^-(a_1^{n},a_2^n,\mathbf{\bar{c}}^{n})$, for $i=1,\,2, \ p=1,\,2,\,3$;

            \State Compute $F_p^{+,n} = F_p^+(a_1^{n},a_2^n,\mathbf{\bar{c}}^{n})$ and $F_p^{-,n} = F_p^-(a_1^{n},a_2^n,\mathbf{\bar{c}}^{n})$, for $p=1,\,2\,3$;

            \State Solve equation~\eqref{eq:reduction1d:charge:1d:final:1} with Patankar--Euler scheme~\eqref{eq:reduction1d:patankar-euler}:

            \State Solve the volume charge concentration equation:

            \begin{equation*}
                \dfrac{\partial \bar{c}_p}{\partial t} = P_p(\bar{\mathbf{c}}) - D_p(\bar{\mathbf{c}}),\ p=1,\,2,\,3, \quad \text{in\ } (t_n,t_{n+1}),
            \end{equation*}

            \noindent with $\bar{c}_p(t_n) = \bar{c}_p^n$, discretized as in equation~\eqref{eq:reduction1d:predictor:chemistry:volume};

            \State Solve the surface charge concentration equations 

            \begin{equation*}
                \dfrac{d\bar{c}_{\Gamma,\star}^n}{dt} = f_{\Gamma,\star},\ \star=e,\,h, \quad \text{in\ } (t_n,t_{n+1}),
            \end{equation*}
            
            \noindent with $\bar{c}_{\Gamma,\star}(t_n) = \bar{c}_{\Gamma,\star}^n$, discretized as in equation~\eqref{eq:reduction1d:predictor:chemistry:surface}.
        \end{algorithmic}
        \label{algorithm:chemistry}
    \end{algorithm}

    \noindent The Patankar--Euler scheme for the chemical reaction part of equation~\eqref{eq:reduction1d:charge:1d:final:1}, after cancellation of the multiplicative term $\bar{c}_p^n$ in the destruction term, reads as follows:
    
    \begin{subnumcases}{\label{eq:reduction1d:predictor:chemistry:volume}}
   		\begin{split}
   			\dfrac{\bar{c}_1^{n+1/2} - \bar{c}_1^n}{\Delta t} = 
	   		\bar{\mu}_1\dfrac{\bar{c}_1^b}{|\mathcal{D}|}F_1^{-,n} &+ \alpha\mu_1|\mathbf{E}_g| \bar{c}_1^{n} + \\
	   		&- \left( \eta\mu_1|\mathbf{E}_g| + \dfrac{|\partial\mathcal{D}|}{|\mathcal{D}|}K_1 + \bar{\mu}_1\dfrac{\bar{c}_1^n}{|\mathcal{D}|}F_1^{+,n} \right)\bar{c}_1^{n+1/2} ,
   		\end{split}
   		\label{eq:reduction1d:predictor:chemistry:volume:1}\\
   		\dfrac{\bar{c}_2^{n+1/2} - \bar{c}_2^n}{\Delta t} =
   		\bar{\mu}_2\dfrac{\bar{c}_2^b}{|\mathcal{D}|}F_2^{-,n} + \eta\mu_1|\mathbf{E}_g| \bar{c}_1^{n+1/2}
   		- \left( \dfrac{|\partial\mathcal{D}|}{|\mathcal{D}|}K_2 + \bar{\mu}_2\dfrac{\bar{c}_2^n}{|\mathcal{D}|}F_2^{+,n} \right)\bar{c}_2^{n+1/2} ,
   		\label{eq:reduction1d:predictor:chemistry:volume:2}\\
   		\dfrac{\bar{c}_3^{n+1/2} - \bar{c}_3^n}{\Delta t} =
   		\bar{\mu}_3\dfrac{\bar{c}_3^b}{|\mathcal{D}|}F_3^{-,n} + \alpha\mu_1|\mathbf{E}_g| \bar{c}_1^{n}
   		- \left( \dfrac{|\partial\mathcal{D}|}{|\mathcal{D}|}K_3 + \bar{\mu}_3\dfrac{\bar{c}_3^n}{|\mathcal{D}|}F_3^{+,n} \right)\bar{c}_3^{n+1/2} .
   		\label{eq:reduction1d:predictor:chemistry:volume:3}
   	\end{subnumcases}

    \noindent Notice that all the nonlinearities are treated explicitly, thus the equation can be directly solved for $\bar{c}_p^{n+1/2}$, and that the ionization terms are treated explicitly, ensuring unconditional positivity of the solutions, as pointed out in~\cite{villa2017implicit}.\\
    Finally, for the chemical reaction part of equations~\eqref{eq:reduction1d:charge:1d:final:2} and~\eqref{eq:reduction1d:charge:1d:final:3}, we introduce the following numerical scheme:
           
    \begin{equation}
        \begin{dcases}
            \bar{c}_{\Gamma,e}^{n+1/2} =
            \bar{c}_{\Gamma,e}^{n} 
            + \Delta t \sum_{p=1}^2 \bar{\mu}_p F_p^{+,n} \dfrac{\bar{c}_p^{n+1/2}}{|\mathcal{D}|}
            - \Delta t \bar{\mu}_3 F_3^{-,n} \dfrac{\bar{c}_3^b}{|\mathcal{D}|}
            + \sum_{p=1}^2 K_p \dfrac{|\partial\mathcal{D}|}{|\mathcal{D}|} \bar{c}_p^{n+1/2} ,
            \\
            \bar{c}_{\Gamma,h}^{n+1/2} =
            \bar{c}_{\Gamma,h}^{n} 
            + \Delta t \bar{\mu}_3 F_3^{+,n} \dfrac{\bar{c}_3^{n+1/2}}{|\mathcal{D}|}
            % + \Delta t \mathcal{N}_3^-(\bar{\mathbf{c}}^n).
            - \Delta t \sum_{p=1}^2 \bar{\mu}_p F_p^{+,n}  \dfrac{\bar{c}_p^b}{|\mathcal{D}|}
            + K_3 \dfrac{|\partial\mathcal{D}|}{|\mathcal{D}|} \bar{c}_3^{n+1/2} 
        \end{dcases}
        \label{eq:reduction1d:predictor:chemistry:surface}
    \end{equation}

    \noindent We can show that the proposed numerical scheme preserves the properties of monotonicity and conservation of charge of the continuous problem, stated in Theorems~\ref{theorem:ai:positive},~\ref{theorem:cp:positive},~\ref{theorem:cgamma:positive} and~\ref{theorem:charge:conservative}:

    \begin{theorem}
        Consider the numerical scheme described in Algorithms~\ref{algorithm:time-splitting} and~\ref{algorithm:chemistry}, applied to solve problem~\eqref{eq:reduction1d:charge:1d:final}. The discrete solutions are non--negative at every time step and the discrete problem conserves the total charge.
        \label{theorem:predictor}
    \end{theorem}

    \begin{proof}		    	
        Let us start from the monotonicity.\\
        From equation~\eqref{eq:reduction1d:predictor:chemistry:volume} we can directly compute $\bar{c}_{p}^{n+1/2}$ as follows:

        \begin{align*}
            \bar{c}_1^{n+1/2} &= \dfrac{\bar{c}_1^n + \Delta t \left( \alpha \mu_1 |\mathbf{E}_g|\bar{c}_1^n  - \bar{\mu}_1\dfrac{\bar{c}_1^b}{|\mathcal{D}|}F_1^{-,n} \right)}{1 + \Delta t \left(\eta \mu_1 |\mathbf{E}_g| + \dfrac{|\partial\mathcal{D}|}{|\mathcal{D}|}K_1 + \bar{\mu}_1\dfrac{\bar{c}_1^n}{|\mathcal{D}|}F_1^{+,n}\right)}, \quad n=0,\dots,N-1,\\
            \bar{c}_2^{n+1/2} &= \dfrac{\bar{c}_2^n 
            	+ \Delta t \left( 
           	\eta \mu_1 |\mathbf{E}_g|\bar{c}_1^{n+1/2} - \bar{\mu}_2\dfrac{\bar{c}_2^b}{|\mathcal{D}|}F_2^{-,n}\right)}{1 + \Delta t \left(\dfrac{|\partial\mathcal{D}|}{|\mathcal{D}|}K_2 + \bar{\mu}_2\dfrac{\bar{c}_2^n}{|\mathcal{D}|}F_2^{+,n}\right)}, \quad n=0,\dots,N-1,\\
            \bar{c}_3^{n+1/2} &= \dfrac{\bar{c}_3^n + \Delta t \left( \alpha \mu_1 |\mathbf{E}_g|\bar{c}_1^n  - \bar{\mu}_3\dfrac{\bar{c}_3^b}{|\mathcal{D}|}F_3^{-,n} \right)}{1 + \Delta t \left( \dfrac{|\partial\mathcal{D}|}{|\mathcal{D}|}K_3 + \bar{\mu}_3\dfrac{\bar{c}_3^n}{|\mathcal{D}|}F_3^{+,n}\right)}, \quad n=0,\dots,N-1.
        \end{align*}

        \noindent Provided that $c_p^n\geq 0$, both the numerator and denominator of the right--hand side are non-negative, as a consequence of Lemma~\ref{theorem:FG:sign}. If we start from a non-negative initial condition $c_p^0 \geq 0$, we obtain $c_p^{1/2} \geq 0$ and, according to~\cite[Theorem 7]{crippa2024numericalmethods}, $\bar{c}_p^1 \geq 0$. Thus, by induction, $c_p^{n+1/2} \geq 0$ and $\bar{c}_p^{n+1} \geq 0,\, \forall n=0,\,\dots,\,N-1$.\\
        As a consequence of Lemma~\ref{theorem:FG:sign}, the right--hand side of equation~\eqref{eq:reduction1d:predictor:chemistry:surface} is non-negative, thus $\bar{c}_{\Gamma,e}^{n+1/2}\geq 0$ and $\bar{c}_{\Gamma,h}^{n+1/2}\geq 0,\, \forall n=0,\,\dots,\,N-1$, provided that $\bar{c}_{\Gamma,e}^n\geq 0$ and $\bar{c}_{\Gamma,h}^n\geq 0$. If we start from a non-negative initial condition, we obtain $\bar{c}_{\Gamma,e}^{1/2} \geq 0$ and $\bar{c}_{\Gamma,h}^{1/2}\geq 0$, ensuring $\bar{c}_{\Gamma,e}^1 \geq 0$ and $\bar{c}_{\Gamma,e}^1\geq 0$, according to~\cite[Theorem 1]{crippa2024numericalmethods}, and by induction we can conclude that $\bar{c}_{\Gamma,e}^{n+1/2}\geq 0$, $\bar{c}_{\Gamma,h}^{n+1/2}\geq 0$, $\bar{c}_{\Gamma,e}^{n+1}\geq 0$ and $\bar{c}_{\Gamma,h}^{n+1}\geq 0$, $\forall n=0,\,\dots,\,N-1$.\\
        Finally, the Patankar--Euler scheme~\eqref{eq:reduction1d:patankar-euler}, applied to the dipole moment equation~\eqref{eq:reduction1d:charge:1d:final:4}, ensures monotonicity of the corresponding discrete solution at each time step (see~\cite{patankar2018numerical}).

        In order to show that the discrete problem is conservative, we multiply equation~\eqref{eq:reduction1d:predictor:chemistry:volume} by $\omega_p$, and sum equations~\eqref{eq:reduction1d:predictor:chemistry:surface} and~\eqref{eq:reduction1d:predictor:chemistry:volume} all together, multiplied by the electron charge $e$, as in the proof of the Theorem~\ref{theorem:charge:conservative}, obtaining:

        \begin{align*}
            \dfrac{\bar{c}^{n+1/2} -\bar{c}^n}{\Delta t} =&
            \ e \sum_{p=1}^3 \omega_p\left(- \bar{\mu}_p\dfrac{\bar{c}_p^b}{|\mathcal{D}|}F_p^{-,n} \right) -
            e\sum_{p=1}^3 \omega_p \left(\dfrac{|\partial\mathcal{D}|}{|\mathcal{D}|}K_p - \bar{\mu}_p\dfrac{\bar{c}_p^n}{|\mathcal{D}|}F_p^{+,n}
            \right) \bar{c}_p^{n+1/2} + \\
            &+e \sum_{p=1}^2 \bar{\mu}_p\dfrac{\bar{c}_p^b}{|\mathcal{D}|}F_p^{-,n} +
            e\sum_{p=1}^2 \left(\dfrac{|\partial\mathcal{D}|}{|\mathcal{D}|}K_p - \bar{\mu}_p\dfrac{\bar{c}_p^n}{|\mathcal{D}|}F_p^{+,n}
            \right) \bar{c}_p^{n+1/2} +\\
            & + e \left(-\alpha \bar{c}_1^{n} + \eta \bar{c}_1^{n+1/2}
            + \alpha \bar{c}_1^n - \eta \bar{c}_1^{n+1/2}\right) \mu_1 |\mathbf{E}_g|+\\
            &+e \bar{\mu}_3\dfrac{\bar{c}_3^b}{|\mathcal{D}|}F_3^{-,n}  +
            e \left(\dfrac{|\partial\mathcal{D}|}{|\mathcal{D}|}K_3 - \bar{\mu}_3\dfrac{\bar{c}_3^n}{|\mathcal{D}|}F_3^{+,n}
            \right) \bar{c}_3^{n+1/2} = 0, \, n=0,\,\dots,\,N-1.
        \end{align*}

        \noindent Thus, the discrete problem conserves the total net charge.
    \end{proof}

    \begin{remark}
    \label{remark:splitting}
        The time splitting scheme described in Algorithm~\ref{algorithm:time-splitting} allows us to separate the chemical reactions from the drift and diffusion of charged particles, introducing an approximation of the same order of accuracy as the time discretization scheme. Indeed, substituting the system of equations~\eqref{eq:reduction1d:predictor:chemistry:volume} and~\eqref{eq:reduction1d:predictor:chemistry:surface} into~\eqref{eq:reduction1d:predictor:physics}, we obtain a semi-implicit discretization in time of equations~\eqref{eq:reduction1d:charge:1d:final:1} -~\eqref{eq:reduction1d:charge:1d:final:1}, where the drift-diffusion part is treated implicitly and the positive reactions are treated explicitly.
    \end{remark}
    
	\section{Results}
	\label{section:reduction1d:results}
    In this section we test the 1D model describing the evolution of charge concentrations in the initial stage of the electrical treeing, when the displacement of charges inside the defect does not significantly influence the electric field. Since in this work we consider the tangential electric field as a fixed known quantity and the charge concentrations in later phases are high enough to influence its value, we are limited to the avalanche. However, a future coupling of this model to the mixed--dimensional electrostatic problem presented in~\cite{crippa2024mixed} will allow simulations of complete partial discharges. We simulate, first on a straight line, then on a simple branched domain, and finally on the geometry of a realistic electrical treeing, the electron avalanche phase, consisting of an exponential increase of the concentration of charged particles. In all cases we suppose that at the initial time all the charge concentrations are equal to 0, and that electrons enter the domain from an inlet point and are free to leave from outflow nodes. The charged particles are transported along the 1D domain by a fixed tangential electric field $E_\Lambda = 10 \ \frac{MV}{m}$, multiplied by the charge mobilities
    $\bar{\mu}_1 = 8.1\cdot10^{-5} \ \frac{m^2}{kV \ \mu s},\ \bar{\mu}_2 = 2.7\cdot10^{-7} \ \frac{m^2}{kV \ \mu s},\ \bar{\mu}_3 = 2.3\cdot10^{-7} \ \frac{m^2}{kV \ \mu s},\ \bar{\mu}_{\Gamma,e} = \bar{\mu}_{\Gamma,h} = 10^{-13} \ \frac{m^2}{kV \ \mu s}.$ The diffusion coefficient is set to 0 for all the charged species, except for the electrons, with $\bar{\nu}_1 = 2.3\cdot 10^{-4}\ \frac{m^2}{\mu s}$. Finally, the corresponding coefficients of the chemical reactions described in equation~\eqref{eq:chemistry:expression} are fixed to $\alpha = 10^{5} \ \frac{1}{m}$ and $\eta = 3.4\cdot 10^{-3}\frac{1}{m}$. All the mentioned parameters are dependent on the electric field in the gas and have been computed according to the estimates provided by Morrow and Lowke~\cite{morrow1997streamer} for given $E_\Lambda$.

    \subsection{TC1: Straight line}
    \label{section:reduction1d:TC1}
    We start by considering as 1D domain the segment $\Lambda = \{s\in [0,10^{-4}]m\}$, with an electric field $\mathbf{E}_g = E_\Lambda\mathbf{s}$ directed along the line, from $s=0$ to $s=10^{-4}\ m$. The positive charges are transported in the same direction as the electric field, while the negative charges in the opposite direction. In particular, at $s=10^{-4}\ m$ we set a Dirichlet boundary condition $\bar{c}_1 = 10^{6} \ \frac{1}{m}$ on the electrons, which are transported inside the domain with velocity $v_1 = \bar{\mu}_1 E_\Lambda = 0.81 \ \frac{m}{\mu s}$. We simulate the evolution of the concentration of all the charged particles for $T = 0.2\ ns$, computing at each time step the radial components of the electric field induced by the volume charge. For this test case, we suppose that the external electric field is much smaller than $E_\Lambda$, and that its magnitude can be considered zero. As a consequence, the dipole moments remain null, since the surface charge distribution is homogeneous.
    
    In Figures~\ref{fig:TC1:cp:x} and~\ref{fig:TC1:cgamma:x} we can observe the evolution of the volume and surface charge concentrations, respectively, at different time steps. At the first time step (Figure~\ref{fig:TC1:cp:x:t0}), the electrons injected from $s=10^{-4} \ m$ are transported into the domain $\Lambda$, with an advection velocity that is opposite to the electric field, while all the other charged species have zero density. Then, thanks to the ionization and attachment, also positive and negative ions are formed, as we can see in Figure~\ref{fig:TC1:cp:x:t50}. From Figures~\ref{fig:TC1:cp:x:t50} and~\ref{fig:TC1:cp:x:t100}, we notice that the elevated ionization coefficient $\alpha$ makes the reaction dominant over the transport of positive ions. A similar behavior is also observed in Figure~\ref{fig:TC1:cgamma:x}, showing the displacement of positive and negative charge densities on the surface. Since we have set very small transport velocities for $c_{\Gamma,e}$ and $c_{\Gamma,h}$, their evolution over time is almost entirely due to the chemical reactions involving the volume charge transfer from the gas domain. We observe that the surface electron concentration increases quickly as the volume charge increases, while there is a slight delay in the formation of holes. This indicates that chemical reactions in the gas contribute to peak of volume electrons more than the detachment of electrons from the surface, and that electron attachment to the surface is more likely than detachment in this phase. The evolution of volume charge densities over the considered time interval at the outflow point $s = 0\ m$ is represented in Figure~\ref{fig:TC1:cp:t}: it shows an initially fast exponential increase of the concentrations of all the charged species, that slows down after $0.1\ ns$. At this time the peaks of charge concentrations in the domain reach the outflow point, thus most of the charge is transported out of the domain and the solution is almost steady. At the end of the simulation, the concentrations of electrons and positive ions become of order $10^{10}\ \frac{1}{m}$ and $10^{11}\ \frac{1}{m}$, respectively (much higher than the inflow condition), while the concentration of negative ions is of order $10^0\ \frac{1}{m}$, reflecting the notable difference in the values of the ionization and attachment reaction coefficients.\\
    These results are in agreement with the numerical and experimental tests performed in similar conditions by Morrow and Lowke~\cite{morrow1997streamer} and Tran et al.~\cite{tran2010numerical}.

    Finally, in Figure~\ref{fig:TC1:Eg} we observe the magnitude of the transverse electric field at different time steps: since there is no external electric field, the non--radial components are null and the only visible effect is due to the presence of volume charge. In particular, the radial electric field is negative, i.e. directed towards the 1D domain, at all the points of the domain where the net charge density $\bar{q}=\sum_{p=1}^2 \omega_p \bar{c}_p$ is negative, and positive, i.e. directed away from the 1D domain, where the net charge density is positive.

    \begin{figure}
        \centering
        \subfloat[$t=0.01 \ ns$\label{fig:TC1:cp:x:t0}]
        % {\includegraphics[width=\linewidth]{concentrations_volume_t0.png}}\\
        {
            \begin{tikzpicture}\scriptsize
                \begin{axis}[
                	width=.2\textwidth,height=.1\textwidth,
                    scale only axis,
                    xtick={0,0.25e-4,0.5e-4,0.75e-4,1e-4},
                    xlabel={$s\ \left(m\right)$},
                    ylabel={$c_1 \ \left(1/m\right)$},
                    title style={yshift=4pt},
                    grid=major,
                    grid style={dotted}
                  ]
                    \addplot[
                      mark=none,
                      line width=1.5pt
                    ] 
                    table[
                      col sep=comma,
                      x=arc_length,
                      y=c1
                    ] {./data/negative_charges_0.csv};
                \end{axis}
                \begin{axis}[
                	width=.2\textwidth,height=.1\textwidth,
                    at={(4.5cm,0)},
                    scale only axis,
                    ymin=-1e-4,
                    ymax=1e-3,
                    xtick={0,0.25e-4,0.5e-4,0.75e-4,1e-4},
                    xlabel={$s\ \left(m\right)$},
                    ylabel={$c_2 \ \left(1/m\right)$},
                    title style={yshift=2pt},
                    grid=major,
                    grid style={dotted}
                  ]
                    \addplot[
                      mark=none,
                      line width=1.5pt
                    ] 
                    table[
                      col sep=comma,
                      x=arc_length,
                      y=field
                    ] {./data/c2_0.csv};
                \end{axis}
                \begin{axis}[
                	width=.2\textwidth,height=.1\textwidth,
                    at={(9cm,0)},
                    scale only axis,
                    ymin=-1e-4,
                    ymax=1e-3,
                    xtick={0,0.25e-4,0.5e-4,0.75e-4,1e-4},
                    xlabel={$s\ \left(m\right)$},
                    ylabel={$c_3 \ \left(1/m\right)$},
                    title style={yshift=2pt},
                    grid=major,
                    grid style={dotted}
                  ]
                    \addplot[
                      mark=none,
                      line width=1.5pt
                    ] 
                    table[
                      col sep=comma,
                      x=arc_length,
                      y=field
                    ] {./data/c3_0.csv};
                \end{axis}
            \end{tikzpicture}
        }\\
        \subfloat[$t=0.1 \ ns$\label{fig:TC1:cp:x:t50}]
        % {\includegraphics[width=\linewidth]{concentrations_volume_t50.png}}
        {
            \begin{tikzpicture}\scriptsize
                \begin{axis}[
                	width=.2\textwidth,height=.1\textwidth,
                    scale only axis,
                    xtick={0,0.25e-4,0.5e-4,0.75e-4,1e-4},
                    xlabel={$s\ \left(m\right)$},
                    ylabel={$c_1 \ \left(1/m\right)$},
                    title style={yshift=4pt},
                    grid=major,
                    grid style={dotted}
                  ]
                    \addplot[
                      mark=none,
                      line width=1.5pt
                    ] 
                    table[
                      col sep=comma,
                      x=arc_length,
                      y=c1
                    ] {./data/negative_charges_10.csv};
                \end{axis}
                \begin{axis}[
                	width=.2\textwidth,height=.1\textwidth,
                	at={(4.5cm,0)},
                    scale only axis,
                    xtick={0,0.25e-4,0.5e-4,0.75e-4,1e-4},
                    xlabel={$s\ \left(m\right)$},
                    ylabel={$c_2 \ \left(1/m\right)$},
                    title style={yshift=2pt},
                    grid=major,
                    grid style={dotted}
                  ]
                    \addplot[
                      mark=none,
                      line width=1.5pt
                    ] 
                    table[
                      col sep=comma,
                      x=arc_length,
                      y=c2
                    ] {./data/negative_charges_10.csv};
                \end{axis}
                \begin{axis}[
                	width=.2\textwidth,height=.1\textwidth,
                    scale only axis,
                    at={(9cm,0)},
                    xtick={0,0.25e-4,0.5e-4,0.75e-4,1e-4},
                    xlabel={$s\ \left(m\right)$},
                    ylabel={$c_3 \ \left(1/m\right)$},
                    title style={yshift=2pt},
                    grid=major,
                    grid style={dotted}
                  ]
                    \addplot[
                      mark=none,
                      line width=1.5pt
                    ] 
                    table[
                      col sep=comma,
                      x=arc_length,
                      y=c3
                    ] {./data/c3_10.csv};
                \end{axis}
            \end{tikzpicture}
        }\\
        \subfloat[$t=0.2 \ ns$\label{fig:TC1:cp:x:t100}]
        % {\includegraphics[width=\linewidth]{concentrations_volume_t100.png}}
        {
            \begin{tikzpicture}\scriptsize
                \begin{axis}[
                	width=.2\textwidth,height=.1\textwidth,
                    scale only axis,
                    xtick={0,0.25e-4,0.5e-4,0.75e-4,1e-4},
                    xlabel={$s\ \left(m\right)$},
                    ylabel={$c_1 \ \left(1/m\right)$},
                    title style={yshift=4pt},
                    grid=major,
                    grid style={dotted}
                  ]
                    \addplot[
                      mark=none,
                      line width=1.5pt
                    ] 
                    table[
                      col sep=comma,
                      x=arc_length,
                      y=c1
                    ] {./data/negative_charges_20.csv};
                \end{axis}
                \begin{axis}[
                	width=.2\textwidth,height=.1\textwidth,
                    scale only axis,
                    at={(4.5cm,0)},
                    xtick={0,0.25e-4,0.5e-4,0.75e-4,1e-4},
                    xlabel={$s\ \left(m\right)$},
                    ylabel={$c_3 \ \left(1/m\right)$},
                    title style={yshift=2pt},
                    grid=major,
                    grid style={dotted}
                  ]
                    \addplot[
                      mark=none,
                      line width=1.5pt
                    ] 
                    table[
                      col sep=comma,
                      x=arc_length,
                      y=c2
                    ] {./data/negative_charges_20.csv};
                \end{axis}
                \begin{axis}[
                	width=.2\textwidth,height=.1\textwidth,
                    scale only axis,
                    at={(9cm,0)},
                    xtick={0,0.25e-4,0.5e-4,0.75e-4,1e-4},
                    xlabel={$s\ \left(m\right)$},
                    ylabel={$c_3 \ \left(1/m\right)$},
                    title style={yshift=2pt},
                    grid=major,
                    grid style={dotted}
                  ]
                    \addplot[
                      mark=none,
                      line width=1.5pt
                    ] 
                    table[
                      col sep=comma,
                      x=arc_length,
                      y=c3
                    ] {./data/c3_20.csv};
                \end{axis}
            \end{tikzpicture}
        }
        \caption{\textbf{TC1 -} Volume concentrations of charged species on $\Lambda=[0,10^{-4}]\ m$, for $\Delta t = 10^{-6}\ \mu s$.}
        \label{fig:TC1:cp:x}
    \end{figure}
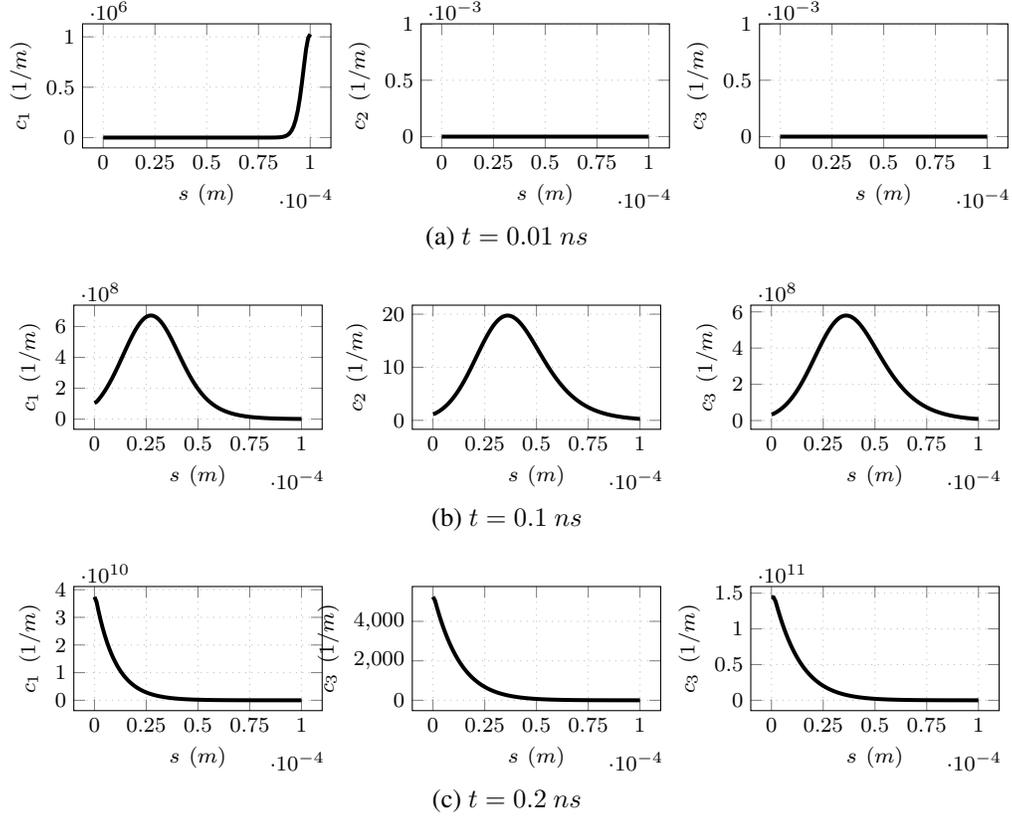

    \begin{figure}
        \centering
        \subfloat[$t=0.1 \ ns$\label{fig:TC1:cgamma:x:t50}]
        % {\includegraphics[width=\linewidth]{concentrations_surface_t50.png}}\\
        {
            \begin{tikzpicture}\scriptsize
                \begin{axis}[
                	width=.2\textwidth,height=.1\textwidth,
                    scale only axis,
                    xtick={0,0.25e-4,0.5e-4,0.75e-4,1e-4},
                    xlabel={$s\ \left(m\right)$},
                    ylabel={$c_ {\Gamma,e} \ \left(1/m\right)$},
                    title style={yshift=4pt},
                    grid=major,
                    grid style={dotted}
                  ]
                    \addplot[
                      mark=none,
                      line width=1.5pt
                    ] 
                    table[
                      col sep=comma,
                      x=arc_length,
                      y=ce
                    ] {./data/ce_10.csv};
                \end{axis}
                \begin{axis}[
                	width=.2\textwidth,height=.1\textwidth,
                    at={(5cm,0)},
                    scale only axis,
                    xtick={0,0.25e-4,0.5e-4,0.75e-4,1e-4},
                    xlabel={$s\ \left(m\right)$},
                    ylabel={$c_{\Gamma,h} \ \left(1/m\right)$},
                    title style={yshift=2pt},
                    grid=major,
                    grid style={dotted}
                  ]
                    \addplot[
                      mark=none,
                      line width=1.5pt
                    ] 
                    table[
                      col sep=comma,
                      x=arc_length,
                      y=ch
                    ] {./data/ch_10.csv};
                \end{axis}
            \end{tikzpicture}
        }\\
        \subfloat[$t=0.2 \ ns$\label{fig:TC1:cgamma:x:t100}]
        % {\includegraphics[width=\linewidth]{concentrations_surface_t100.png}}
        {
             \begin{tikzpicture}\scriptsize
                \begin{axis}[
                	width=.2\textwidth,height=.1\textwidth,
                    scale only axis,
                    xtick={0,0.25e-4,0.5e-4,0.75e-4,1e-4},
                    xlabel={$s\ \left(m\right)$},
                    ylabel={$c_{\Gamma,e} \ \left(1/m\right)$},
                    title style={yshift=4pt},
                    grid=major,
                    grid style={dotted}
                  ]
                    \addplot[
                      mark=none,
                      line width=1.5pt
                    ] 
                    table[
                      col sep=comma,
                      x=arc_length,
                      y=ce
                    ] {./data/ce_20.csv};
                \end{axis}
                \begin{axis}[
                	width=.2\textwidth,height=.1\textwidth,
                    at={(5cm,0)},
                    scale only axis,
                    xtick={0,0.25e-4,0.5e-4,0.75e-4,1e-4},
                    xlabel={$s\ \left(m\right)$},
                    ylabel={$c_{\Gamma,h} \ \left(1/m\right)$},
                    title style={yshift=2pt},
                    grid=major,
                    grid style={dotted}
                  ]
                    \addplot[
                      mark=none,
                      line width=1.5pt
                    ] 
                    table[
                      col sep=comma,
                      x=arc_length,
                      y=ch
                    ] {./data/ch_20.csv};
                \end{axis}
            \end{tikzpicture}
        }
        \caption{\textbf{TC1 -} Surface charge concentrations on $\Lambda=[0,10^{-4}]\ m$, for $\Delta t = 10^{-6}$.}
        \label{fig:TC1:cgamma:x}
    \end{figure}
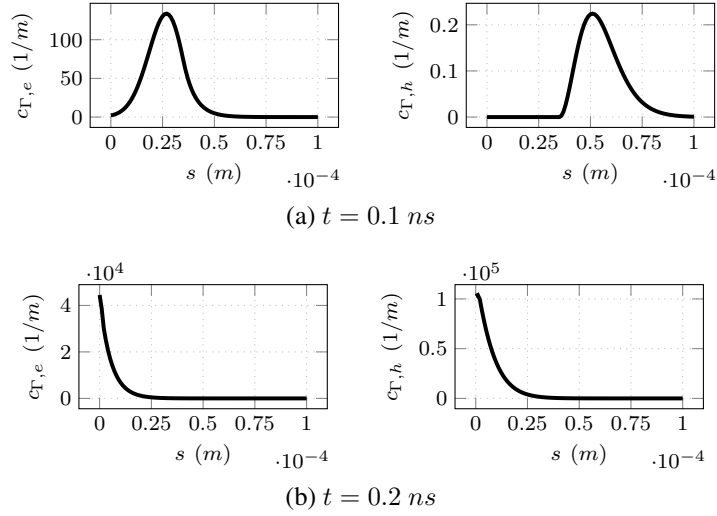

    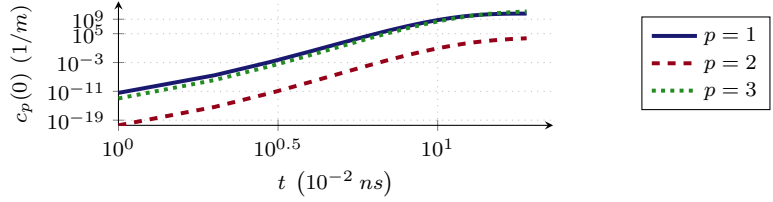
\begin{figure}
    \centering
    \begin{tikzpicture}\scriptsize
        \begin{axis}[
            width=.35\textwidth,
            height=.1\textwidth,
            scale only axis,
            axis y line=left,
            axis x line=bottom,
            ymode=log,
            xmode=log,
            ymax=1e+14,
            xmax=23,
            ytick={1e-19,1e-11,1e-3,1e+5,1e+9},
            xlabel={$t\ \left(10^{-2} \ ns\right)$},
            ylabel={$c_p(0) \ \left(1/m\right)$},
            legend entries={
                $p=1$, $p=2$, $p=3$
            },
            legend style={
                at={(1.5,.5)},
                anchor=east
            },
            grid=major,
            grid style={dotted}
          ]
            \addplot[
              color=darkBlue,
              mark=none,
              line width=1.5pt
            ] 
            table[
              col sep=comma,
              x=Time,
              y=c1
            ] {./data/c1_over_time_20.csv};
            \addplot[
              color=darkRed,
              mark=none,
              dashed,
              line width=1.5pt
            ] 
            table[
              col sep=comma,
              x=Time,
              y=c2
            ] {./data/c2_over_time_20.csv};
            \addplot[
              color=lightGreen,
              mark=none,
              dotted,
              line width=1.5pt
            ] 
            table[
              col sep=comma,
              x=Time,
              y=c3
            ] {./data/c3_over_time_20.csv};
        \end{axis}
    \end{tikzpicture}
    \caption{\textbf{TC1 -} Volume charge concentration at $s=0$ from $t=0 \ ns$ to $t=0.2 \ ns$, for $\Delta t = 10^{-6}$.}
    \label{fig:TC1:cp:t}
    \end{figure}

    \begin{figure}
        \centering
        \subfloat[$t=0.01 \ ns$\label{fig:TC1:Eg:t0}]
        %{\includegraphics[width=.3\linewidth]{electric_field_r_t0.png}}
        {
            \begin{tikzpicture}
                \scriptsize
                \begin{axis}[
                    width=.28\textwidth,height=.1\textwidth,
                    scale only axis,
                    name=leftaxis,
                    xlabel={$s\ (m)$},
                    ylabel={charge concentration $\left(1/m\right)$},
                    axis y line*=left,
                    axis x line=bottom,
                    grid=major,
                    grid style={dotted}
                  ]
                    \addplot[
                      color=darkRed,
                      mark=none,
                      dashed,
                      line width=1.5pt
                    ] 
                    table[
                      col sep=comma,
                      x=arc_length,
                      y=c1
                    ] {./data/negative_charges_0.csv};
                    \addplot[
                        color=lightGreen,
                        mark=none,
                        dotted,
                        line width=1.5pt
                    ]
                    table[
                        col sep=comma,
                        x=arc_length,
                        y=field
                    ]{./data/c3_0.csv};
                  \end{axis}

                  \begin{axis}[
                  	width=.28\textwidth,height=.1\textwidth,
                    scale only axis,
                    name=rightaxis,
                    axis y line*=right,
                    axis x line =bottom,
                    ylabel={electric field $\left(V/m \right)$},
                    ymin=-1e-17,
                    ymax=1e-17,
                    grid=major,
                    grid style={dashdotted}
                  ]
                    \addplot[
                        color=darkBlue,
                        mark=none,
                        line width=1.5pt
                    ]
                    table[
                        col sep=comma,
                        x=arc_length,
                        y=field
                    ]{./data/Eg_0.csv};
                    \addplot[darkGray,mark=none] coordinates{(0,0)(1e-4,0)};
                  \end{axis}
            \end{tikzpicture}
        }\qquad
        % \hspace{20pt}
        \subfloat[$t=0.1 \ ns$\label{fig:TC1:Eg:t50}]
        %{\includegraphics[width=.3\linewidth]{electric_field_r_t50.png}}
        {
        \begin{tikzpicture}
                \scriptsize
                \begin{axis}[
                    width=.28\textwidth,height=.1\textwidth,
                    scale only axis,
                    name=leftaxis,
                    xlabel={$s\ (m)$},
                    ylabel={charge concentration $\left(1/m\right)$},
                    axis y line*=left,
                    axis x line=bottom,
%                    legend entries={
%                        $c_1+c_2$, $c_3$, $\mathbf{E}_g$
%                    },
                    grid=major,
                    grid style={dotted}
                  ]
                    \addplot[
                      color=darkRed,
                      mark=none,
                      dashed,
                      line width=1.5pt
                    ] 
                    table[
                      col sep=comma,
                      x=arc_length,
                      y=negative_charge
                    ] {./data/negative_charges_10.csv};
                    \addplot[
                        color=lightGreen,
                        mark=none,
                        dotted,
                        line width=1.5pt
                    ]
                    table[
                        col sep=comma,
                        x=arc_length,
                        y=c3
                    ]{./data/c3_10.csv};

                    \addlegendimage{darkBlue, ultra thick, mark=none}
                  \end{axis}

                  \begin{axis}[
                  	width=.28\textwidth,height=.1\textwidth,
                    scale only axis,
                    name=rightaxis,
                    axis y line*=right,
                    axis x line =bottom,
                    ylabel={electric field $\left(V/m \right)$},
                    ymin=-1e-15,
                    ymax=1e-15,
                    grid=major,
                    grid style={dashdotted}
                  ]
                    \addplot[
                        color=darkBlue,
                        mark=none,
                        line width=1.5pt
                    ]
                    table[
                        col sep=comma,
                        x=arc_length,
                        y=field
                    ]{./data/Eg_10.csv};
                    \addplot[darkGray,mark=none] coordinates{(0,0)(1e-4,0)};
                  \end{axis}
            \end{tikzpicture}
        }\\
        \subfloat[$t=0.2 \ ns$\label{fig:TC1:Eg:t100}]
        %{\includegraphics[width=.3\linewidth]{electric_field_r_t100.png}}
        {
        \begin{tikzpicture}
                \scriptsize
                \begin{axis}[
                    width=.28\textwidth,height=.1\textwidth,
                    scale only axis,
                    name = leftaxis,
                    xlabel={$s\ (m)$},
                    ylabel={charge concentration $\left(1/m\right)$},
                    axis y line*=left,
                    axis x line=bottom,
                     legend entries={
                         $\bar{c}_1+\bar{c}_2$, $\bar{c}_3$, $\mathbf{E}_g$
                     },
                    % legend style={
                    %   at={(1.9,0.5)},
                    %   anchor=east,
                    % },
                    legend style={
                    at={(1.5,0.5)},
                    anchor=west
                    },
                    grid=major,
                    grid style={dotted}
                  ]
                    \addplot[
                      color=darkRed,
                      mark=none,
                      dashed,
                      line width=1.5pt
                    ] 
                    table[
                      col sep=comma,
                      x=arc_length,
                      y=negative_charge
                    ] {./data/negative_charges_20.csv};
                    \addplot[
                        color = lightGreen,
                        mark = none,
                        dotted,
                        line width=1.5pt
                    ]
                    table[
                        col sep = comma,
                        x = arc_length,
                        y = c3
                    ]{./data/c3_20.csv};

                    \addlegendimage{darkBlue, ultra thick, mark = none}
                  \end{axis}

                  \begin{axis}[
                  	width=.28\textwidth,height=.1\textwidth,
                    scale only axis,
                    name = rightaxis,
                    axis y line*=right,
                    axis x line =bottom,
                    ylabel = {electric field $\left(V/m \right)$},
                    ymin = -5e-14,
                    ymax = 5e-14,
                    grid=major,
                    grid style={dashdotted}
                  ]
                    \addplot[
                        color = darkBlue,
                        mark = none,
                        line width=1.5pt
                    ]
                    table[
                        col sep = comma,
                        x = arc_length,
                        y = field
                    ]{./data/Eg_20.csv};
                    \addplot[darkGray,mark=none] coordinates{(0,0)(1e-4,0)};
                  \end{axis}
            \end{tikzpicture}
        }
        \caption{\textbf{TC1 -} Radial electric field on $\Lambda=[0,10^{-4}]\ m$ (blue solid line), compared to the concentration of positive and negative volume charge (green dotted line and red dashed line, respectively), obtained with $\Delta t = 10^{-6}$.}
        \label{fig:TC1:Eg}
    \end{figure}
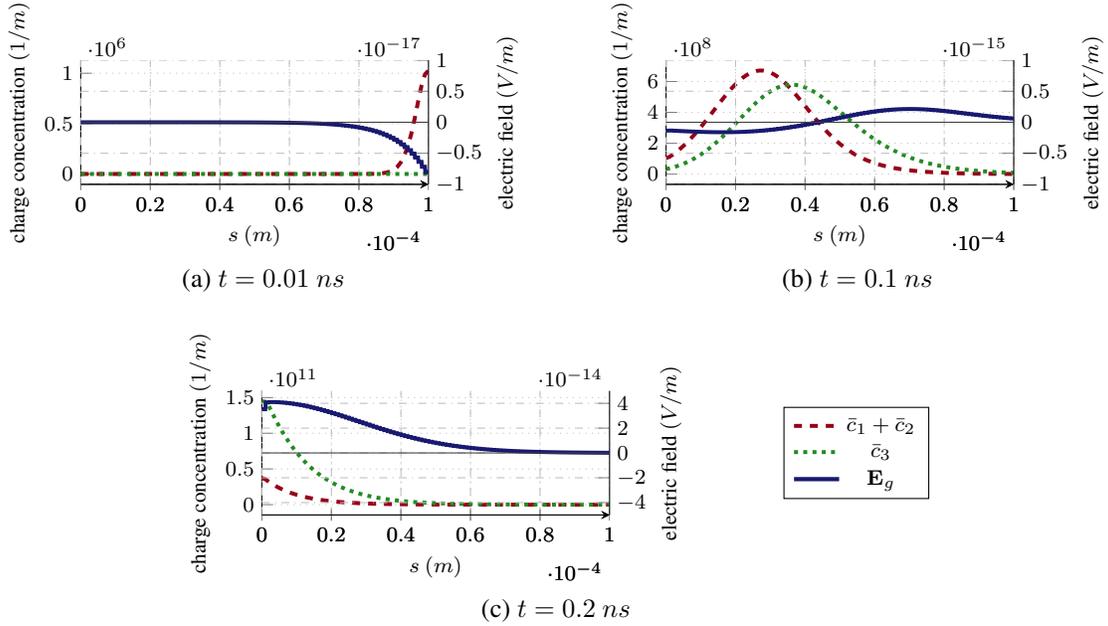

    \subsubsection{TC1 - Convergence of the numerical scheme}
    \label{section:reduction1d:TC1:convergence}
    We consider now a numerical solution, computed on a very fine time grid, with $\Delta t = 10^{-9}\ \mu s$, and study the convergence of the time discretization scheme to this reference solution. We fix the space mesh size to $h=10^{-7} m$ and observe the approximated error in $L^2$-norm at the final time $T=0.2\ ns$, defined as follows:

    \begin{equation*}
        \text{err}_p(\bar{c}) := \dfrac{1}{\|\bar{c}_p^{\text{ref},N}\|}
        \sqrt{h \sum_{i=1}^M \left( \bar{c}_{p,i}^N - \bar{c}_{p,i}^{\text{ref},N}\right)^2},
    \end{equation*}

    \noindent where $N$ denotes the number of time steps, $M$ the number of cells, and $\|\bar{c}_p^{\text{ref},N}\|$ the approximate $L^2$-norm of the reference solution at the final time $T$.

    Since both the Patankar--Euler and Explicit Euler schemes have convergence order 1, and the time-splitting scheme described in Algorithm~\ref{algorithm:time-splitting} is also of order 1, we expect the error to decrease linearly with $\Delta t$. The expected convergence of order 1 is confirmed by the numerical tests represented in Figure~\ref{fig:TC1:convergence} for all the charged species, where numerical solutions obtained for $\Delta t \in \{10^{-5},\,3\cdot 10^{-6},\,10^{-6},\,3\cdot 10^{-7},\,10^{-7}\}\ \mu s$ are compared with the reference solution obtained for $\Delta t = 10^{-9}\ \mu s$. The convergence of the adopted Finite Volume schemes, instead, is discussed in~\cite{crippa2024numericalmethods}.

    \begin{figure}
        \centering
        \subfloat[Concentration of volume electrons \label{fig:TC1:c1:convergence}]{
            \begin{tikzpicture}\scriptsize
                \begin{axis}[
                        scale only axis,
                        width=.2\textwidth,height=.1\textwidth,
                        ymode=log,
                        xmode=log,
                        axis x line=bottom,
                        axis y line=left,
                        ylabel={$\text{err}_1$},
                        xlabel={$\Delta t\ \left(\mu s\right)$},
                        xmax = 3e-5,
                        % legend entries={
                        %     err,
                        %     $y=dt$
                        % },
                        % legend pos=south east,
                        grid=major,
                        grid style={dotted}
                    ]
                    \addplot[
                        mark=o,
                        line width=1.5pt
                    ]
                    table[
                        x=dt,
                        y=errors1,
                        col sep=comma
                    ]{./data/convergence_errors_normalized};
%                    \addplot[
%                        color=darkBlue,
%                        line width=1.5pt
%                    ]
%                    table[
%                        x=dt,
%                        y=errors1,
%                        col sep=comma
%                    ]{./data/strang/convergence_errors_normalized};
                \addplot[
                    domain=3e-9:1e-5,
                    samples=100,
                    line width=1.5pt,
                    dashed,
                    color=darkRed
                ] {x*50};
                \end{axis}
            \end{tikzpicture}
        }\quad
        \subfloat[Concentration of negative ions \label{fig:TC1:c2:convergence}]{
            \begin{tikzpicture}\scriptsize
                \begin{axis}[
                        scale only axis,
                        width=.2\textwidth,height=.1\textwidth,
                        ymode=log,
                        xmode=log,
                        axis x line=bottom,
                        axis y line=left,
                        ylabel={$\text{err}_2$},
                        xlabel={$\Delta t\ \left(\mu s\right)$},
                        xmax = 3e-5,
                        % legend entries={
                        %     err,
                        %     $y=dt$
                        % },
                        % legend pos=south east,
                        grid=major,
                        grid style={dotted}
                    ]
                    \addplot[
                        mark=o,
                        line width=1.5pt
                    ]
                    table[
                        x=dt,
                        y=errors2,
                        col sep=comma
                    ]{./data/convergence_errors_normalized};
%                    \addplot[
%                        color=darkBlue,
%                        line width=1.5pt
%                    ]
%                    table[
%                        x=dt,
%                        y=errors2,
%                        col sep=comma
%                    ]{./data/strang/convergence_errors_normalized};
                \addplot[
                    domain=3e-9:1e-5,
                    samples=100,
                    line width=1.5pt,
                    dashed,
                    darkRed
                ] {x*1e+4};
                \end{axis}
            \end{tikzpicture}
        }\quad
        \subfloat[Concentration of positive ions \label{fig:TC1:c3:convergence}]{
            \begin{tikzpicture}\scriptsize
                \begin{axis}[
                        scale only axis,
                        width=.2\textwidth,height=.1\textwidth,
                        ymode=log,
                        xmode=log,
                        axis x line=bottom,
                        axis y line=left,
                        ylabel={$\text{err}_3$},
                        xlabel={$\Delta t\ \left(\mu s\right)$},
                        xmax = 3e-5,
                        % legend entries={
                        %     err,
                        %     $y=dt$
                        % },
                        % legend pos=south east,
                        grid=major,
                        grid style={dotted}
                    ]
                    \addplot[
                        mark=o,
                        line width=1.5pt
                    ]
                    table[
                        x=dt,
                        y=errors3,
                        col sep=comma
                    ]{./data/convergence_errors_normalized};
%                    \addplot[
%                        color=darkBlue,
%                        line width=1.5pt
%                    ]
%                    table[
%                        x=dt,
%                        y=errors3,
%                        col sep=comma
%                    ]{./data/strang/convergence_errors_normalized};
                \addplot[
                    domain=3e-9:1e-5,
                    samples=100,
                    line width=1.5pt,
                    dashed,
                    darkRed
                ] {x*1e+4};
                \end{axis}
            \end{tikzpicture}
        }\quad
        \hspace{15pt}
        \subfloat[Concentration of surface electrons \label{fig:TC1:ce:convergence}]{
            \begin{tikzpicture}\scriptsize
                \begin{axis}[
                        scale only axis,
                        width=.2\textwidth,height=.1\textwidth,
                        ymode=log,
                        xmode=log,
                        axis x line=bottom,
                        axis y line=left,
                        ylabel={$\text{err}_{\Gamma,e}$},
                        xlabel={$\Delta t\ \left(\mu s\right)$},
                        xmax = 3e-5,
                        % legend entries={
                        %     err,
                        %     $y=dt$
                        % },
                        % legend pos=south east,
                        grid=major,
                        grid style={dotted}
                    ]
                    
                \addplot[
                    mark=o,
                    line width=1.5pt
                ]
                table[
                    x=dt,
                    y=errorse,
                    col sep=comma
                ]{./data/convergence_errors_normalized};
%                    \addplot[
%                        color=darkBlue,
%                        line width=1.5pt
%                    ]
%                    table[
%                        x=dt,
%                        y=errorse,
%                        col sep=comma
%                    ]{./data/strang/convergence_errors_normalized};
                \addplot[
                    domain=3e-9:1e-5,
                    samples=100,
                    line width=1.5pt,
                    dashed,
                    darkRed
                ] {x*1e+4};
            \end{axis}
        \end{tikzpicture}
        }\quad
        \subfloat[Concentration of surface holes \label{fig:TC1:ch:convergence}]{
            \begin{tikzpicture}\scriptsize
                \begin{axis}[
                        scale only axis,
                        width=.2\textwidth,height=.1\textwidth,
                        ymode=log,
                        xmode=log,
                        axis x line=bottom,
                        axis y line=left,
                        ylabel={$\text{err}_{\Gamma,h}$},
                        xlabel={$\Delta t\ \left(\mu s\right)$},
                        xmax = 3e-5,
                        legend entries={
                            error,
%                            second order splitting,
                            $y=\Delta t$
                        },
                        legend style={
                            at={(2,.5)},
                            anchor=east
                        },
                        grid=major,
                        grid style={dotted}
                    ]
                \addplot[
                    mark=o,
                    line width=1.5pt
                ]
                table[
                    x=dt,
                    y=errorsh,
                    col sep=comma
                ]{./data/convergence_errors_normalized};
%                    \addplot[
%                        color=darkBlue,
%                        line width=1.5pt
%                    ]
%                    table[
%                        x=dt,
%                        y=errorsh,
%                        col sep=comma
%                    ]{./data/strang/convergence_errors_normalized};
                \addplot[
                    domain=3e-9:1e-5,
                    samples=100,
                    line width=1.5pt,
                    dashed,
                    darkRed
                ] {x*1e+4};
            \end{axis}
        \end{tikzpicture}
        }
        \caption{\textbf{TC1 -} Normalized $L^2$-error at the final time $T=1 \ ns$, for $\Delta t$ from $10^{-7}\ \mu s$ to $10^{-5}\ \mu s$. The error is computed comparing the approximate charge concentrations with a reference numerical solution obtained with $\Delta t=10^{-9}\ \mu s$.}
        \label{fig:TC1:convergence}
    \end{figure}
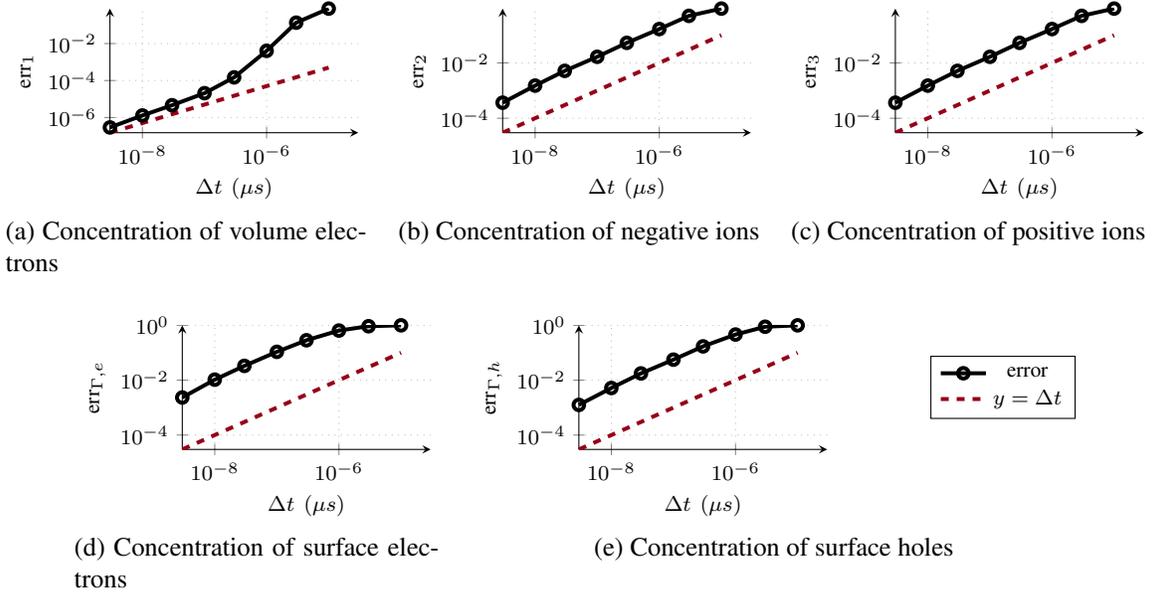

    \subsection{TC2: Domain with a bifurcation}
\label{section:reduction1d:TC2}
We consider now the simple bifurcated domain represented in Figure~\ref{figure:reduction1d:TC2:domain}, and suppose that the tangential component of the electric field in each branch is constant and equal to $E_\Lambda = 10^7 \frac{V}{m}$. Initially there is no charge of any species inside the domain, and we impose a Dirichlet condition on the volume concentration of electrons $\bar{c}_1 = 10^6 \frac{1}{m}$ on one inflow node (point D in Figure~\ref{figure:reduction1d:TC2:domain}).

\begin{figure}
	\centering
	\begin{tikzpicture}
		[decoration={markings, mark= at position 0.5 with {\arrow{stealth}}}]
		\scriptsize
		\filldraw (0,0) circle (2pt) node [xshift = .3cm, yshift = .3cm]{C};
		\filldraw (0,2.5) circle (2pt) node [xshift = .3cm]{B};
		\filldraw (-2.5,0) circle (2pt) node[xshift = .3cm, yshift = .3cm]{A};
		\filldraw (2.5,0) circle (2pt) node (D)[xshift = -.3cm, yshift = .3cm]{D};
		
		\draw [postaction={decorate}, very thick] (-2.5,0) -- (-0,0);
		\draw [postaction={decorate}, very thick] (0,0) -- (2.5,0);
		\draw [postaction={decorate}, very thick] (0,2.5) -- (0,0);
		
		\begin{scope}[overlay]
			\node at (5,1.5) (inflow) {\scriptsize inflow end--node for electrons};
			\draw[->, thin, bend left=25] (inflow.south) to (2.55,.1);
		\end{scope}
	\end{tikzpicture}
	\caption{\textbf{TC2 -} Branched domain with three edges AB, CD, CB, three end nodes A, B, D, and one bifurcation node C. The arrows represent the direction of the tangential electric field on each branch. Non--homogeneous Dirichlet condition on the volume concentration of electrons is imposed on node D.}
	\label{figure:reduction1d:TC2:domain}
\end{figure}
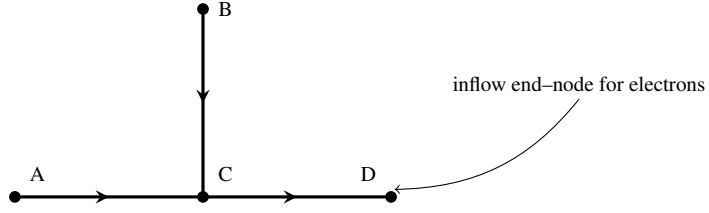

\noindent As discussed in Section~\ref{section:reduction1d:TC1}, the evolution of the charge densities is mainly governed by the chemical reaction of each charged species with the volume electrons, thus the only evident effects of diffusion and transport can be observed for $\bar{c}_1$. Figure~\ref{figure:reduction1d:TC2:c1} displays the behavior of the solution $\bar{c}_1$ in presence of a bifurcation node. The charge flowing into the domain from the inflow node is equally split between the two branches: the edge AB of Figure~\ref{figure:reduction1d:TC2:domain} corresponds to the segment $(0,2\cdot 10^{-4})m$ in Figure~\ref{figure:reduction1d:TC2:AB}, while the vertical segment BC of Figure~\ref{figure:reduction1d:TC2:domain} corresponds to the segment $(0,2\cdot 10^{-4})m$ in Figure~\ref{figure:reduction1d:TC2:BC}. As discussed in~\cite{crippa2024numericalmethods}, this ensures conservation of charge at the junction.

Finally, in Figure~\ref{figure:reduction1d:TC2:c}, we report all the final concentrations of charges in the branched domain at time $T=0.35 \ ns$. As for the test case TC1, discussed in Section~\ref{section:reduction1d:TC1}, the final volume concentration of positive charged species is comparable to the volume concentration of electrons, while the concentration of negative ions is much smaller, due to the small attachment coefficient. The surface concentrations on electrons and ions are the slowest to move, but in this case we observe a higher accumulation of charge on the interface than in TC1, due to the longer domain and time interval. We point out that in a realistic physical application the result would be different, because the modified distribution of charge would influence the tangential electric field and, accordingly, the chemical reaction coefficients, making attachment reactions more likely and introducing a non--negligible recombination effect.

\begin{figure}\centering
	% \subfloat[$t=0.01 \ ns$\label{figure:reduction1d:TC2:c1:0}]
	% {\centering\includegraphics[width=.25\textwidth]{branch_c1_0.png}}
	% \quad
	% \subfloat[$t=0.05 \ ns$\label{figure:reduction1d:TC2:c1:1}]
	% {\centering\includegraphics[width=.25\textwidth]{branch_c1_50.png}}
	% \quad
	% \subfloat[$t=0.12 \ ns$\label{figure:reduction1d:TC2:c1:2}]
	% {\centering\includegraphics[width=.38\textwidth]{branch_c1_120.png}}
	% \quad
	% \caption{\textbf{TC2 -} Concentration of volume electrons at different time steps.}
	\subfloat[Concentration $\bar{c}_1$ on AD\label{figure:reduction1d:TC2:AB}]
	{\begin{tikzpicture}
			\scriptsize
			\begin{axis}[
				xlabel={$s\ \left(m\right)$},
				ylabel={$\bar{c}_1\ \left(\frac{1}{m}\right)$},
				ymax=2.1e+10,
				width=.35\textwidth,height=.2\textwidth,
				grid=major,
				grid style={dotted}
				]
				\addplot+[const plot, black, line width=1.5pt, mark=none]
				table[col sep=comma, x=arc_length, y=Result] {./data/c1_overline.csv};
			\end{axis}
	\end{tikzpicture}}\qquad
	\subfloat[Concentration $\bar{c}_1$ on CB\label{figure:reduction1d:TC2:BC}]
	{\begin{tikzpicture}
			\scriptsize
			\begin{axis}[
				xlabel={$s\ \left(m\right)$},
				ylabel={$\bar{c}_1\ \left(\frac{1}{m}\right)$},
				ymax=2.1e+10,
				width=.35\textwidth,height=.2\textwidth,
				grid=major,
				grid style={dotted}
				]
				\addplot+[const plot, black, line width=1.5pt, mark=none]
				table[col sep=comma, x=arc_length_input_1, y=Result_input_1] {./data/c1_overline.csv};
			\end{axis}
	\end{tikzpicture}}
	\caption{\textbf{TC2 -} Concentration of volume electrons at time $t=0.2 \ ns$ on the branched domain represented in Figure~\ref{figure:reduction1d:TC2:domain}. At the bifurcation point C ($s=2\cdot 10^{-4}$ in Figure~\ref{figure:reduction1d:TC2:AB}) the charge density, flowing from the end node D ($s=4\cdot 10^{-4}$ in Figure~\ref{figure:reduction1d:TC2:AB}), is equally distributed on the two branches AB ($s\in(0,2\cdot 10^{-4})$ in Figure~\ref{figure:reduction1d:TC2:AB}) and BC ($s\in(0,2\cdot 10^{-4})$ in Figure~\ref{figure:reduction1d:TC2:BC}).}
	\label{figure:reduction1d:TC2:c1}
\end{figure}
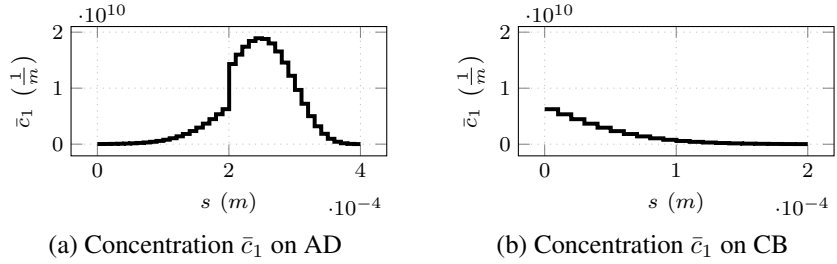

\begin{figure}\centering
	\subfloat[$\bar{c}_1$\label{figure:reduction1d:TC2:c:c1}]
	{\centering\includegraphics[width=.2\textwidth]{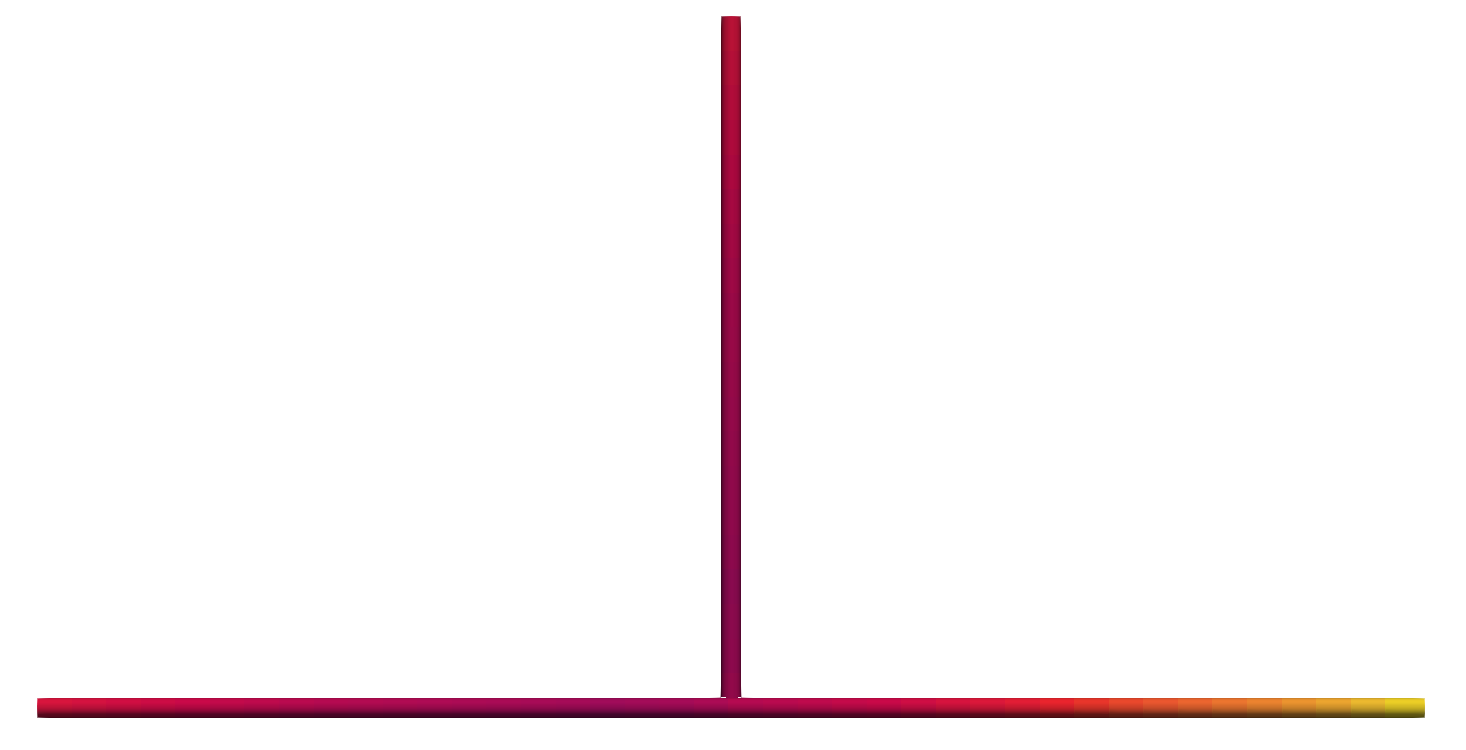}}
	\quad
	\subfloat[$\bar{c}_2$\label{figure:reduction1d:TC2:c:c2}]
	{\centering\includegraphics[width=.2\textwidth]{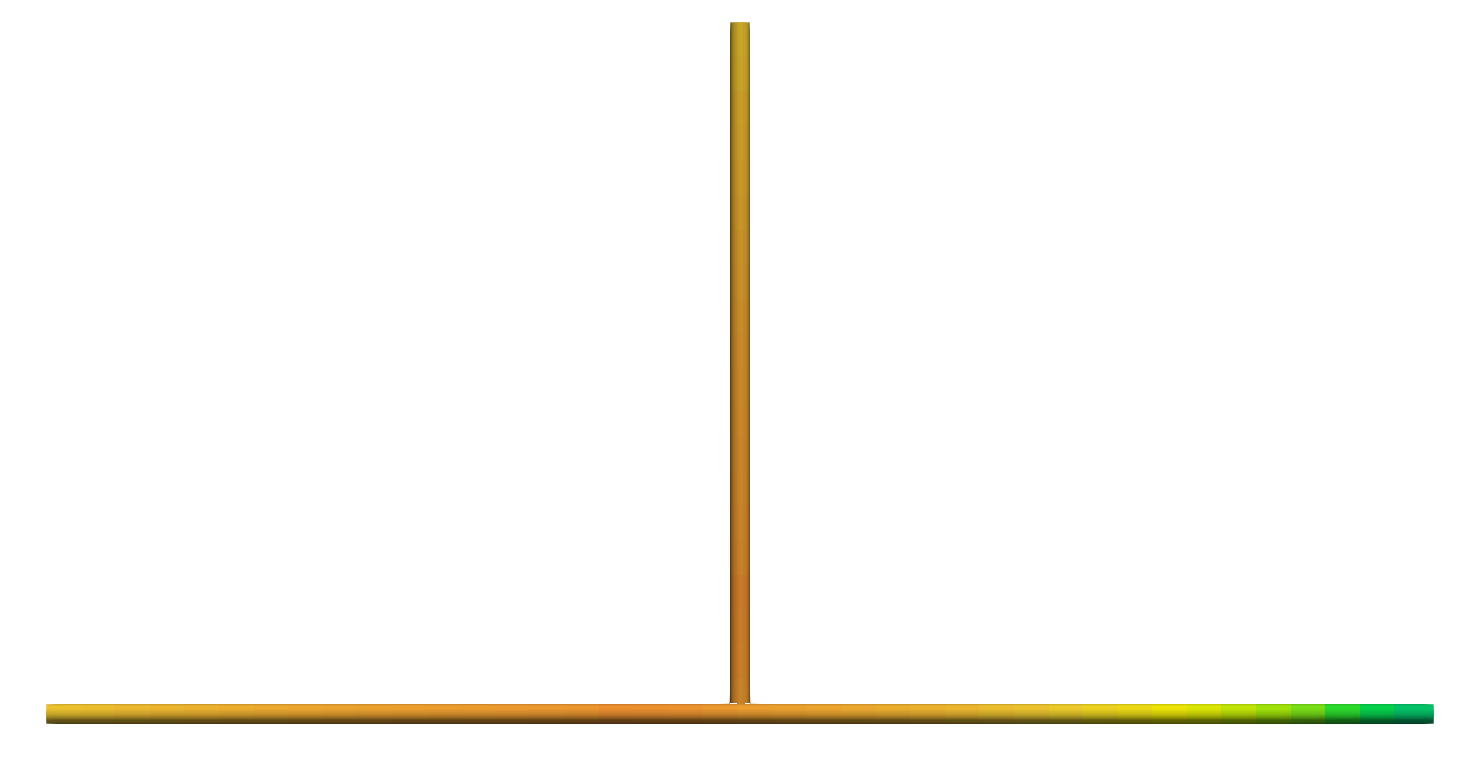}}
	\quad
	\subfloat[$\bar{c}_3$\label{figure:reduction1d:TC2:c:c3}]
	{\centering\includegraphics[width=.3\textwidth]{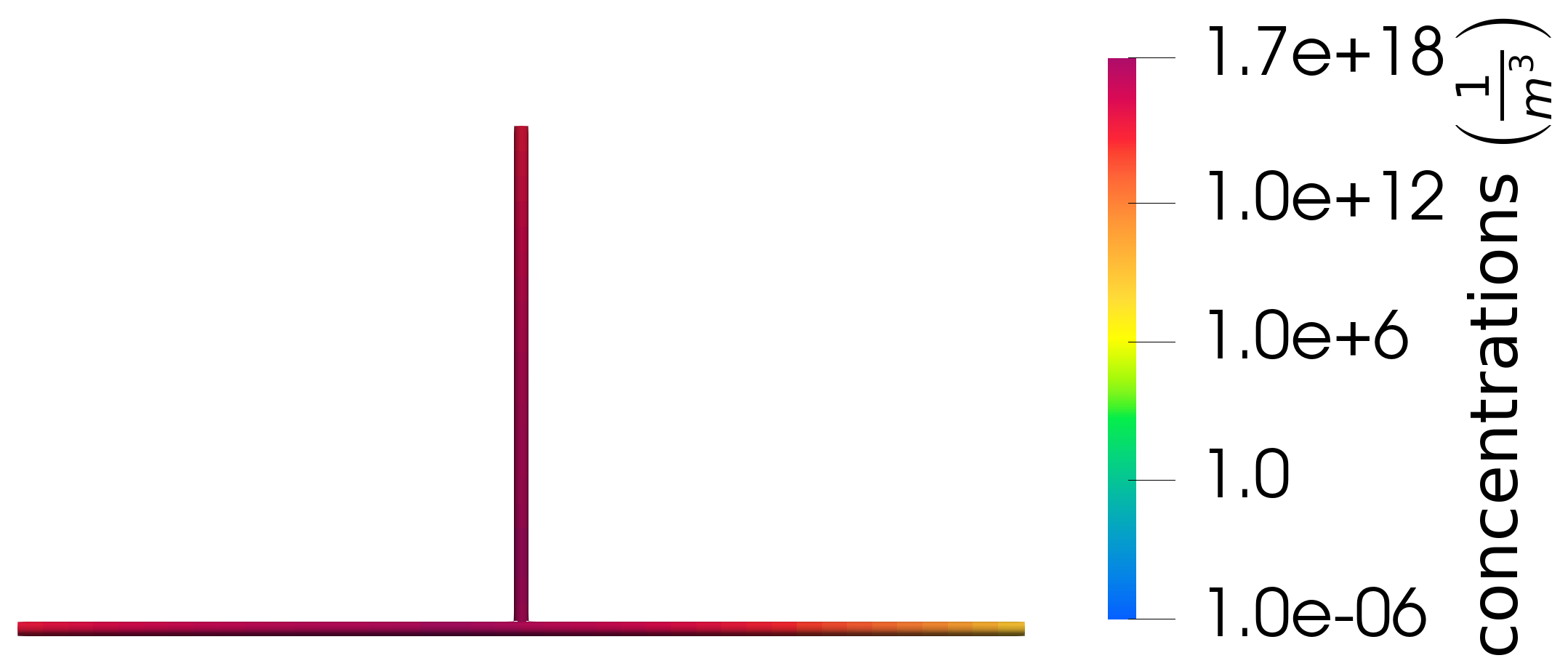}}\\
	\subfloat[$\bar{c}_{\Gamma,e}$\label{figure:reduction1d:TC2:c:ce}]
	{\centering\includegraphics[width=.2\textwidth]{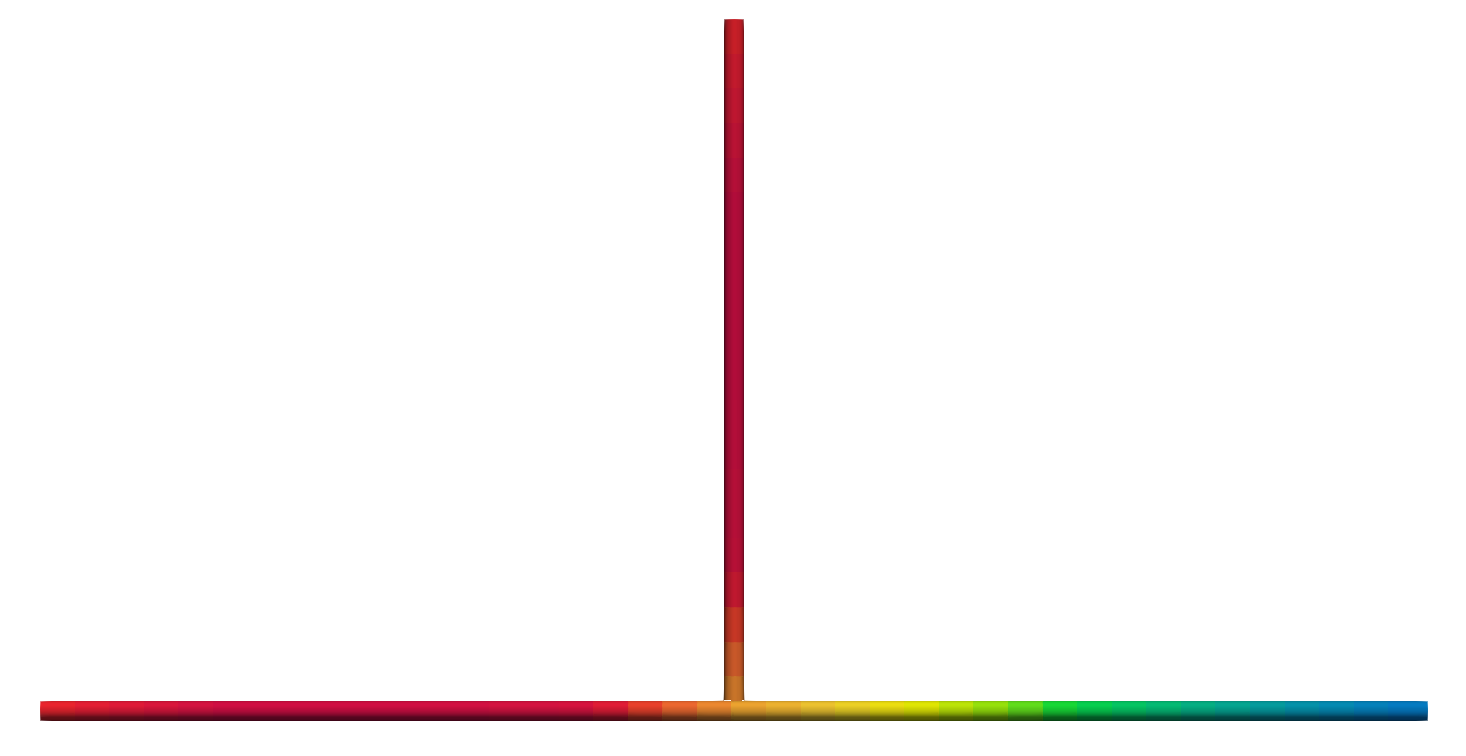}}
	\quad
	\subfloat[$\bar{c}_{\Gamma,h}$\label{figure:reduction1d:TC2:c:ch}]
	{\centering\includegraphics[width=.3\textwidth]{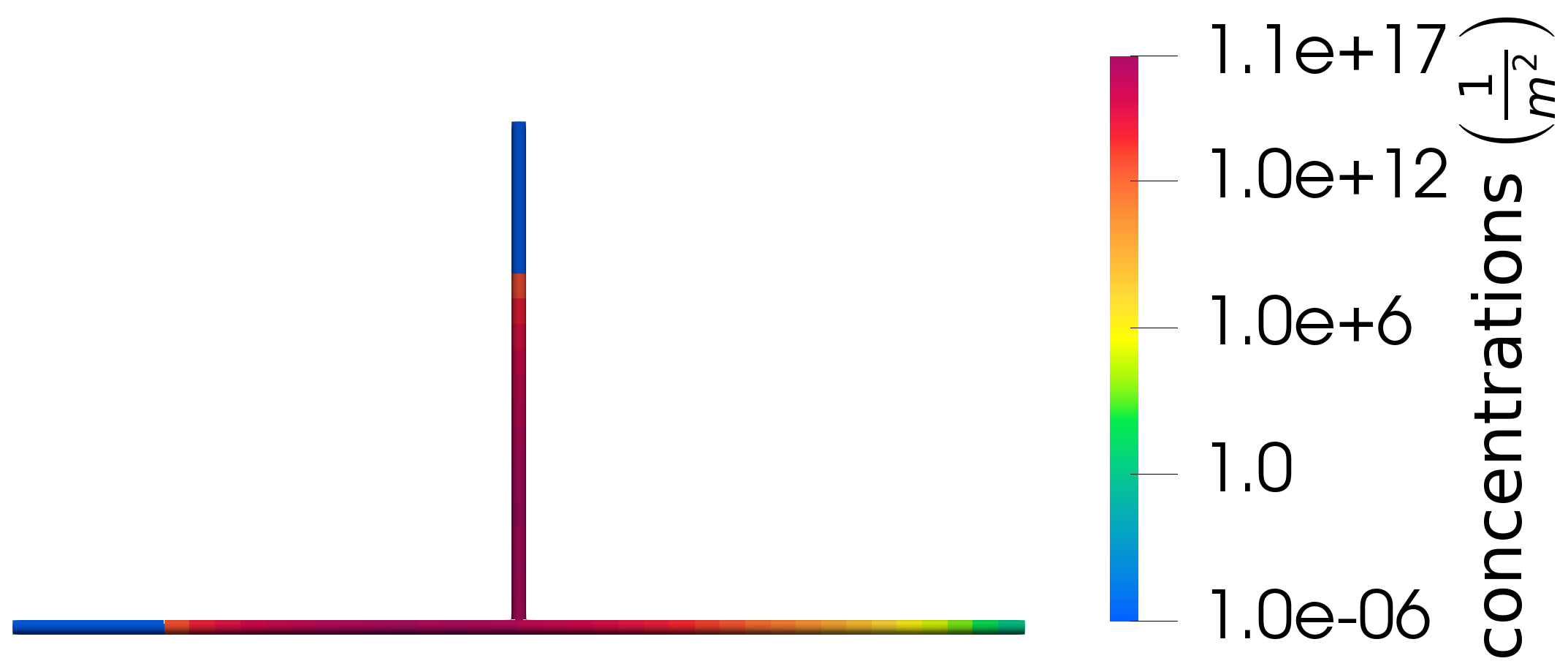}}
	\caption{\textbf{TC2 -} Final concentrations of charged particles at time $T=0.35 \ ns$.}
	\label{figure:reduction1d:TC2:c}
\end{figure}

\subsection{TC3: Electrical Treeing}
\label{section:reduction1d:TC3}
We now apply the reduced 1D model to simulate the movement of charged particles in a realistic geometry of an electrical treeing, under the effect of a constant electric field $E_\Lambda = 10^7 \frac{V}{m}$, directed along each branch of the graph. The treeing geometry was extracted as the skeleton of a 3D image detected via X--ray computed tomography~\cite{schurch2014imaging,schurch20193d,schurch2015comparison}, and it is made of 12544 segments. We set a Dirichlet boundary condition $\bar{c}_1 = 10^6 \frac{1}{m}$ for the volume electrons at the inflow, that is on the top node of the domain and simulate the phenomenon for $T=0.4\ ns$. We point out that this test case, although performed on the reduced geometry of an experimentally detected electrical treeing, is not completely physical: indeed, the influence of the increasing charge density on both the internal and external electric field should be taken into account by a coupling with the electrostatic problem. However, this coupling is beyond the scope of this work, where we only aim at simulating the evolution of charge densities under the effect of a constant electric field. We address the problem of coupling the movement of charges with the electrostatic model in~\cite{crippa2024mixed}.

In Figure~\ref{figure:reduction1d:TC3:c} we can observe that, under a constant longitudinal electric field, the concentration of charges is spread across the whole complex ramified treeing geometry. At each bifurcation node, the volume electron concentration is equally distributed among the outgoing edges, as discussed in Section~\ref{section:reduction1d:TC2} on a simple domain with a single bifurcation node. However, this effect is not clearly observable in Figure~\ref{figure:reduction1d:TC3:c} because of the dominant positive chemical reactions.\\
Because of the huge gap between the values of the ionization and attachment coefficients, we observe that the concentration of positive ions exceeds greatly that of the negative ions, both in the gas volume and on the dielectric surface. Moreover, the electron mobility is higher than the mobility of the positive ions, and it causes the concentration of positive ions to be higher than the volume concentration of electrons.
In more realistic applications, the increasing charge concentration would generate an intense radial electric field, that influences both the chemical reaction rates and the value of the potential in the gas. This process would end up with the decrease of the longitudinal electric field and a higher attachment coefficient, stopping the production of positive charges and promoting the production of negative ions.

\begin{figure}\centering
	\subfloat[$\bar{c}_1$\label{figure:reduction1d:TC3:c:c1}]
	{\centering\includegraphics[width=.23\textwidth]{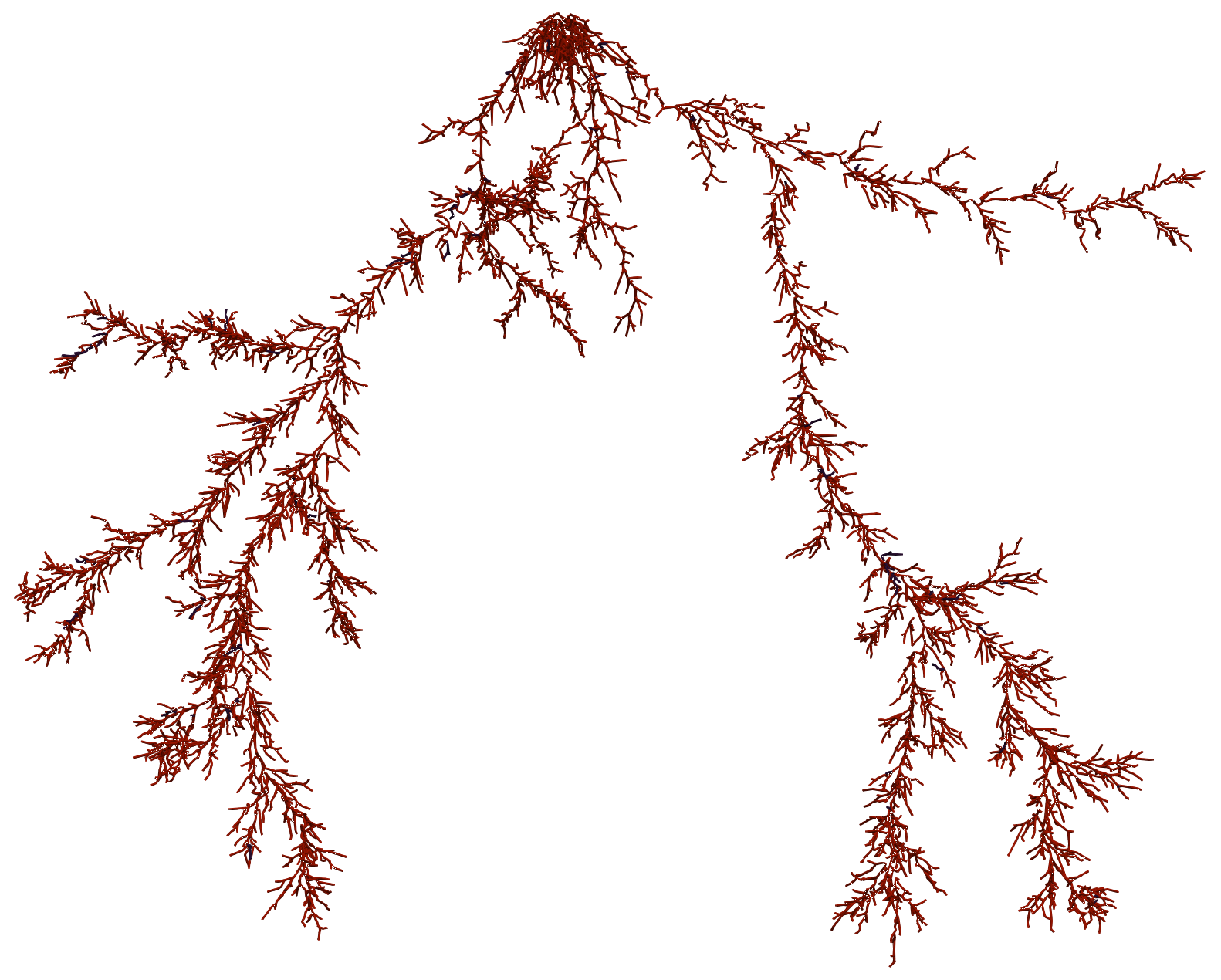}}
	\quad
	\subfloat[$\bar{c}_2$\label{figure:reduction1d:TC3:c:c2}]
	{\centering\includegraphics[width=.23\textwidth]{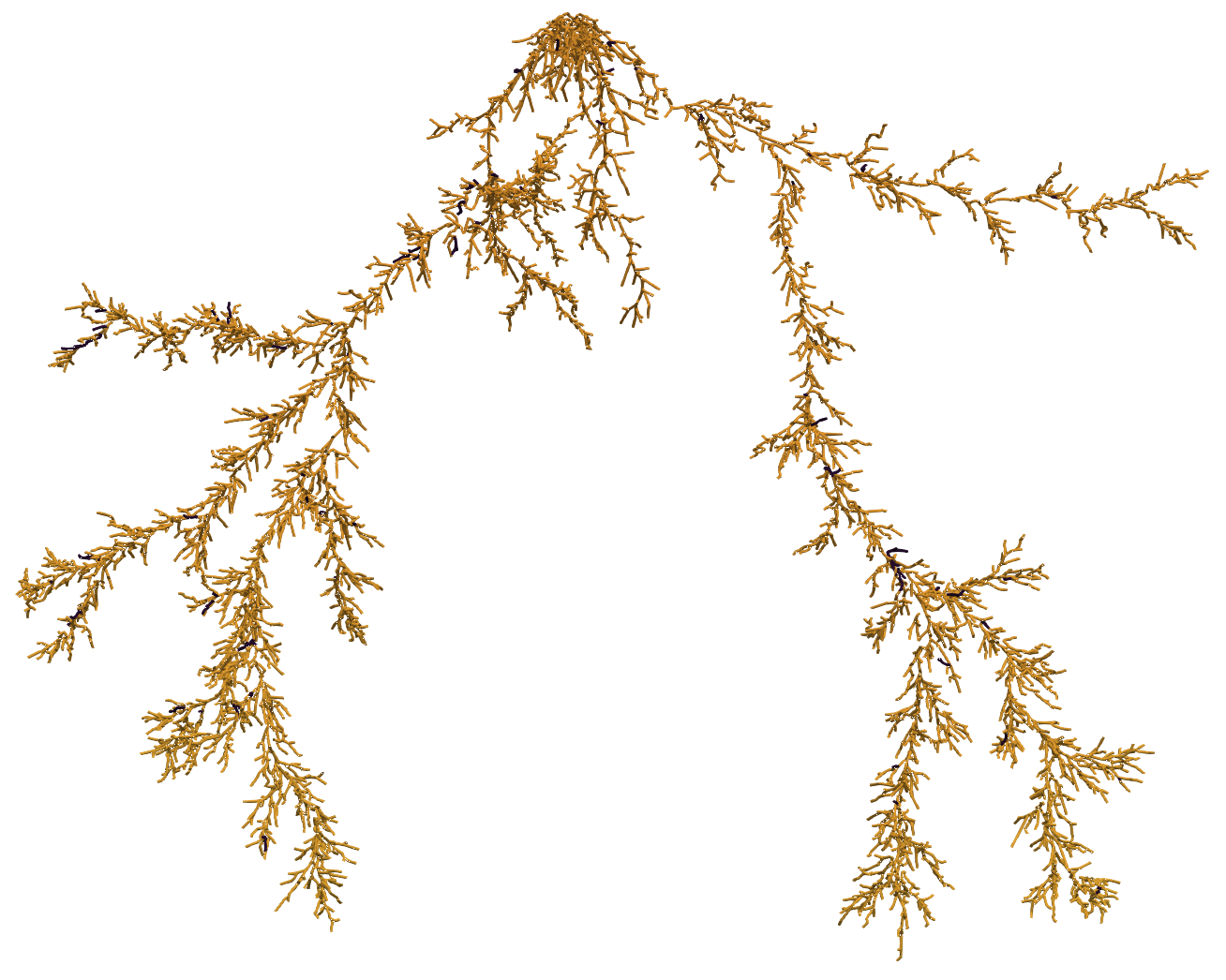}}
	\quad
	\subfloat[$\bar{c}_3$\label{figure:reduction1d:TC3:c:c3}]
	{\centering\includegraphics[width=.35\textwidth]{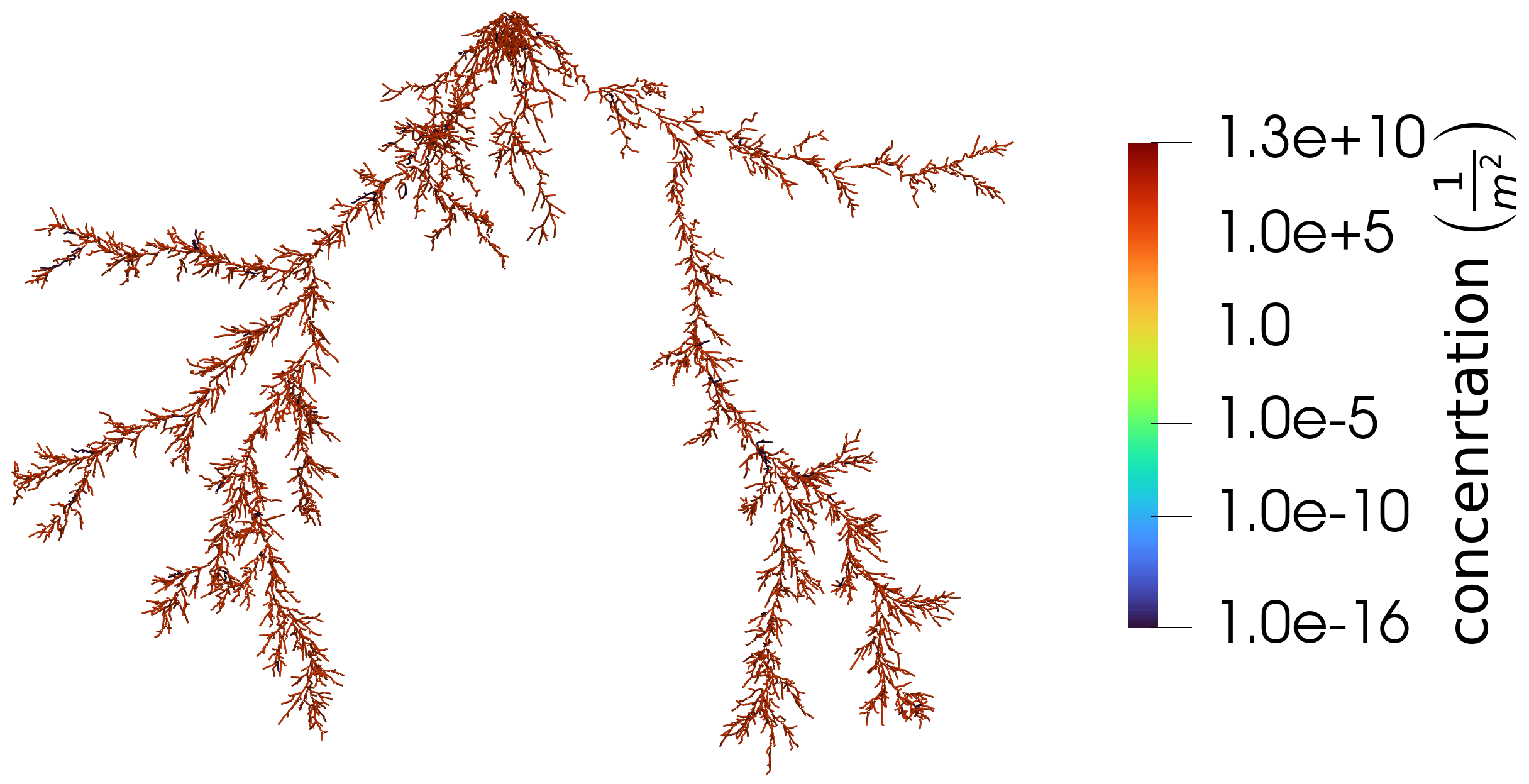}}\\
	\subfloat[$\bar{c}_{\Gamma,e}$\label{figure:reduction1d:TC3:c:ce}]
	{\centering\includegraphics[width=.23\textwidth]{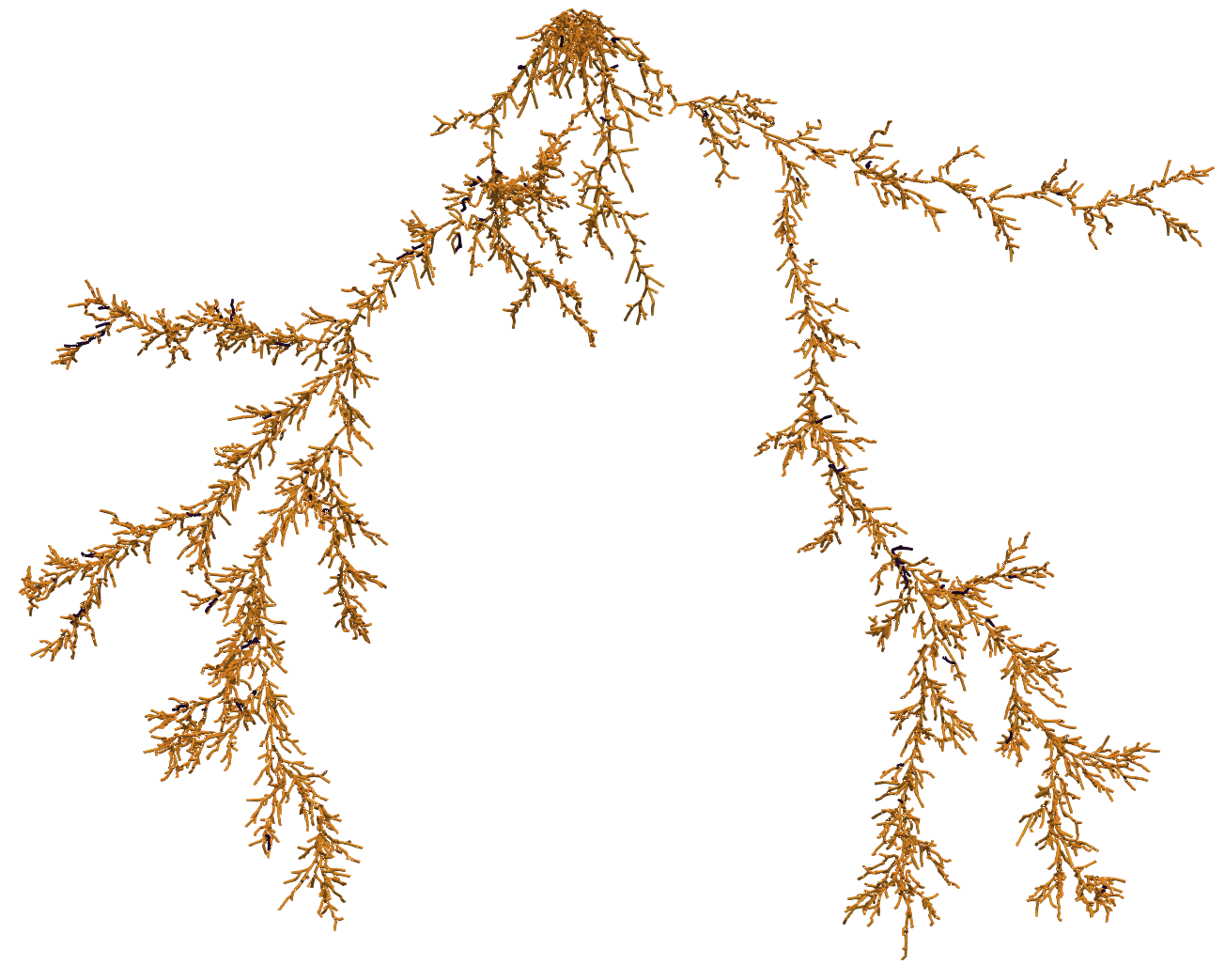}}
	\quad
	\subfloat[$\bar{c}_{\Gamma,h}$\label{figure:reduction1d:TC3:c:ch}]
	{\centering\includegraphics[width=.35\textwidth]{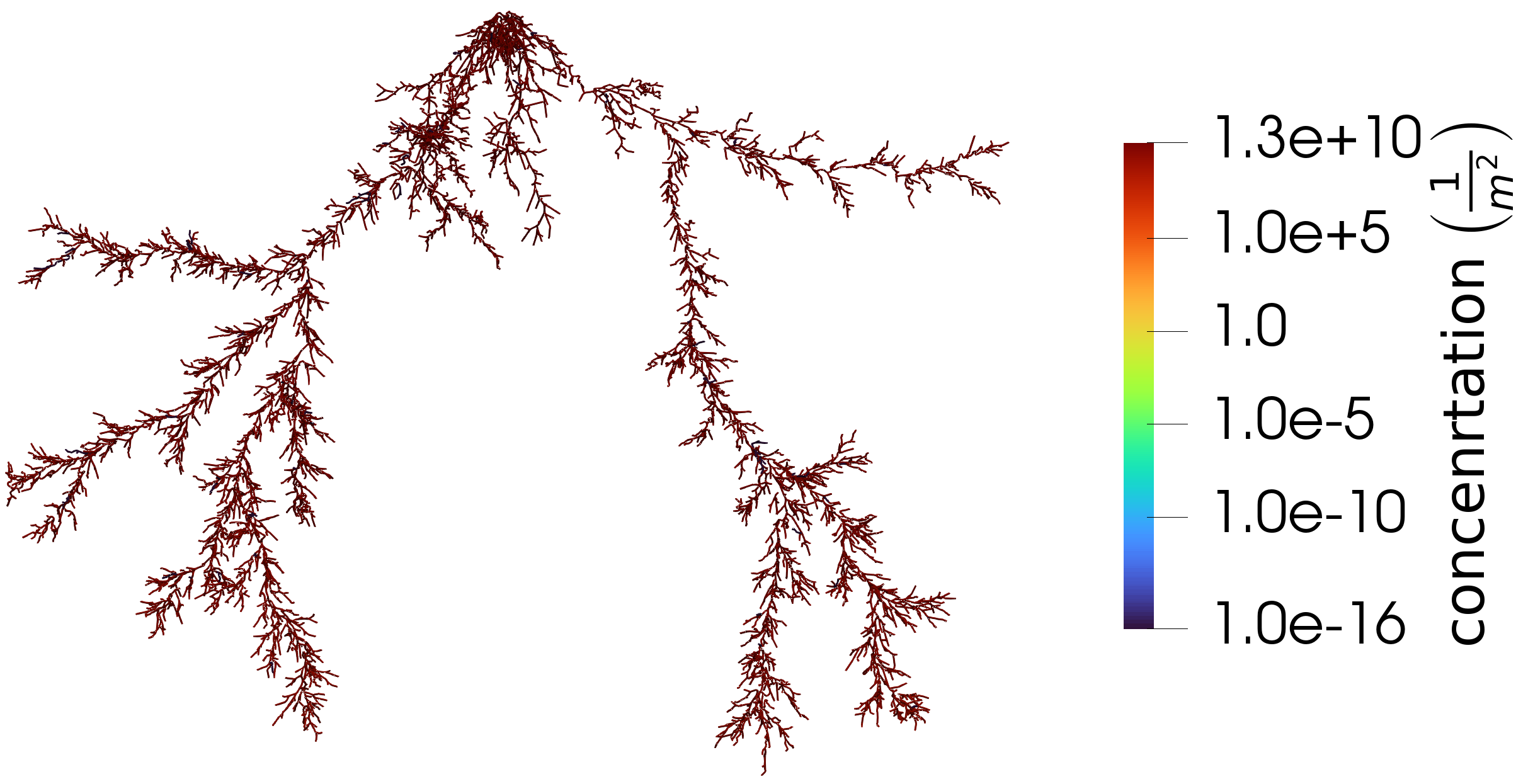}}
	\caption{\textbf{TC3 -} Final concentrations of charged particles at time $T=0.4 \ ns$.}
	\label{figure:reduction1d:TC3:c}
\end{figure}
	
	\section{Conclusions}
	\label{section:reduction1d:conclusions}
	We have reduced the advection--diffusion--reaction problem of charge concentrations in the electrical treeing to a system of equations defined on a 1D graph representing its skeleton. The resulting model takes into account the evolution of volume concentration of electrons and positive and negative ions, due to their movement under the effect of an electric field and the interactions among different charged species via chemical reactions. Moreover, we accounted for the exchange of electrons with the dielectric surface, also collapsed onto the 1D graph. This model complements the mixed--dimensional electrostatic problem derived in~\cite{crippa2024mixed}, possibly enabling in the future fast simulations of partial discharges in the electrical treeing.\\
	Moreover, we introduced an efficient evaluation of the transverse electric field, due to the volume and surface charge distributions, which overcomes the need to solve 2D Poisson problems on cross--sections and allows to discretize the model entirely on a 1D mesh. This representation of the transverse electric field introduces some non--linear reaction terms in the reduced equations, describing the flux of charges through the dielectric surface. However, we showed that the key physical properties of the original 3D problem, namely conservation of the total net charge and monotonicity of the solution are still satisfied by the reduced model.
	
	We assumed that the charge concentrations can be approximated as constants on cross--sections of the gas domain, which is in agreement with the hypothesis of the mixed--dimensional electrostatic model~\cite{crippa2024mixed}, but neglecting the distribution of charge has a strong impact on the produced electric field. To take this into account, we derived ODEs describing the evolution of two orthogonal components of the dipole moment.

    To ensure that the relevant physical properties of the problem are also satisfied by the discrete solutions, we adopted a consistent conservative Finite Volume scheme to solve the drift--diffusion equations and a monotone Patankar--Euler approach for time discretization. The considered time scheme involves an explicit treatment of nonlinear reaction terms, which makes it efficient.
    Improvements in the accuracy of the numerical solution can be gained by adopting higher--order schemes, both for the space discretization of the drift--diffusion equations, as pointed out in~\cite{crippa2024numericalmethods}, and for the time integration of the chemistry ODEs, employing for instance higher--order conservative Patankar schemes, as discussed in~\cite{formaggia2011positivity}. Then, the introduction of higher--order splitting schemes would also improve the approximation capability of the method.

    In this work we have fixed the longitudinal electric field in the gas to a known quantity, discarding the effect that the evolution of charge concentrations in the defect has on the electrostatic problem. This only enables the study of the initial stages of the discharges. However, the proposed discretization is compatible with the numerical approach proposed in~\cite{crippa2024mixed} for the 3D--1D electrostatic problem, that complements the description of the electrical treeing phenomenon. For this coupling we will employ more advanced time discretization schemes, as the one presented in~\cite{villa2017efficient}.
    The present work provides a rigorous derivation of a reduced plasma model for the electrical treeing, allowing efficient simulations on complex branched geometries that would otherwise difficult to mesh in 3D. A future coupling with the 3D--1D electrostatic problem and with more complex chemical models~\cite{villa2017simulation} can lead to fast simulations of partial discharges in the electrical treeing.
    
    \section{Acknowledgments}
    This work has been financed by the Research Found for the Italian Electrical System under the Contract Agreement between RSE and the Ministry of Economic Development. 
    B.C. acknowledges the support by the Swedish Research Council under grant no. 2021-06594, during the stay at Institut Mittag-Leffler in Djursholm, Sweden, in 2025.
    The mesh of the treeing structure is courtesy of Prof. R. Schurch, Universidad Técnica Federico Santa Marìa Valparaìso, Chile.
  
	\newpage
	\printbibliography[
	heading=bibintoc,
	title={References}
	]
	
	\newpage
	\appendix
    \section{Exact solutions of superimposition of effects problems}
    \label{appendix:exact_sol_superimposition_of_effects}
    In this section we report the detailed computations of the exact solutions of the systems considered in the superimposition of effects for the electric field in Section~\ref{section:reduction1d:superimposition_of_effects}.
    
    \subsection{Effect of the external electric field}
    \label{appendix:exact_sol_suprimposition_of_effects:Eext}
    We start by the effect of the $x$ component of the electric field~\eqref{eq:reduction1d:superimposition:x}, where we substitute the vector unknown $\bm{\Psi}^1$ with the opposite of the gradient of the corresponding potential $\tilde{\Phi}$, as in equation~\eqref{eq:reduction1d:electric_field:3d:3:adim}:
    
    \begin{equation}
    	\begin{cases}
    		\tilde{\Delta}_{\tilde{\mathcal{D}}} \tilde{\Phi}_g = 0, & \text{in}\ \tilde{\mathcal{D}}, \\
    		\tilde{\Delta}_{\tilde{\mathcal{D}}} \left(\epsilon_s \tilde{\Phi}_s \right) = 0, & \text{in}\ \tilde{\mathcal{D}}_s, \\
    		\tilde{\nabla}_{\tilde{\mathcal{D}}} \tilde{\Phi}_g\cdot\mathbf{n}_g + 
    		\epsilon_s \tilde{\nabla}_{\tilde{\mathcal{D}}} \tilde{\Phi}_s \cdot\mathbf{n}_s = 0, & \text{on}\ \partial\tilde{\mathcal{D}}, \\
    		-\tilde{\nabla}_{\tilde{\mathcal{D}}} \tilde{\Phi}_s \cdot \mathbf{n} = 
    		\cos\tilde{\theta}, & \text{on}\ \partial(\tilde{\mathcal{D}}\cup\tilde{\mathcal{D}}_s).
    	\end{cases}
    	\label{eq:reduction1d:superimposition:Eext:primal}
    \end{equation}
    
    \noindent We can explicitly rewrite problem~\eqref{eq:reduction1d:superimposition:Eext:primal} in cylindrical coordinates as follows:
    
    \begin{subnumcases}{}
    	-\dfrac{1}{\tilde{r}}\dfrac{\partial}{\partial \tilde{r}}\left( \tilde{r} \dfrac{\partial\tilde{\Phi}_g}{\partial \tilde{r}} \right) + \dfrac{1}{\tilde{r}^2}\dfrac{\partial^2\tilde{\Phi}_g}{\partial \tilde{\theta}^2} = 0, & $\text{in}\ \tilde{\mathcal{D}},$ 
    	\label{eq:reduction1d:superimposition:Eext:cylindrical:1}\\
    	-\dfrac{1}{\tilde{r}}\dfrac{\partial}{\partial \tilde{r}}\left( \tilde{r} \dfrac{\partial\tilde{\Phi}_s}{\partial \tilde{r}} \right) + \dfrac{1}{\tilde{r}^2}\dfrac{\partial^2\tilde{\Phi}_s}{\partial \tilde{\theta}^2} = 0, & $\text{in}\ \tilde{\mathcal{D}}_s,$ 
    	\label{eq:reduction1d:superimposition:Eext:cylindrical:2}\\
    	\dfrac{\partial \tilde{\Phi}_g}{\partial \tilde{r}} = \epsilon_s\dfrac{\partial \tilde{\Phi}_s}{\partial \tilde{r}}, & $\text{on}\ \partial\tilde{\mathcal{D}}, $
    	\label{eq:reduction1d:superimposition:Eext:cylindrical:3}\\
    	\dfrac{\partial\tilde{\Phi}_s}{\partial \tilde{r}} = -\cos\tilde{\theta}, & $\text{on}\ \partial(\tilde{\mathcal{D}}\cup\tilde{\mathcal{D}}_s).$
    	\label{eq:reduction1d:superimposition:Eext:cylindrical:4}
    \end{subnumcases}
    
    \noindent We can solve equations~\eqref{eq:reduction1d:superimposition:Eext:cylindrical:1} and~\eqref{eq:reduction1d:superimposition:Eext:cylindrical:2} by separation of variables: assuming that $\tilde{\Phi}(\tilde{r},\tilde{\theta}) = Q(\tilde{r})P(\tilde{\theta})$, we can find $k\in\mathbb{Q}$ such that
    
    \begin{equation*}
    	\dfrac{1}{Q}\tilde{r}\dfrac{d}{d\tilde{r}}\left( \tilde{r} \dfrac{d Q}{d \tilde{r}} \right) = - \dfrac{1}{P}\dfrac{d^2 P}{d \tilde{\theta}^2} = k.
    \end{equation*}
    
    \noindent We require both $Q$ and $P$ to be bounded functions, and that $P$ is periodic of period $2\pi$. Separating the two equations and applying the change of variables $\tilde{\rho} = \log(\tilde{r})$, we can rewrite the problem as:
    
    \begin{equation*}
    	\begin{dcases}
    		-\dfrac{d^2 Q^{(k)}}{d \tilde{\rho}^2} = kQ^{(k)}, \\
    		- \dfrac{d^2 P^{(k)}}{d \tilde{\theta}^2} = kP^{(k)}.
    	\end{dcases}
    \end{equation*}
    
    \noindent In order for $P^{(k)}$ to be $2\pi$-periodic, we need to consider $k\in\mathbb{N}$. Then the solutions to these equations are $Q^{(k)}(\tilde{r}) = \alpha_1^{(k)} \tilde{r}^k + \alpha_2^{(k)} \tilde{r}^{-k}$ and $P^{(k)}(\tilde{\theta}) = \beta_1^{(k)} \cos(k\tilde{\theta}) + \beta_2^{(k)} \sin(k\tilde{\theta})$. Moreover, we have a bounded potential for $\tilde{r}\to 0$ only if we set $\alpha_2^{(k)}=0\ \forall k$. The candidate solutions corresponding to each coefficient $k\in\mathbb{N}$ are therefore given by:
    
    \begin{align*}
    	\tilde{\Phi}_s^{(k)}(\tilde{r},\tilde{\theta}) &= 
    	Q_s^{(k)}(\tilde{r}) P_s^{(k)}(\tilde{\theta}) =
    	\tilde{r}^k \left( c_{1,s}^{(k)} \cos(k\tilde{\theta}) + c_{2,s}^{(k)} \sin(k\tilde{\theta}) \right),\\
    	\tilde{\Phi}_g^{(k)}(\tilde{r},\tilde{\theta}) &= 
    	Q_g^{(k)}(\tilde{r}) P_g^{(k)}(\tilde{\theta}) =
    	\tilde{r}^k \left( c_{1,g}^{(k)} \cos(k\tilde{\theta}) + c_{2,g}^{(k)} \sin(k\tilde{\theta}) \right).
    \end{align*}
    
    % \begin{align*}
    	%     \tilde{\Phi}_s(\tilde{r},\tilde{\theta}) = \sum_{k=1}^{+\infty} Q_s^{(k)}(\tilde{r}) P_s^{(k)}(\tilde{\theta}) =
    	%     \sum_{k=1}^{+\infty} \tilde{r}^k \left( c_{1,s}^{(k)} \cos(k\tilde{\theta}) + c_{2,s}^{(k)} \sin(k\tilde{\theta}) \right), \\
    	%     \tilde{\Phi}_g(\tilde{r},\tilde{\theta}) = \sum_{k=1}^{+\infty} Q_g^{(k)}(\tilde{r}) P_g^{(k)}(\tilde{\theta}) =
    	%     \sum_{k=1}^{+\infty} \tilde{r}^k \left( c_{1,g}^{(k)} \cos(k\tilde{\theta}) + c_{2,g}^{(k)} \sin(k\tilde{\theta}) \right).
    	% \end{align*}
    
    \noindent We compute the values of the constants $c_{i,s}^{(k)}$, $i=1,2$, by applying the boundary condition~\eqref{eq:reduction1d:superimposition:Eext:cylindrical:4}:
    
    \begin{equation*}
    	-\cos(\tilde{\theta}) = \dfrac{\partial\tilde{\Phi}_s^{(k)}}{\partial \tilde{r}}\big|_{\tilde{r} = \tilde{R}_s} =
    	k \tilde{R}_s^{k-1} \left( c_{1,s}^{(k)} \cos(k\tilde{\theta}) + c_{2,s}^{(k)} \sin(k\tilde{\theta}) \right).
    \end{equation*}
    
    \noindent This equation is only satisfied for $k=1$, with coefficients
    %$c_{1,s}^{(k)} = 0, \ \forall k\in\mathbb{N}\setminus\{1\}$, $c_{2,s}^{(k)} = 0,\ \forall k\in\mathbb{N}$,
    $c_{1,s}^{(1)} = -1$ and $c_{2,s}^{(1)} = 0$. Therefore, $\tilde{\Phi}_s(\tilde{r},\tilde{\theta}) = -\tilde{r}\cos(\tilde{\theta})$. Similarly, we compute $c_{i,g}^{(k)}$, $i=1,2$, by applying the interface condition~\eqref{eq:reduction1d:superimposition:Eext:cylindrical:3}:
    
    \begin{equation*}
    	-\epsilon_s\cos(\tilde{\theta}) =
    	\epsilon_s \dfrac{\partial\tilde{\Phi}_s}{\partial \tilde{r}}\big|_{\tilde{r} = 1}
    	= \dfrac{\partial\tilde{\Phi}_g^{(k)}}{\partial \tilde{r}}\big|_{\tilde{r} = 1}
    	= k \left( c_{1,g}^{(k)} \cos(k\tilde{\theta}) + c_{2,g}^{(k)} \sin(k\tilde{\theta}) \right),
    \end{equation*}
    
    %\noindent and obtain $c_{1,g}^{(k)} = 0, \ \forall k\in\mathbb{N}\setminus\{1\}$, $c_{2,g}^{(k)} = 0,\ \forall k\in\mathbb{N}$, $c_{1,g}^{(1)} = -\epsilon_s$.
    \noindent which holds only if $k=1$, with coefficients $c_{1,g}^{(1)} = -\epsilon_s$ and $c_{2,g}^{(1)} = 0$.
    Thus, $\tilde{\Phi}_g(\tilde{r},\tilde{\theta}) = -\epsilon_s \tilde{r}\cos(\tilde{\theta})$. The potential in cartesian coordinates is given by:
    
    \begin{align*}
    	\tilde{\Phi}_g(\tilde{x},\tilde{y},\tilde{z}) = -\tilde{x}\epsilon_s, & \quad\text{in}\ \Omega_g, \\
    	\tilde{\Phi}_s(\tilde{x},\tilde{y},\tilde{z}) = -\tilde{x}, & \quad\text{in}\ \Omega_s,
    \end{align*}
    
    \noindent and the vector solution $\bm{\Psi}$ by the opposite of its gradient:
    
    \begin{align*}
    	\bm{\Psi}_g^1 = \epsilon_s\mathbf{x}, &\quad\text{in} \ \Omega_g \\
    	\bm{\Psi}_s^1 = \mathbf{x}, &\quad\text{in}\ \Omega_s.
    \end{align*}
    
    \noindent An analogous reasoning can be applied to the solution to problem~\eqref{eq:reduction1d:superimposition:y}, describing the effect of the $y$ component of the electric field, with exact solution
    
    \begin{align*}
    	\bm{\Psi}_g^2 = \epsilon_s\mathbf{y}, &\quad\text{in} \ \Omega_g \\
    	\bm{\Psi}_s^2 = \mathbf{y}, &\quad\text{in}\ \Omega_s.
    \end{align*}
    
    \subsection{Effect of the volume charge concentration}
    \label{appendix:exact_sol_superimpostion_of_effects:q}
    For the solution of problem~\eqref{eq:reduction1d:superimposition:q} we apply the divergence theorem to equations~\eqref{eq:reduction1d:electric_field:3d:1:adim} and~\eqref{eq:reduction1d:electric_field:3d:2:adim}. First, integrate equation~\eqref{eq:reduction1d:electric_field:3d:1:adim} over a circle $\mathcal{W}_g\subset\tilde{\mathcal{D}}$, concentric to $\tilde{\mathcal{D}}$, with radius $\rho\leq1$. Thanks to the divergence theorem, we obtain:
    
    \begin{equation}
    	\int_{\partial\mathcal{W}_g} \bm{\Psi}_g^5\cdot\mathbf{n} = \int_{\mathcal{W}_g} \tilde{\nabla}_{\tilde{\mathcal{D}}} \cdot\bm{\Psi}_g^5 = \int_{\mathcal{W}_g} 1 = |\mathcal{W}_g|,
    	\label{eq:reduction1d:superimposition:q:divergence}
    \end{equation}
    
    \noindent where $|\mathcal{W}_g| = \pi \rho^2$ denotes the area of $\mathcal{W}_g$. Because of the symmetry of the domain and the homogeneous charge distribution in the section, the field represented by $\bm{\Psi}_g^5$ is radial and uniform, i.e. $\bm{\Psi}_g^5 = |\bm{\Psi}_g^5|(\tilde{r})\mathbf{n}$, $\tilde{r}\in[0,1]$. As a consequence:
    
    \begin{equation*}
    	\int_{\partial\mathcal{W}_g} \bm{\Psi}_g^5\cdot\mathbf{n} = 2\pi \rho |\bm{\Psi}_g^5|(\rho).
    \end{equation*}
    
    \noindent Substituting in equation~\eqref{eq:reduction1d:superimposition:q:divergence}, we obtain:
    
    \begin{equation*}
    	\bm{\Psi}_g^5 = \dfrac{\tilde{r}}{2}\mathbf{n},\ \tilde{r}\in[0,1].
    \end{equation*}
    
    \noindent Integrate now on a circle $\mathcal{W}_s$ such that $\tilde{\mathcal{D}}\subset\mathcal{W}_s\subset\tilde{\mathcal{D}}_s$, concentric to $\tilde{\mathcal{D}}$, with different radius $\rho\in[1,\tilde{R}_s]$. Here the total charge concentration is 0 outside $\tilde{\mathcal{D}}$ and the field $\bm{\Psi}_s^5$ is radial and uniform in $\tilde{\mathcal{D}}_s$, due to the homogeneous charge distribution. With the same passages as before, applying the divergence theorem on $\mathcal{W}_s\setminus\tilde{\mathcal{D}}$, we obtain:
    
    \begin{equation*}
    	\bm{\Psi}_s^5 = \dfrac{1}{2\tilde{r}\epsilon_s}\mathbf{n},\ \tilde{r}\geq 1.
    \end{equation*}
    
    \subsection{Effect of the mean surface charge concentration}
    \label{appendix:exact_sol_superimposition_of_effects:q0}
    According to Gauss theorem of electrostatics, the flux of the electric field $\bm{\Psi}^g_s$ through a closed line is equal to the total charge in the area enclosed by the line. Thus,
    
    \begin{equation*}
    	\int_{\partial\mathcal{W}_s} \epsilon_s\bm{\Psi}_s^6\cdot\mathbf{n} = \int_{\mathcal{W}_s} \tilde{\nabla}_{\tilde{\mathcal{D}}}\cdot\bm{\Psi}_s^6 = \int_{\partial\tilde{\mathcal{D}}} 1 = |\partial\tilde{\mathcal{D}}|,
    \end{equation*}
    
    \noindent where $|\partial\tilde{\mathcal{D}}| = 2\pi$ denotes the perimeter of the section $\tilde{\mathcal{D}}$. Moreover, such field is radial and uniform, since it is generated by a constant charge distribution on a circumference, i.e. $\bm{\Psi}_s^6 = |\bm{\Psi}_s^6|(\tilde{r}) \mathbf{n}$, $\tilde{r}\geq 1$. Thus, the exact solution to problem~\eqref{eq:reduction1d:superimposition:q0} in $\Omega_s$ is:
    
    \begin{equation*}
    	\bm{\Psi}_s^6 = \dfrac{1}{\tilde{r}\epsilon_s}\mathbf{n}, \ \tilde{r}\geq 1.
    \end{equation*}
    
    \noindent On the other hand, the total charge inside any circle $\mathcal{W}_g\subset\tilde{\mathcal{D}}$ is 0, producing a null electric field $\bm{\Psi}_g^6 = 0$, in $\tilde{\mathcal{D}}$.
    
    \section{Properties of the coefficients F and G}
    \label{appendix:FG}
    In this section we compute the values of the coefficients $F_p^\pm$ and $G_{p,i}^\pm$, for $p=1,\,2,\,3$, $i=1,\,2$, defined in equations~\eqref{eq:reduction1d:def:F} and~\eqref{eq:reduction1d:def:G}, given the transverse electric field $\mathbf{E}_g$ computed in Section~\ref{section:reduction1d:superimposition_of_effects}, and discuss some of their properties.
    
    Substitute the expression of the electric field in the gas obtained by superimposition of effects in equation~\eqref{eq:reduction1d:superimposition:fields:g} into equations~\eqref{eq:reduction1d:def:F} and~\eqref{eq:reduction1d:def:G}, and recall that $\bm{\Psi}_g^1 = \epsilon_s \mathbf{x}$, $\bm{\Psi}_g^2 = \epsilon_s \mathbf{y}$, while $\bm{\Psi}_g^5$ is radial and has magnitude $|\bm{\Psi}_g^5|(R_g) = \dfrac{1}{2}$ on the boundary of the cross-sections $\mathcal{D}$ of $\Omega_g$. Then the coefficients can be rewritten as follows:
    
    \begin{alignat}{4}
    	F_p^+ &= \int_{\partial\mathcal{D}} \max\{0,\omega_p(\mathbf{v}\cdot\mathbf{n} + \rho)\}, \quad
    	&F_p^- &= \int_{\partial\mathcal{D}} \min\{0,\omega_p(\mathbf{v}\cdot\mathbf{n} + \rho)\},
    	\label{eq:reduction1d:appendix:F:rho}\\
    	G_{p,i}^+ &= \int_{\partial\mathcal{D}}  \max\left\{ 0, \omega_p(\mathbf{v} \cdot \mathbf{n} + \rho) \right\} \phi_i, \quad
    	&G_{p,i}^- &= \int_{\partial\mathcal{D}} \min\left\{ 0, \omega_p (\mathbf{v} \cdot \mathbf{n} + \rho) \right\} \phi_i,
    	\label{eq:reduction1d:appendix:G:rho}
    \end{alignat}
    
    \noindent where we have defined $\mathbf{v} = a_1\epsilon_s\mathbf{x} + a_2\epsilon_s\mathbf{y}$ and $\rho = \dfrac{a_5}{2}$, while $\mathbf{n}$ denotes the outgoing unit normal vector to $\partial\mathcal{D}$.\\
    We denote by $\partial\mathcal{D}^+$ the subregion of $\partial\mathcal{D}$ where $\mathbf{v}\cdot\mathbf{n}+\rho > 0$ and by $\partial\mathcal{D}^-$ the subregion of $\partial\mathcal{D}$ where $\mathbf{v}\cdot\mathbf{n}+\rho < 0$. In particular, when $\omega_p=1$, i.e. $p=3$, the function $\max\{0,\omega_p(\mathbf{v}\cdot\mathbf{n} + \rho)\}$ is non-zero only on $\partial\mathcal{D}^+$ and $\min\{0,\omega_p(\mathbf{v}\cdot\mathbf{n} + \rho)\}$ is non-zero only on $\partial\mathcal{D}^-$, while for $\omega_p=-1$, i.e. $p=1,\,2$, the two regions are switched. Figure~\ref{figure:appendix:rho_0} represents the region where $\mathbf{v}\cdot\mathbf{n} > 0$, corresponding to all the points where the normal vector $\mathbf{n}$ forms an acute angle with the vector $\mathbf{v}$. The line separating the region where $\mathbf{v}\cdot\mathbf{n} + \rho$ is positive from the region where it is negative is shifted by a factor $\rho$ along the direction of $\mathbf{v}$ when $\rho\neq 0$, in particular enlarging the positivity region when $\rho > 0$ and reducing it when $\rho < 0$ (see Figures~\ref{figure:appendix:rho_positive} and~\ref{figure:appendix:rho_negative}), until it reaches the limit cases where the positivity region is either empty or coincident with the whole $\partial\mathcal{D}$.
    
    \begin{figure}
    	\centering
    	\subfloat[$\rho = 0$\label{figure:appendix:rho_0}]
    	{\centering
    		\begin{tikzpicture}[scale=0.5]
    			\begin{axis}
    				[axis equal=true,
    				axis lines=middle,
    				xlabel={$x$},
    				ylabel={$y$},
    				xmin=-1.2, xmax=1.2,
    				ymin=-1.2, ymax=1.2,
    				xtick=\empty, ytick=\empty
    				]
    				\addplot[domain=0:360, samples=100, thick] ({cos(x)}, {sin(x)})
    				node[pos = 0.45, anchor=east] {$\partial\mathcal{D}$};
    				\addplot [domain=-45:135, samples=100, thick, color=gray!40, fill=gray!40, opacity=0.5] ({cos(x)}, {sin(x)});
    				\addplot [domain=135:315, samples=100, thick, color=darkBlue!40, fill=darkBlue!40, opacity=0.5] ({cos(x)}, {sin(x)});
    				\addplot[domain=-0.71:0.71, samples=100, thick, dashed] {-x};
    				\addplot[domain=-0.7:-0.9, samples=10, thick, ->] {-x}
    				node[pos=0.5, anchor=north east] {$\mathbf{n}$};
    				\addplot[domain=0:0.9, samples=100, thick, ->] {x}
    				node[pos=0.9, anchor=south east] {$\mathbf{v}$};
    				\draw(axis cs:0,0) 
    				-- (axis cs:0.1,0.1) 
    				-- (axis cs:0.2,0) 
    				-- (axis cs:0.1,-0.1);
    			\end{axis}
    	\end{tikzpicture}}\quad
    	\subfloat[$\rho > 0$\label{figure:appendix:rho_positive}]
    	{\centering
    		\begin{tikzpicture}[scale=0.5]
    			\begin{axis}
    				[axis equal=true,
    				axis lines=middle,
    				xlabel={$x$},
    				ylabel={$y$},
    				xmin=-1.2, xmax=1.2,
    				ymin=-1.2, ymax=1.2,
    				xtick=\empty, ytick=\empty
    				]
    				\addplot[domain=0:360, samples=100, thick] ({cos(x)}, {sin(x)})
    				node[pos = 0.45, anchor=east] {$\partial\mathcal{D}$};
    				\addplot [domain=-65:155, samples=100, thick, color=gray!40, fill=gray!40, opacity=0.5] ({cos(x)}, {sin(x)});
    				\addplot [domain=155:295, samples=100, thick, color=darkBlue!40, fill=darkBlue!40, opacity=0.5] ({cos(x)}, {sin(x)});
    				\addplot[domain=-.91:0.41, samples=100, thick, dashed] {-x-0.5};
    				\addplot[domain=-0.7:-0.9, samples=10, thick, ->] {-x}
    				node[pos=0.5, anchor=north east] {$\mathbf{n}$};
    				\addplot[domain=0:0.9, samples=100, thick, ->] {x}
    				node[pos=0.9, anchor=south east] {$\mathbf{v}$};
    				\addplot[domain=-0.25:0, samples=10, thick, <->] {x};
    			\end{axis}
    	\end{tikzpicture}}\quad
    	\subfloat[$\rho < 0$\label{figure:appendix:rho_negative}]
    	{\centering
    		\begin{tikzpicture}[scale=0.5]
    			\begin{axis}
    				[axis equal=true,
    				axis lines=middle,
    				xlabel={$x$},
    				ylabel={$y$},
    				xmin=-1.2, xmax=1.2,
    				ymin=-1.2, ymax=1.2,
    				xtick=\empty, ytick=\empty
    				]
    				\addplot[domain=0:360, samples=100, thick] ({cos(x)}, {sin(x)})
    				node[pos = 0.45, anchor=east] {$\partial\mathcal{D}$};
    				\addplot [domain=-25:115, samples=100, thick, color=gray!40, fill=gray!40, opacity=0.5] ({cos(x)}, {sin(x)});
    				\addplot [domain=115:335, samples=100, thick, color=darkBlue!40, fill=darkBlue!40, opacity=0.5] ({cos(x)}, {sin(x)});
    				\addplot[domain=-.41:0.91, samples=100, thick, dashed] {-x+0.49};
    				\addplot[domain=-0.7:-0.9, samples=10, thick, ->] {-x}
    				node[pos=0.5, anchor=north east] {$\mathbf{n}$};
    				\addplot[domain=0:0.9, samples=100, thick, ->] {x}
    				node[pos=0.9, anchor=south east] {$\mathbf{v}$};
    				\addplot[domain=0:0.25, samples=10, thick, <->] {x};
    			\end{axis}
    	\end{tikzpicture}}
    	\caption{Boundary of a section $\partial\mathcal{D}$, with corresponding regions where $\mathbf{v}\cdot\mathbf{n} + \rho$ is positive (colored in gray) and negative (colored in blue). The two regions are separated by the dashed line, orthogonal to the direction of $\mathbf{v}$.}
    	\label{figure:appendix:rho}
    \end{figure}
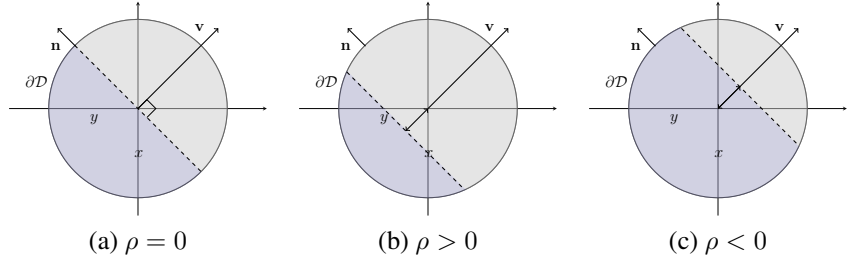
    
    \noindent In particular, if we denote by $\hat{\theta}$ the angle between $\mathbf{v}$ and the positive direction of the $x$ axis (i.e. $\tan\hat{\theta} = \frac{a_2}{a_1}$), and by $\theta$ the angular coordinate around $\partial\mathcal{D}$, then $\mathbf{v}\cdot\mathbf{n}>0$ only if $\cos(\theta-\hat{\theta}) > -\dfrac{\rho}{v}$, where $v=\epsilon_s\sqrt{a_1^2+a_2^2}$. As a consequence,
    
    \begin{equation*}
    	\partial\mathcal{D}^+ =
    	\begin{cases}
    		[0,2\pi], & \text{if\ } \rho \geq v,\\
    		\emptyset, & \text{if\ } \rho \leq -v, \\
    		\left( \hat{\theta} - \tilde{\theta},\ \hat{\theta} + \tilde{\theta} \right), & \text{if\ } v > \rho > -v,
    	\end{cases}
    	\quad
    	\partial\mathcal{D}^- =
    	\begin{cases}
    		\emptyset, & \text{if\ } \rho \geq v,\\
    		[0,2\pi], & \text{if\ } \rho \leq -v, \\
    		\left( \hat{\theta} - \tilde{\theta},\ \hat{\theta} + \tilde{\theta} \right), & \text{if\ } v > \rho > -v,
    	\end{cases}
    	\label{eq:reduction1d:appendix:inflow_outflow_boundary}
    \end{equation*}
    
    \noindent where $\tilde{\theta} = \arccos\left( -\dfrac{\rho}{v} \right)$.
    
    \subsection{Computation of the coefficients $F_p^\pm$}
    \label{appendix:F}
    We start by rewriting the coefficients $F_p^\pm,\, p=1,\,2,\,3,$ as follows:
    
    \begin{equation}
    	-F_1^- = -F_2^- = F_3^+ = \int_{\partial\mathcal{D}^+} (\mathbf{v}\cdot\mathbf{n} + \rho), \quad
    	-F_1^+ = -F_2^+ = F_3^- = \int_{\partial\mathcal{D}^-} (\mathbf{v}\cdot\mathbf{n} + \rho).
    	\label{eq:reduction1d:F:appendix:1}
    \end{equation}
    
    \noindent Since $\rho$ is constant on $\partial\mathcal{D}^+$ and $\partial\mathcal{D}^-$, we only need to compute the following integral:
    
    \begin{equation*}
    	I(a,b) = \int_a^b \mathbf{v}\cdot\mathbf{n}\ d\theta,
    \end{equation*}
    
    \noindent for $a,b\in[0,2\pi]$, and evaluate it for $(a,b) = (0,2\pi)$ and $(a,b) = (\hat{\theta} -\tilde{\theta}, \hat{\theta} + \tilde{\theta}) $. Indeed, we can rewrite equation~\eqref{eq:reduction1d:F:appendix:1} as follows:
    
    \begin{equation}
    	\begin{aligned}
    		-F_1^- = -F_2^- = F_3^+ &=
    		\begin{cases}
    			R_g \left( v I(0,2\pi) + (b-a) \rho \right), &\text{if}\ \rho\geq v, \\
    			R_g \left( v I(\hat{\theta} - \tilde{\theta}, \hat{\theta} + \tilde{\theta}) + (b-a) \rho \right), &\text{if}\ v>\rho>-v, \\
    			0, &\text{if}\ \rho\leq -v,
    		\end{cases}\\
    		-F_1^+ = -F_2^+ = F_3^- & =
    		\begin{cases}
    			0, &\text{if}\ \rho\geq v, \\
    			R_g \left( v I(\hat{\theta} + \tilde{\theta}, \hat{\theta} + 2\pi -\tilde{\theta}) + 2\tilde{\theta}\rho \right), &\text{if}\ v>\rho>-v, \\
    			R_g \left( v I(0,2\pi) + 2\pi\rho \right), &\text{if}\ \rho\leq -v.
    		\end{cases}
    	\end{aligned}
    	\label{eq:reduction1d:F:appendix:2}
    \end{equation}
    
    \noindent Since $v$ is constant on $\partial{D}$ and $\cos(\theta - \hat{\theta}) = \cos\theta \cos\hat{\theta} - \sin\theta\sin\hat{\theta}$, we obtain:
    
    \begin{multline*}
    	I(a,b) = \int_a^b v \cos(\theta - \hat{\theta}) d\theta =
    	\cos\hat{\theta}\int_a^b \cos\theta d\theta + \sin\hat{\theta}\int_a^b \sin\theta d\theta =\\
    	= \cos\hat{\theta}(\sin b - \sin a) + \sin\hat{\theta}(\cos a - \cos b).
    \end{multline*}
    
    \noindent Then, exploiting the properties of $\sin$ an $\cos$ of sum and difference of angles, we can evaluate the integral on the intervals $(a,b)$ of our interest:
    
    \begin{equation}
    	I(0,2\pi) = 0, \quad
    	I(\hat{\theta} - \tilde{\theta}, \hat{\theta} + \tilde{\theta}) = 2\sin\tilde{\theta}, \quad
    	I(\hat{\theta} + \tilde{\theta}, \hat{\theta} + 2\pi - \tilde{\theta}) = -2\sin\tilde{\theta}.
    	\label{eq:reduction1d:F:appendix:integral}
    \end{equation}
    
    \noindent Finally substituting~\eqref{eq:reduction1d:F:appendix:integral} into~\eqref{eq:reduction1d:F:appendix:2}, we obtain:
    
    \begin{equation}
    	\begin{aligned}
    		-F_1^- = -F_2^- = F_3^+ &=
    		\begin{cases}
    			2 R_g \pi \rho, &\text{if}\ \rho\geq v, \\
    			2 R_g \left( v\sin\tilde{\theta} + \rho\tilde{\theta} \right), &\text{if}\ v>\rho>-v, \\
    			0, &\text{if}\ \rho\leq -v,
    		\end{cases}\\
    		-F_1^+ = -F_2^+ = F_3^- & =
    		\begin{cases}
    			0, &\text{if}\ \rho\geq v, \\
    			2 R_g \left( \rho (\pi - \tilde{\theta}) - v \sin(\pi - \tilde{\theta}) \right), &\text{if}\ v>\rho>-v, \\
    			2 R_g \pi \rho, &\text{if}\ \rho\leq -v.
    		\end{cases}
    	\end{aligned}
    	\label{eq:reduction1d:F:appendix:expression}
    \end{equation}
    
    \begin{remark}
    	\label{remark:appendix:F:sign}
    	We notice from equation~\eqref{eq:reduction1d:F:appendix:1} that $F_p+$ are integrals of positive functions, while $F_p^-$ are integrals of negative functions. As a consequence, $F_p^+ \geq 0$, and $F_p^-\leq 0$, $\forall p=1,2,3$.
    \end{remark}
    
    \subsection{Computation of the coefficients $G^\pm_{p,i}$}
    \label{appendix:G}
    In order to compute the coefficients $G_{p,i}^\pm$, for $p=1,\,2,\,3$ and $i=1,\,2$, as in equation~\eqref{eq:reduction1d:appendix:G:rho}, we need to evaluate on the regions $\partial\mathcal{D}^\pm$, defined in~\eqref{eq:reduction1d:appendix:inflow_outflow_boundary}, the following integrals:
    
    \begin{equation*}
    	I_1(a,b) = \int_a^b \left(\mathbf{v}\cdot\mathbf{n} + \rho\right) \cos\theta d\theta,
    	\quad
    	I_2(a,b) = \int_a^b \left(\mathbf{v}\cdot\mathbf{n} + \rho\right) \sin\theta d\theta,
    	\quad
    	a,\,b \in \left[ 0,\,2\pi \right].
    \end{equation*}
    
    \noindent Indeed,
    
    \begin{equation}
    	\begin{aligned}
    		\int_{\partial\mathcal{D}^+} (\mathbf{v}\cdot\mathbf{n} + \rho)\phi_i &=
    		\begin{cases}
    			-\epsilon_s I_i(0,2\pi), & \text{if\ } \rho\geq v, \\
    			0, & \text{if\ } \rho \leq -v, \\
    			-\epsilon_s I_i(\hat{\theta}-\tilde{\theta}, \hat{\theta}+\tilde{\theta}), & \text{if\ } v > \rho > -v,
    		\end{cases}\\
    		\int_{\partial\mathcal{D}^-} (\mathbf{v}\cdot\mathbf{n} + \rho)\phi_i &=
    		\begin{cases}
    			0, & \text{if\ } \rho\geq v, \\
    			-\epsilon_s I_i(0,2\pi), & \text{if\ } \rho \leq -v, \\
    			-\epsilon_s I_i(\hat{\theta}+\tilde{\theta}, \hat{\theta} + 2\pi - \tilde{\theta}), & \text{if\ } v > \rho > -v.
    		\end{cases}
    	\end{aligned}
    	\label{eq:reduction1d:G:v_rho_+-}
    \end{equation}
    
    \noindent Since $v$, $\hat{\theta}$ and $\rho$ are constant on $\partial\mathcal{D}$, and $\cos(\theta-\hat{\theta}) = \cos\theta \cos\hat{\theta} - \sin\theta \sin\hat{\theta}$, we obtain:
    
    \begin{align*}
    	I_1(a,b) &= \int_a^b v\cos(\theta-\hat{\theta})\cos\theta d\theta + \int_a^b \rho\cos\theta d\theta =\\
    	&= v \cos\hat{\theta} \int_a^b \cos^2\theta d\theta + v\sin\hat{\theta} \int_a^b \sin\theta\cos\theta d\theta + \rho\int_a^b \cos\theta d\theta = \\
    	&= \dfrac{v\cos\hat{\theta}}{2} (\cos b \sin b - \cos a \sin a + b - a)
    	+ \dfrac{v\sin\hat{\theta}}{2} (\sin^2 b - \sin^2 a)
    	+ \rho (\sin b - \sin a),
    \end{align*}
    
    \begin{align*}
    	I_2(a,b) &= \int_a^b v\cos(\theta-\hat{\theta})\sin\theta d\theta + \int_a^b \rho\sin\theta d\theta =\\
    	&= v \cos\hat{\theta} \int_a^b \cos\theta\sin\theta d\theta + v\sin\hat{\theta} \int_a^b \sin^2\theta d\theta + \rho\int_a^b \sin\theta d\theta = \\
    	&= \dfrac{v\cos\hat{\theta}}{2} (\sin^2 b - \sin^2 a)
    	+ \dfrac{v\sin\hat{\theta}}{2} (\cos a \sin a - \cos b \sin b + b - a)
    	- \rho (\cos b - \cos a).
    \end{align*}
    
    \noindent Then, we can evaluate these two integrals on the intervals of our interest:
    
    \begin{equation}
    	\begin{alignedat}{3}
    		I_1(0,2\pi) &= \pi v \cos\hat{\theta}, \quad 
    		I_1(\hat{\theta}-\tilde{\theta}, \hat{\theta}+\tilde{\theta}) & =
    		\cos\hat{\theta} \left( v\tilde{\theta} + \rho \sin\tilde{\theta} \right), \\
    		I_2(0,2\pi) & = \pi v \sin\hat{\theta}, \quad
    		I_2(\hat{\theta}-\tilde{\theta}, \hat{\theta}+\tilde{\theta}) & =
    		\sin\hat{\theta} \left( v\tilde{\theta} + \rho \sin\tilde{\theta} \right).
    		% I_1(0,2\pi) &= \pi v \cos\hat{\theta}
    		% \begin{dcases}
    			%     > 0, & \forall v, \forall \hat{\theta} \in \left( -\dfrac{\pi}{2}, \dfrac{\pi}{2} \right), \\
    			%     < 0, & \forall v, \forall \hat{\theta} \in \left( \dfrac{\pi}{2}, \dfrac{3}{2}\pi \right),
    			% \end{dcases} \\
    		% I_2(0,2\pi) & = \pi v \sin\hat{\theta} 
    		% \begin{dcases}
    			%     > 0, & \forall v,\ \forall \hat{\theta} \in \left( 0,\pi \right), \\
    			%     < 0, & \forall v,\ \forall \hat{\theta} \in \left( \pi, 2\pi \right),
    			% \end{dcases}     \\
    		% I_1(\hat{\theta}-\tilde{\theta}, \hat{\theta}+\tilde{\theta}) & =
    		% \cos\hat{\theta} \left( v\tilde{\theta} + \rho \sin\tilde{\theta} \right)
    		% \begin{dcases}
    			%     > 0, & \forall v,\ \forall \rho,\ \forall \hat{\theta}\in\left( -\dfrac{\pi}{2}, \dfrac{\pi}{2} \right),\\
    			%     < 0, & \forall v,\ \forall \rho,\ \forall \hat{\theta}\in\left( -\dfrac{\pi}{2}, \dfrac{\pi}{2} \right),
    			% \end{dcases}     \\
    		% I_2(\hat{\theta}-\tilde{\theta}, \hat{\theta}+\tilde{\theta}) & =
    		% \sin\hat{\theta} \left( v\tilde{\theta} + \rho \sin\tilde{\theta} \right)
    		% \begin{dcases}
    			%     > 0, & \forall v,\ \forall \rho,\ \forall \hat{\theta}\in\left( 0,\pi \right),\\
    			%     < 0, & \forall v,\ \forall \rho,\ \forall \hat{\theta}\in\left( \pi,2\pi \right).
    			% \end{dcases}
    	\end{alignedat}
    	\label{eq:reduction1d:appendix:G:integrals}
    \end{equation}
    
    \noindent Since $\hat{\theta}$ was defined as the angle between $\mathbf{v}$ and the positive direction of the $x$ axis, we can trace it back to the values of $a_1$ and $a_2$, and, by properly replacing equations~\eqref{eq:reduction1d:appendix:G:integrals} into~\eqref{eq:reduction1d:G:v_rho_+-}, we obtain an explicit expression of the coefficients $G_{p,i}^+$ and $G_{p,i}^-$, $p=1,2,3,\ i=1,2$:
    
    \begin{equation}
    	\label{eq:reduction1d:G:expression}
    	% \begin{aligned}
    		-G_{1,i}^- = -G_{2,i}- = G_{3,i}^+ =
    		\begin{cases}
    			-\epsilon_s^2 R_g \pi a_i, & \text{if}\ \rho \geq v, \\
    			0, & \text{if}\ \rho \leq -v, \\
    			- R_g \dfrac{a_i}{\sqrt{a_1^2+a_2^2}} \left(\sqrt{a_1^2+a_2^2}\tilde{\theta}+\rho\sin\tilde{\theta}\right), & \text{if}\ v>\rho>-v,
    		\end{cases}
    		% \\
    	\end{equation}
    	\begin{equation*}
    		-G_{1,i}^+ = -G_{2,i}^+ = G_{3,i}^- =
    		\begin{cases}
    			0, & \text{if}\ \rho \geq v, \\
    			-\epsilon_s^2 R_g \pi a_i, & \text{if}\ \rho \leq -v, \\
    			- R_g \dfrac{a_i}{\sqrt{a_1^2+a_2^2}} \left(\sqrt{a_1^2+a_2^2}(\pi - \tilde{\theta}) -\rho\sin(\pi-\tilde{\theta})\right), & \text{if}\ v>\rho>-v.
    		\end{cases}
    		% \end{aligned}
    \end{equation*}
    
    \begin{remark}
    	\label{remark:appendix:G:sign}
    	Notice that%, since we defined $\hat{\theta}$ such that $\tan\hat{\theta} = \dfrac{a_2}{a_1}$,
        % \begin{itemize}
        %     \item $\hat{\theta}\in\left(0,\frac{\pi}{2}\right)$ if $a_1>0$ and $a_2 >0$,
        %     \item $\hat{\theta}\in\left(\frac{\pi}{2},\pi\right)$ if $a_1<0$ and $a_2 >0$,
        %     \item $\hat{\theta}\in\left(\pi,\frac{3}{2}\pi\right)$ if $a_1<0$ and $a_2<0$,
        %     \item $\hat{\theta}\in\left(\frac{3}{2}\pi, 2\pi\right)$ if $a_1>0$ and $a_2<0$.
        % \end{itemize}
        % \noindent Moreover, $v = \epsilon_s\sqrt{a_1^2+a_2^2} > 0$ for all $a_1$ and $a_2$. Thus, $I_i\geq 0 \iff a_i\geq 0$.
        % As a consequence,
        $G_{p,i}^+\leq 0$ and $G_{p,i}^-\leq 0$ if $a_i\geq 0$ for $p=3$, while $G_{p,i}^+\geq 0$ and $G_{p,i}^-\geq 0$ if $a_i\geq 0$ for $p=1,2$.
    \end{remark}

    % \subsection{Proof of Lemma~\ref{theorem:reaction:positive}}
    % \label{appendix:F:proof:positive}
    % For all $p=1,\,2,\,3$, we evaluate $f_p$ at $\mathbf{\tilde{c}}=(\tilde{c}_1,\,\tilde{c}_2,\,\tilde{c}_3)$, with $\tilde{c}_p = 0$ and $\tilde{c}_q \geq 0, \ \forall q \neq p$. Then,
    
    % \begin{equation*}
    	%     f_p(\mathbf{\tilde{c}}) = -\dfrac{|\partial\mathcal{D}|}{|\mathcal{D}|} K_p 0 + 0 \mathcal{N}_p^+(\mathbf{\tilde{c}}) + \mathcal{N}_p^-(\mathbf{\tilde{c}}) + \bar{C}_p(\tilde{c}_1) = \mathcal{N}_p^-(\mathbf{\tilde{c}}) + \bar{C}_p(\tilde{c}_1),
    	% \end{equation*}
    
    % \noindent with, by definition~\eqref{eq:reduction1d:chemistry:expression},
    
    % \begin{equation*}
    	%     \bar{C}_p(\mathbf{\tilde{c}}) =
    	%     \begin{cases}
    		%         (\alpha - \eta)0, & p=1,\\
    		%         \alpha \tilde{c}_1 \geq 0, & p=2, \\
    		%         \eta \tilde{c}_1 \geq 0, & p=3.
    		%     \end{cases}
    	% \end{equation*}
    
    % \noindent Moreover, $\mathcal{N}_p^-$ is defined in equation~\eqref{eq:reduction1d:Np:def} as a negative quantity, multiplied by $F_p^-$, which is non-positive for any value of $\mathbf{\tilde{c}}$, by Theorem~\ref{theorem:FG:sign}. Thus, $\mathcal{N}_p^-(\mathbf{\tilde{c}}) \geq 0$, and, as a consequence, $f_p(\mathbf{\tilde{c}}) \geq 0$.
    
    \subsection{Proof that the destruciton term~\eqref{eq:reduction1d:destruction} satisfies Property~\ref{property:D:limit}}
    \label{appendix:production-destruction}
    The destruction term~\eqref{eq:reduction1d:destruction} is given by 
    
    \begin{equation*}
    	\begin{split}
    		D(a_i) &= \dfrac{\bar{\sigma}_\Gamma}{|\partial\mathcal{D}|}\epsilon_s^2\pi a_i - 
    		\sum_{p=1}^3 e\dfrac{\bar{\mu}_p}{|\mathcal{D}|} \left(\bar{c}_p G_{p,i}^+ + \bar{c}_p^b G_{p,i}^- \right) =\\
    		&=
    		\begin{cases}
    			\dfrac{\bar{\sigma}_\Gamma}{|\partial\mathcal{D}|}\epsilon_s^2\pi a_i +
    			\dfrac{e}{|\mathcal{D}|}\epsilon_s^2 R_g \pi a_i \bar{\mu}_e\bar{c}_3 
    			, & \text{if}\ \rho\geq v, \\
    			\dfrac{\bar{\sigma}_\Gamma}{|\partial\mathcal{D}|}\epsilon_s^2\pi a_i + 
    			\dfrac{e}{|\mathcal{D}|}\epsilon_s^2 R_g \pi a_i \bar{\mu}_e\bar{c}_3^b 
    			, & \text{if}\ \rho \leq -v, \\
    			\dfrac{\bar{\sigma}_\Gamma}{|\partial\mathcal{D}|}\epsilon_s^2\pi a_i
    			+ \dfrac{e}{|\mathcal{D}|} R_g a_i 
    			\left( \bar{\mu}_3\bar{c}_3 + \sum_{p=1}^2 \bar{\mu}_p \bar{c}_p^b \right) K^+(\tilde{\theta}),
    			&\text{if}\ v>\rho>-v,
    		\end{cases}
    	\end{split}
    \end{equation*}
    
    \noindent where $ K^+(\tilde{\theta})  = \tilde{\theta} + \dfrac{\rho}{\sqrt{a_1^2+a_2^2}}\sin\tilde{\theta}$ is finite for $a_i\to 0$, $i=1,2$. Thus, $\lim_{a_i\to 0} D(a_i) = 0$ for $i=1,2$.

    \subsection{Proof of Property~\ref{property:reduction1d:F:bounded}}
    \label{appendix:F:proof:bounded}
    
    \begin{itemize}
    	\item \textit{Case $\rho \geq v$:}
    	\begin{equation*}
    		F_p^+ =
    		\begin{cases}
    			0, & \text{if\ } p=1,\,2,\\
    			2R_g\pi\rho \leq
    			\hat{F}_3^+ \sum_{p=1}^3 |\bar{c}_p|, & \text{if\ } p=3,
    		\end{cases}
    	\end{equation*}
    	\noindent by triangular inequality, with $\hat{F}_3^+ = 2R_g\pi$;
    	
    	\item \textit{Case $p\leq -v$:}
    	\begin{equation*}
    		F_p^+ =
    		\begin{cases}
    			-2R_g\pi\rho =
    			R_g^2 \pi \left( \bar{c}_1 + \bar{c}_2 - \bar{c}_3 \right)
    			\leq \hat{F}_p^+ \sum_{p=1}^3|\bar{c}_p|, & \text{if\ } p=1,\,2,\\
    			0, & \text{if\ } p=3,
    		\end{cases}
    	\end{equation*}
    	\noindent by triangular inequality, with $\hat{F}_p^+ = R_g^2 \pi$;
    	
    	\item \textit{Case $v>\rho>-v$:}\\
    	Since $\tilde{\theta}\in[0,\,\pi]$, $\sin(\gamma) \in [-1,\,1],\, \forall \gamma\in[0,\,2\pi$, and $\rho\in(-v,\,v)$,
    	\begin{equation*}
    		F_p^+ =
    		\begin{cases}
    			2R_g\left( v\sin(\pi-\tilde{\theta}) +\rho(\pi-\tilde{\theta})\right)
    			< 2R_g v (\pi +1), & \text{if\ } p=1,\,2,\\
    			2R_g\left( \rho\tilde{\theta} + v\sin\tilde{\theta}\right) 
    			< 2R_g v (\pi +1), &\text{if\ } p=3.
    		\end{cases}
    	\end{equation*}
    	
    	\noindent As a consequence, $F_p^+ \leq \hat{F}_p^+ \sum_{p=1}^3 |\bar{c}_p| + \tilde{F}_p^+$, with
    	
    	\begin{align*}
    		\hat{F}_p^+ &=
    		\begin{cases}
    			2 R_g \pi, & \text{if}\ p=3 \text{\ and\ } \rho \geq v,\\
    			R_g^2 \pi, & \text{if}\ p=1,\,2 \text{\ and\ } \rho \leq -v,\\
    			0, & \text{otherwise},
    		\end{cases}\\
    		\tilde{F}_p^+ &=
    		\begin{cases}
    			2R_gv(\pi+1), & \text{if\ } v>\rho>-v, \ p=1,\,2,\,3,\\
    			0, & \text{otherwise}.
    		\end{cases}
    	\end{align*}
    	
    \end{itemize}

    \section{Characteristic time scale of the right--hand side of the dipole moment equation}
    \label{appendix:reduction1d:characteristic_time}
    The typical time scale of the forcing term on right--hand side of~\eqref{eq:reduction1d:dipole:a} can be estimated as

        \begin{multline}
            \sum_{p=1}^3 \dfrac{e}{\epsilon_0\epsilon_s^2 \pi R_g} \frac{\bar{\mu}_p}{|\mathcal{D}|}\left( \bar{c}_p G_{p,i}^+ 
            + \bar{c}_p^b G_{p,i}^- \right) \approx\\
            \approx \sum_{p=1}^3 \frac{1.6 \cdot 10^{-19} C}{8.8 \cdot 10^{-12} \frac{C}{V \ m} \epsilon_s^2 \pi R_g} \frac{10^{-2} \frac{m^2}{V \ s}}{\pi R_g^2} 10^{20} \frac{1}{m} \left( G_{p,i}^+ 
            + G_{p,i}^-\right) \approx
            10^{10} \dfrac{1}{s} \sum_{p=1}^3 \dfrac{G_{p,i}^+ 
            + G_{p,i}^-}{4\pi R_g \epsilon_s^2},
            \label{eq:reduction1d:Gp:order_of_magnitude}
        \end{multline}

        \noindent where, from the expression derived in Appendix~\ref{appendix:G}, we have 
        
        \begin{equation*}
        	\sum_{p=1}^3 \left(G_{p,i}^+ + G_{p,i}^- \right) = \epsilon_s^2 R_g\pi a_i.
        \end{equation*}
        
        \noindent Substituting this approximation into~\eqref{eq:reduction1d:Gp:order_of_magnitude}, we obtain 
            
        \begin{equation*}
           	\sum_{p=1}^3 \dfrac{e}{\epsilon_0\epsilon_s^2 \pi R_g} \dfrac{\bar{\mu}_p}{|\partial \mathcal{D}|}\left( \bar{c}_p G_{p,i}^+ 
        	+ \bar{c}_p^b G_{p,i}^- \right) \approx
		    10^{10}\dfrac{1}{s} a_i.
        \end{equation*}
            
        \noindent Similarly,

        \begin{equation*}
            \dfrac{\bar{\sigma}_\Gamma}{|\partial\mathcal{D}|}\dfrac{a_i}{\epsilon_0 R_g}\approx
            \dfrac{1.6 \cdot 10^{-19} C \cdot 10^{20} \frac{1}{m} \cdot 10^{-16} \frac{m^2}{V \ s}}{2\pi R_g}\dfrac{a_i}{8.8\cdot 10^{-12}\frac{C}{V \ m} R_g}\approx
            \dfrac{1}{R_g} 10^4 \dfrac{m}{s} a_i,
        \end{equation*}

        \noindent with $R_g$ small. In particular, with a radius of the gas domain $R_g \approx 10^{-6}$ we have $\dfrac{\bar{\sigma}_\Gamma}{|\partial\mathcal{D}|}\dfrac{a_i}{\epsilon_0 R_g}\approx a_i\cdot 10^{10}\dfrac{1}{s}$.
    	
\end{document}